\documentclass[journal]{IEEEtran}

\ifCLASSINFOpdf
\else
\fi

\usepackage{graphicx}
\usepackage{amsmath}
\usepackage{amsfonts}
\usepackage{bm}
\usepackage{enumitem}
\usepackage{amsthm}
\usepackage{multirow}
\usepackage{booktabs}
\usepackage{xcolor}

\usepackage{tikz}
\usetikzlibrary{spy}
\usepackage{makecell}
\usepackage{thmtools}
\usepackage{bbm}

\usepackage[normalem]{ulem}

\def\qut#1{\left(#1\right)}

\def\bmy{\bm{y}}
\def\bmyh{\bm{\hat{y}}}
\def\bmx{\bm{x}}
\def\bmz{\bm{z}}
\def\bme{\bm{e}}
\def\bmp{\bm{p}}
\def\bmph{\bm{\hat{p}}}
\def\bmn{\bm{n}}
\def\bmnh{\bm{\hat{n}}}

\def\bbE{\mathbb{E}}

\def\calY{\mathcal{Y}}
\def\calX{\mathcal{X}}

\def\qut#1{{\left( #1 \right)}}
\def\qutc#1{{\left\{ #1 \right\}}} %
\def\qutn#1{{\left\| #1 \right\|}} %
\def\qutl#1{\mathopen{}\left(#1\right)\mathclose{}}
\def\qutan#1{\langle#1\rangle}

\newtheorem{theorem}{Theorem}
\newtheorem{lemma}[theorem]{Lemma}
\newtheorem{proposition}[theorem]{Proposition}

\newtheorem{remark}{Remark}
\newtheorem*{example}{Example}

\begin{document}
\title{Unsupervised Image Restoration Using Partially Linear Denoisers}

\author{Rihuan Ke  and 
        Carola-Bibiane Sch\"onlieb
\thanks{R. Ke and C.-B. Sch\"onlieb are with DAMTP,  University  of  Cambridge, Cambridge CB3 0WA, UK.}%
}

\maketitle
\begin{abstract}
Deep neural network based methods are the state of the art in various image restoration problems. Standard supervised learning frameworks require a set of noisy measurement and clean image pairs for which a distance between the output of the restoration model and the ground truth, clean images is minimized. The ground truth images, however, are often unavailable or very expensive to acquire in real-world applications. We circumvent this problem by proposing a class of structured denoisers that can be decomposed as the sum of a nonlinear image-dependent mapping, a linear noise-dependent term and a small residual term. We show that these denoisers can be trained with only noisy images under the condition that the noise has zero mean and known variance. 
The exact distribution of the noise, however, is not assumed to be known. We show the superiority of our approach for image denoising, and demonstrate its extension to solving other restoration problems such as image deblurring where the ground truth is not available.
Our method outperforms some recent unsupervised and self-supervised deep denoising models that do not require clean images for their training. For deblurring problems, the method, using only one noisy and blurry observation per image, reaches a quality not far away from its fully supervised counterparts on a benchmark dataset.
\end{abstract}


\section{Introduction}

\IEEEPARstart{T}{he}
acquisition of real life images usually suffers from noise corruption due to different factors such as sensor sensitivity, variations of photon numbers, air turbulence, just to name a few.
Image denoising \cite{rudin1992nonlinear,mallat1992singularity, donoho1995noising, weickert1998anisotropic} aims to remove noise from corrupted images. It is one of the most fundamental and central problems tackled by the image processing community. 
A variety of image denoising approaches have been developed in the past decades, which can be broadly classified into model based denoisers (see e.g., \cite{rudin1992nonlinear,chan2005salt, dabov2007image, buades2005non}) and data driven denoisers \cite{roth2005fields, chen2016trainable, zhang2017beyond}.

Recent data driven techniques
outperform conventional model based methods and achieve the state of the art denoising quality \cite{liu2018non, zhang2017beyond, zhang2020residual, lefkimmiatis2018universal}. 
These methods take advantage of large sets of image data and use the models, particularly deep convolutional neural networks (CNN), to learn the image prior from the datasets rather than relying on predefined image features. 
Compared to many model based methods, which need to solve a difficult optimization problem for each test image, CNN based denoisers are computationally efficient once the CNN is trained. Nevertheless, CNN based denoising approaches usually require a massive amount of ground truth images in the training phase. Specifically, conventional training pipelines consist of a loss function or a metric, which defines the distance between a clean ground truth image and its reconstructed version, and an optimization step in which the parameters of the models are adjusted so as to minimize the loss function. One of the most commonly used metrics is the mean squared error (MSE), the calculation of which depends explicitly on the ground truth image.
While these learning processes can lead to high-quality denoisers, they are problematic for application domains where ground truth images are not accessible. 

In the past years, several unsupervised deep learning techniques have been developed to overcome this difficulty. It is found that it is possible to train a deep neural network denoiser by using only the noisy data if multiple versions of noisy images are available for each unknown clean image \cite{lehtinen2018noise2noise}, or if the noise is independent within different regions of the image \cite{batson2019noise2self, krull2019noise2void, krull2019probabilistic}. 
Under the assumption of i.i.d. Gaussian noise, loss functions can also be adapted, based on Stein's unbiased risk estimate of the MSE  \cite{soltanayev2018training},
such that they are defined only on the noisy images while providing good approximations of the MSE. 
\begin{figure}[ht!]
    \centering
        \includegraphics[width=1\linewidth, trim=100 0 90 0, clip]{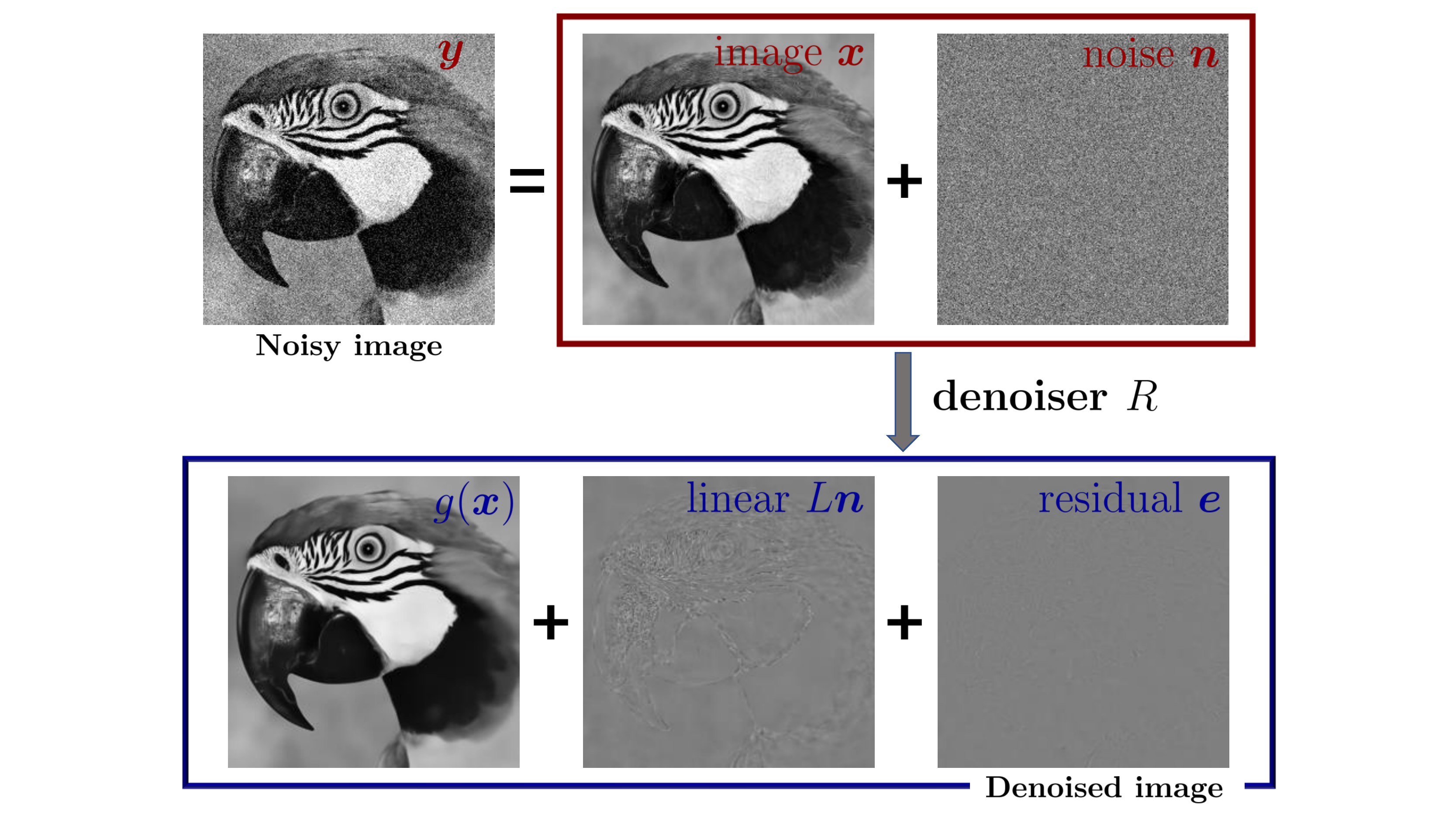}
    \caption{Partially linear denoiser. 
    The clean image $\bmx$ (top middle), the noise $\bmn$ (top right), and the noisy image $\bmy:=\bmx+\bmn$ (top left) are modeled as random variables.
    A denoiser $R$ is decomposed as $R(\bm{y}) = g(\bm{x}) + L\bm{n} + \bm{e}$ with
    a function $g\qut{\bmx}$ (can be nonlinear), 
    a linear mapping $L$ (can depend on $\bmx$), and a residual term $\bme$. If the random variable $\bme$ is of \textbf{small variance}, then $R$ is called a \textit{partially linear denoiser}. Such denoisers can be learned from only noisy images.} 
    \label{fig:diagram}
\end{figure}

In this work, we address the problem of learning efficient denoisers from a set of \textit{noisy images} without the need for precise modeling of the noise, and without multiple noisy observations per image.
We investigate a class of structured nonlinear denoisers and their applications to \textbf{inverse imaging problems} including denoising and deblurring. 
Specifically, we formulate a given denoiser $R\qut{\cdot}$
as the sum of three terms (cf. Fig. \ref{fig:diagram}), including a (nonlinear) function of the ground truth signal $g\qut{\bmx}$, a linear mapping of the noise $L\bmn$, and a residual term $\bme$. 
If the linear factor $L\bmn$ is relatively small and the residual term $\bme$, as a \textit{random variable}, has small variance, then the denoised image mainly depends on the ground truth image and is not sensitive to the changes of the noise. In this paper, we study denoisers with the property that $\bme$ has \textbf{small variance}, and we call denoisers of this type \textit{partially linear denoisers} (cf. Fig. \ref{fig:diagram}). Such denoisers preserve nonlinear image structure (similar to many classic image denoisers), and additionally allow the nonlinearity to be learned from noisy data only.

We observe that some natural denoisers, including deep neural networks, exhibit certain degrees of partial linearity.  By exploiting this property, we show that a partially linear denoiser, in the learning setting, can be trained on noisy images by only assuming that the noise is zero mean conditioned on the images with known variance. Specifically, we introduce an auxiliary random vector together with a partially linear constraint, that we show %
establishes a correspondence between the MSE and a loss function defined without clean images. 
By doing this one can approximate the best partially linear denoisers with a theoretically guaranteed approximation error bound. Moreover, 
we propose an image deblurring approach, 
based on the partially linear denoisers,
that learns from single noisy observation of the images.

Different from other existing unsupervised methods for denoising, our approach does not require a group (or a pair) of noisy observations for each image or an estimate on the underlying noise distribution. Yet our approach leads to a performance close to the fully supervised counterparts on some denoising benchmark datasets.
The proposed partially linear denoisers are learned in an end-to-end manner and can be built on top of any deep neural network architectures for denoising. Once they are trained, the denoised results are obtained without the auxiliary vectors or any post-processing. Furthermore, we demonstrate our denoiser's capacity of handling other image restoration tasks in the absence of ground truth or multiple noisy observations per image, by utilizing the direct approximation to the MSE and its end-to-end learning nature. 

\section{Related work}
In the past decades, a wide variety of image denoising algorithms were developed. They range from analytic approaches such as filtering, variational methods, wavelet transforms, Bayesian estimation to data-driven approaches such as deep learning. In this section we review some major nonlinear denoisers as well as recent deep learning approaches that do not depend on ground truth images for learning. 

\subsection{Nonlinear denoisers}
Natural images are non-Gaussian signals with fine structures such as sharp edges and textures. 
One of the main challenges in the image denoising task is to preserve these structures while removing the noise. 
Traditional linear denoisers such as Gaussian filtering and Tikhonov regularization usually do not achieve satisfactory denoising quality, as they tend to smooth the edges. Most of the existing image denoisers in the literature are therefore nonlinear. 

Total variation (TV) denoising \cite{rudin1992nonlinear} is one of the most fundamental image denoising algorithms. The TV denoising solution is a minimizer of the optimization problem
\[
\min_{x} \frac{\lambda}{2} \| x - y \|^2 + \| x \|_{\rm TV}
\]
in which $y\in\mathbb R^m$ is the noisy image and $\| x \|_{\rm TV}:=\|\nabla x\|_1$ is the total variation of $x\in\mathbb R^m$. TV denoising is known to preserve sharp edges thanks to the non-linearity introduced by the TV norm $\| \cdot \|_{\rm TV}$. 

Patch based methods have gained popularity for image denoising tasks because of their capacity of capturing the self-similarity of images. 
The non-local means (NL-means) algorithm proposed by Buades et al. \cite{buades2005non} is among the most
successful methods in this category. The NL-means algorithm removes noise by calculating a weighted average of the pixel intensities, with weights defined based on a patch similarity measure which emphasizes the connection among pixels in similar patches. It is a nonlinear denoiser, different from the local mean filtering, as the weights are image-dependent. As another nonlinear patch based denoising method, the Block-matching and 3D filtering (BM3D) algorithm \cite{dabov2007image} divides image patches into groups based on a similarity criterion, and collaborative filtering on each of the groups is then performed to clean the patches.  
While patch based denoisers achieve promising denoising performance, they are often time-consuming due to the high computational complexity in calculating the weights or matching with similar patches.

In recent years, convolutional neural networks (CNN) based denoisers became the state of the art for image denoising \cite{zhang2017beyond, zhang2020residual}, thanks to the rapid development of deep learning techniques. In particular, CNN are known to be efficient in modeling image priors \cite{ulyanov2018deep}, which are crucial for the quality of the denoiser. A typical CNN denoiser can be formulated as a composition of mappings called layers, and a basic type of layers $y^{(k)} \rightarrow y^{(k+1)}$ has the form 
\[
y^{(k+1)} = \sigma \qut{W * y^{(k)} + b},
\]
where $*$ denotes the convolution operation, $(W,b)$ are parameters of the network that can be learned from the data, and $\sigma\qut{\cdot}$ is an activation function which is often nonlinear. As such, CNN denoisers are in general nonlinear. There has been growing interest in developing CNN denoisers with new network architectures and building blocks being proposed, such as batch-normalization \cite{zhang2017beyond}, residual connection \cite{fu2017removing}, and residual dense block \cite{zhang2020residual}.

\subsection{Unsupervised deep learning for denoising}
While deep CNN based models have great advantages in image denoising, the most standard learning techniques are limited to the availability of sufficiently many noise-free images. 
Recently, CNN-based learning algorithms for image denoising that do not require clean images as training data have been developed. Soltanayev et al. \cite{soltanayev2018training} propose to use Stein’s unbiased risk estimator (SURE) based loss function, which is computed without knowing ground truth images, as a replacement of the MSE loss function. In the setting of i.i.d. Gaussian noise, SURE provides an unbiased estimate of the MSE, and hence the minimization problems with respect to these two losses are equivalent. Consequently, CNN denoisers can be trained in an unsupervised manner, based on only one noisy realization of each training image. Nevertheless, SURE may not be identical to the MSE in the case of non-Gaussian noise, e.g., shot noise. 
The SURE training scheme has been extended to the setting of multiple noise realizations per image as well as imperfect ground truths \cite{zhussip2019extending}. 

The Noise2Noise approach developed by Lehtinen et al. \cite{lehtinen2018noise2noise} offers a different ground-truth free learning strategy for denoising, based on the assumption that for each image at least two different noisy observations are available. It is found that replacing the clean targets by their noisy observations in the MSE loss function leads to the same minimizers of the original supervised loss if the noise has zero mean and an infinite amount of training data are provided. Specifically, the Noise2Noise loss function is
\begin{equation}\label{eq:n2nloss}
    \sum_{\qut{y, \hat{y}} \in \qut{\calY, \hat{\calY}} } \| f_\theta \qutl{\hat{y}} - y \|^2
\end{equation}
where $y$ and $\hat{y}$ represent two independent noisy observations of the same image sample, and $f_\theta$ is the denoiser parameterized by $\theta$. 
In contrast to the MSE loss
$\sum_{\qut{y, x} \in \qut{\calY, \calX }} \| f_\theta \qutl{y} - x \|^2$
where $x$ denotes the ground truth image, the loss function \eqref{eq:n2nloss} is defined on a noisy pair $\qut{y, \hat{y}}$ only. If the noise in $y$ and the noise in $\hat{y}$ are independent and have zero means, it is shown that the minimization of this loss function \eqref{eq:n2nloss} with respect to $\theta$ is equivalent to minimizing the MSE \cite{lehtinen2018noise2noise}. 
This implies that the parameter $\theta$ can be computed from a training set of noisy pairs. Besides, the Noise2Noise approach does not rely on an explicit image prior or on structural knowledge about the noise models. 

In certain denoising tasks, however, the acquisition of two or more noisy copies per image can be very expensive or impractical, in particular in medical imaging where patients are moving during the acquisition, or in videos with moving cars, etc. 
Several authors have investigated learning techniques that overcome this restriction \cite{ehret2019model, valery2020self, batson2019noise2self, krull2019noise2void, krull2019probabilistic}.
The Frame2Frame \cite{ehret2019model} developed by Ehret et al. fine-tunes a denoiser for blind video denoising. The idea is to use optical flow to warp one video frame to match its neighboring frame, and then minimize the Noise2Noise loss \eqref{eq:n2nloss} with $(y, \hat{y})$ being the pair of matched frames.
The work \cite{moran2020noisier2noise} generalizes the Noise2Noise into the setting of a single noisy realization for each image. A synthetic noise, that is drawn from the same distribution as the underlying noise, is added to the noisy image $y$, and the new noisy image is then used to replace the second observation $\hat{y}$ in the training loss \eqref{eq:n2nloss}. At test time, the raw output of the network is post-processed to obtain the denoising results, by computing a linear combination of the output and the input. In this approach, the synthetic noise has to be of the same noise type as the underlying noise, and the post-processing can magnify the errors of the output of the network.
Batson et al. \cite{batson2019noise2self} show that a denoiser $f_\theta\qut{\cdot}$, satisfying a so called $\mathcal{J}$-invariant property for a partition $\mathcal{J}$ of the image pixels, can be trained without accessing a second observation of the image if the noise is independent across different regions in $\mathcal{J}$. 
A function $f_\theta \qutl{\cdot}$ is said to be $\mathcal{J}$-invariant if for each subset of pixels $J \in \mathcal{J}$, the pixel values of $f_\theta \qut{y}$ at $J$ can be calculated without knowing the values of $y$ at $J$.
By additionally assuming that the noises of $y$ at different elements of $\mathcal{J}$ are independent and have zero mean, one minimizes the loss
\begin{equation}\label{eq:n2sloss}
\begin{split}
    \sum_{y \in \calY}  \qutn{ f_\theta \qut{y} - {y} }^2 & = \sum_{y \in \calY} \sum_{J \in \mathcal{J}} \qutn{ [f_\theta \qut{y}]_J - y_J }^2 \\
    & = \sum_{J \in \mathcal{J}} \sum_{y \in \calY} \qutn{ [f_\theta \qut{y}]_J - y_J }^2 
\end{split}
\end{equation}
where the subscript $J$ in $y_J$ denotes the restriction of the image $y$ to the pixel collection $J$. It should be noted that, for a $\mathcal{J}$-invariant function $f_\theta\qut{\cdot}$ and any $J \in \mathcal{J}$, the loss 
$\sum_{y \in \calY} \qutn{ [f_\theta \qut{y}]_J - y_J }^2$ 
can be interpreted as a variant of \eqref{eq:n2nloss}, given the fact that the noise contained in $y_J$ is independent of $[f_\theta \qut{y}]_J$ and has zero mean.
This enables the training of a model using a set of noisy images $y$ only. 
However, as the denoiser is a $\mathcal{J}$-invariant function, the images can not be perfectly reconstructed as the level of noise decreases to zero, i.e., it can not learn the identity mapping which is clearly not $\mathcal{J}$-invariant. The unused information from $y$ can be further leveraged if the noise model is known or can be estimated.  Krull et al. \cite{krull2019probabilistic} propose to combine $y$ with the network predictions, based on a probabilistic model for each pixel and a Bayesian formulation, to obtain the minimal MSE estimate.
The integration of statistical inference effectively removes the noise remained in the predictions. A downside of this method is that it requires an explicit expression of the posterior distribution of the noisy images beforehand, and it is not end-to-end because the network outputs intermediate results that are improved in the post-processing step.

\section{The proposed method}\label{sec:method}
We consider image restoration problems of the form
\begin{equation}\label{eq:generalproblem}
\bmy = A \bmx + \bmn \end{equation} 
in which $\bmx$ and $\bmn$ are random vectors of \textit{unknown distributions} representing the ground truth images and the noise respectively, $\bmy$ is the noisy image from which we want to restore $\bmx$, and $A$ is a linear forward operator determined by the data acquisition process.  
In this paper, random vectors are always denoted by boldface lower-case letters like $\bmx$, $\bmy$, $\bmz$, etc. For ease of presentation, in this section we first focus on denoising problems, i.e., $A:=I$ defines an identity map. The more general cases where $A$ is not the identity will be discussed in Subsection \ref{subsect:deblurring_1}. 

For Problem \eqref{eq:generalproblem}, the only assumption we make on the noise distribution is that it has zero mean conditioned on the image, i.e.,
\begin{itemize}[noitemsep]
    \setlength\itemindent{15pt} 
    \item[\textbf{(A1)}.] $\bbE\qut{\bmn \mid \bmx} = 0$.
\end{itemize}
This assumption holds for various common types of noises including Gaussian white noise, Poisson noise, as well as some mixed noises. We underline that the noise $\bmn$, however, are not assumed to be independent of the image or pixel-wise independent. %

The central issue in image denoising is to find an operator, denoted by $R$, that takes the noisy image $\bmy$ to the clean image $\bmx$ or its approximations. In this work we are interested in a class of denoisers that can be decomposed as
\begin{equation}\label{eq:part-lin-pure}
    R\qut{\bmy} = g\qut{\bmx} + L \bmn + \bme,
\end{equation}
where $g$ is a (possibly nonlinear) function of the clean image $\bmx$, $L:=L_{\bmx}$ is a linear operator and $\bme:= \bme_{\bmx,\bmn}$ is a residual term with \textbf{small variance}. 
In the rest of the paper, we omit the subscripts of  $L_{\bmx}$ and $\bme_{\bmx,\bmn}$ if there is no ambiguity.

\begin{remark}\label{remark1}
For a given denoiser, 
the decomposition \eqref{eq:part-lin-pure} always exists, but in this work we only consider the setting of $\bme$ having small variance. 
Furthermore, the decomposition \eqref{eq:part-lin-pure} is not unique, given the fact that $L$ can be an arbitrary  linear operator.
However, in order to characterize the structure of the denoiser, we assume $\bbE\qut{\bme \mid \bmx} = 0$ and let $L$ be chosen such that $\bbE\qut{ \qutn{\bme}^2 \mid \bmx}$ is minimized. As $\bmn$ has zero mean, the nonlinear term is then determined by $g\qut{\bmx} = \bbE\qut{R\qut{\bmy} \mid \bmx}.$
\end{remark}

If $R$ satisfies \eqref{eq:part-lin-pure} with $\bme$ of small variance, then we call $R$ a \textit{partially linear denoiser}. The first term $g\qut{\bmx}$, which can be nonlinear, does not depend on realizations of the noise $\bmn$. This formulation implies that the non-linearity of the denoisers in this class is mainly due to intrinsic image structures, and the denoisers respond to the noise in a linear (encoded in $L \bmn$) or less sensitive manner (encoded in the fact that $\bme$ has small variance). In fact, for any denoiser that can effectively remove noise from images, its output has to be minimally dependent on the changes of noise, and therefore when written in the form \eqref{eq:part-lin-pure}, the noise dependent components on the right hand side should be small in variance compared to the noise, which implies that $\bbE(\|\bme\|^2)$ is small.

The selection of a good denoiser $R$ requires knowledge of certain prior information about the target $\bmx$, especially when we want to find the fine details like edges from the corrupted image. Many conventional analytical approaches aim to find the reconstruction $R\qut{\bmy}$ from a lower dimensional space that the images lie in. This can be done, for example, by assuming some sparseness properties in the gradient fields or the wavelet domains. 

In a data-driven setting, given the clean image $\bmx$, one could instead minimize the \textit{mean square error} (MSE) defined as
$\mathcal{J}_{0}\qut{R} := \bbE( \qutn{R(\bmy) - \bmx}^2),$
where the expectation is taken over $\bmx$ and $\bmn$. The minimization of $\mathcal{J}_{0}$ leads to the conditional mean $R_{0}(\bmy):=\bbE\qut{\bmx \mid \bmy}$. Next, we will discuss how to approximate $R_{0}$ in the absence of the clean image $\bmx$. 

\subsection{Auxiliary random vectors}
The motivation for this paper comes from the fact that, in practice the distribution of $\bmx$ is often unknown, and samples of $\bmx$ (i.e., the ground truth images) or of the noise $\bmn$ are not readily available. What can be easily accessed are noisy observations $\bmy$. With these alone, however, a direct evaluation of the MSE is not possible.
In this work, we circumvent the need for $\bmx$ 
by introducing an \textit{auxiliary random vector} and replacing the MSE 
$\mathcal{J}_0$
by an approximation. 
First, let $\bmz$ be a random vector satisfying assumptions
\begin{itemize}[noitemsep]
    \setlength\itemindent{15pt}
    \item[\textbf{(A2)}]  the conditional mean $\bbE\qut{\bmz \mid \bmx} = 0$,
    \item[\textbf{(A3)}]  $\bmz$ is independent of $\bmn$,
    \item[\textbf{(A4)}]  
    the conditional covariance:
    \[{\rm Cov}\qut{\bmz, \bmz \mid \bmx} = {\rm Cov}\qut{\bmn, \bmn \mid \bmx}.\]
\end{itemize}
The auxiliary vector $\bmz$ does not necessarily need to follow the same distribution as $\bmn$. Samples of $\bmz$, therefore, can be easily generated from e.g., Gaussian distributions, once the variances of $\bmz$ are known.
Then, associated with $\bmz$, we define 
\begin{equation}\label{eq:nhat_yhat}
    \begin{cases}
        \bmnh & := \bmn + \alpha \bmz\\
        \bmyh & := \bmy + \alpha \bmz,
    \end{cases}
\end{equation}
where $\alpha$ is a real constant. The new random vector $\bmyh = \bmx + \bmnh$ can be regarded as a noisy version of image $\bmx$ with noise vector $\bmnh$. In the following, the discussion will focus on denoising $\bmyh$ rather than denoising $\bmy$, but the objective remains unchanged, i.e., getting the same clean image $\bmx$. Specifically, if one can obtain a high-quality solution $\bmx$ from $\bmyh$, then there exists an algorithm taking $\bmy$ to the clean image because $\bmyh$ can be computed  from $\bmy$. Such an algorithm can be achieved, for instance, by $R\qut{V\qut{\bmy}}$ where $R$ is a denoiser for the $\bmyh$ problem and $ V(\bmy):=\bmyh$. 
It should be noted that the difficulty of the denoising problem is raised because of the additional uncertainty from the auxiliary random vector $\bmz$ encoded in the observed data $\bmyh$. However, one of the benefits of considering $\bmyh$ is that the quantity $\bmz$ is known and can be leveraged for constructing approximations to the MSE without knowing $\bmx$ or $\bmn$ as we will show in the following.

For the $\bmyh$ denoising problem, the MSE associated with the denoiser $R$ is defined as
\begin{equation}\label{eq:mse-new}
    \mathcal{J}_{\rm mse}\qut{R} := \bbE\qut{ \qutn{R(\bmyh) - \bmx}^2 }
\end{equation}
where the expectation $\bbE$ is taken over random variables $\bmx$, $\bmn$ and $\bmz$. This is connected to $\mathcal{J}_0$ via $\mathcal{J}_{\rm mse} (R) = \mathcal{J}_{0} (R\qut{V})$. Since $V$ is known, it can be shown that $\min_R \mathcal{J}_0(R) \leq \min_R \mathcal{J}_{\rm mse}(R)$. Nevertheless, if $\alpha$ is close to zero, the noise distributions of $\bmy$ and $V(\bmy)$ are close, so the minimum of $\mathcal{J}_0$ can be approximated by the minimum of $\mathcal{J}_{\rm mse}$ which promises similar denoising quality.

Now, using the auxiliary vector $\bmz$, we define the following objective function for our proposed partially linear denoiser
\begin{equation}\label{eq:newmse}
    \mathcal{J}\qut{R} := \bbE\qut{ \qutn{R\qut{\bmyh} - \qut{\bmy - \bmz/\alpha}}^2 }.
\end{equation}
for $\alpha \neq 0$. Indeed, as we will see later in Subsections \ref{subsect:lmnme} and \ref{subsect:opld}, if we consider the partially linear denoising model, then $\mathcal{J}$ provides a good estimate of $\mathcal{J}_{\rm mse}$. More precisely, according to the definition \eqref{eq:part-lin-pure}, we consider a set of denoisers for $\bmyh$ that have the form 
\begin{equation}\label{eq:part-lin}
    R\qut{\bmyh} = g\qut{\bmx} + L \bmnh + \bme
\end{equation}
where $\bmnh$ is defined in \eqref{eq:nhat_yhat}, and $L$ and $\bme$ are depending on $\bmx$ and $\qut{\bmx,\bmnh}$, respectively, and the residual $\bme$ is of small variance.

In the rest of this section, we show that for $\alpha\!\neq\! 0$ and partially linear denoiser $R$, the term $\mathcal{J}\qut{R}$ in \eqref{eq:newmse}
approximates the MSE $\mathcal{J}_{\rm mse}\qut{R}$ up to an additive constant.
This implies that an optimal denoiser of this class can be computed even if the distribution of ground truth images $\bmx$ and the distribution of $\bmn$ are not known. We then discuss unsupervised learning methods based on our proposed objective $\mathcal{J}\qut{R}$
for image restoration tasks like denoising and deblurring. 

\subsection{Linear minimum mean square error estimator}\label{subsect:lmnme}
To start with, we focus the discussion on the simplest case, in which the denoiser $R$ is linear.
This is a subset of the set of partially linear denoisers \eqref{eq:part-lin} with $g = L$ and $\bme = 0$.
The best estimator in this setting is the linear minimum mean square error (LMMSE) estimator, i.e., the minimizer of 
$\mathcal{J}_{\rm mse}\qut{R}$ in \eqref{eq:mse-new}
over all linear denoisers $R$.
The following proposition establishes the equivalence between the MSE loss and $\mathcal{J}\qut{R}$ in \eqref{eq:newmse} over linear denoisers.

\begin{proposition}\label{prop:1}
If $\bmn$ and $\bmz$ satisfy Assumption (A1) - (A4) and $R$ is linear, then there is some constant $c$ not depending on $R$, such that 
\begin{equation}\label{eq:prop1-eq0}
\mathcal{J}\qut{R}
= 
\bbE\qut{ \qutn{{R \bmyh  - \bmx}}^2} + c.
\end{equation}
\end{proposition}
\begin{proof}
If $R$ is linear, then from the definitions of $\bmy$ and $\bmyh$,
\[
\begin{split}
    R{\bmyh} - \qut{\bmy -  \bmz / \alpha} 
     =  \qut{R{\bmyh} - \bmx}
     - \qut{\bmn - \bmz/\alpha},
\end{split}
\]
and $R{\bmyh} - \bmx = \qut{R \bmx  - \bmx}+\qut{R\bmn + \alpha R\bmz }$. 
Let $\qutan{\cdot,\cdot}$ denote the inner product operator. Then, using the above, \eqref{eq:newmse} can be rewritten as
\begin{equation}\label{eq:lmmse-decompose}
    \begin{split}
    \mathcal{J}\qut{R}
    = 
    & \bbE\qut{ \qutn{R\bmyh - \bmx}^2}
     + \bbE\qut{ \qutn{\bmn - \bmz/\alpha}^2} \\
    & - 2\bbE\qut{ \qutan{R\bmx - \bmx, \bmn - \bmz/\alpha}} \\
    &
    - 2\bbE\qut{ \qutan{R\bmn+\alpha R\bmz, \bmn - \bmz/\alpha}}.
    \end{split}
\end{equation}
Here the expectations are taken over $\bmx$, $\bmn$ and $\bmz$. We next show that the last two terms on the right hand side of \eqref{eq:lmmse-decompose} vanish.

First, since $\bmn$ has zero mean conditioned on $\bmx$ according to Assumption (A1), it holds that 
$
\bbE\qut{\qutan{R{\bmx}-\bmx, \bmn}} = 
\bbE\qut{ \bbE\qut{\qutan{R{\bmx} - \bmx, \bmn} \mid \bmx} } = 0.
$
The same property applies to $\bmz$ by Assumption (A2), so $\bbE\qut{\qutan{R{\bmx}-\bmx, \bmz}} = 0$. With a linear combination of these two equalities, 
\begin{equation}\label{eq:prop1_eq1}
\bbE\qut{ \qutan{R \bmx - \bmx, \bmn - \bmz/\alpha}} = 0.     
\end{equation}

Second, given the fact that $R$ is linear and both $\bmn$ and $R\bmn$ have zero mean conditioned on $\bmx$, we have
\[
\begin{split}
  \bbE\qut{\qutan{ R{\bmn}, \bmn }\mid \bmx}
= & {\rm tr} \qut{{\rm Cov}\qut{ R{\bmn}, \bmn \mid \bmx}} \\
= & {\rm tr} \qut{R \ {\rm Cov}\qut{\bmn, \bmn \mid \bmx}} \\
\end{split}
\]
in which ${\rm tr}\qut{\cdot}$ denotes the trace of the matrix. The same equality holds for $\bmz$, i.e.,
\[
\bbE\qut{\qutan{ R{\bmz}, \bmz }\mid \bmx} = 
 {\rm tr} \qut{R \ {\rm Cov}\qut{\bmz, \bmz \mid \bmx}}
\]
Therefore, it follows from Assumption (A4) that 
$
\bbE\qut{\qutan{ R{\bmn}, \bmn }\mid \bmx}
=
\bbE\qut{\qutan{ R{\bmz}, \bmz }\mid \bmx}. 
$
This together with Assumption (A3) gives
\begin{equation}\label{eq:noise_compensate}
\begin{split}
& \bbE\qut{\qutan{ R{\bmn} + \alpha R{\bmz}, \bmn - \bmz/\alpha } \mid \bmx}\\ 
= & \bbE\qut{\qutan{ R{\bmn}, \bmn }\mid \bmx} +
\bbE\qut{\qutan{\alpha R{\bmz},  -\bmz/\alpha }\mid \bmx} 
=  0.
\end{split}
\end{equation}
Taking expectation with respect to $\bmx$, we have
\begin{equation}\label{eq:prop1_eq2}
\bbE\qut{ \qutan{R\bmn+\alpha R \bmz, \bmn - \bmz/\alpha}} = 0.
\end{equation}

Finally, let $c:= \bbE\qut{\qutn{\bmn - \bmz/\alpha}^2}$, then Equation \eqref{eq:lmmse-decompose}, \eqref{eq:prop1_eq1} and  \eqref{eq:prop1_eq2} imply
\eqref{eq:prop1-eq0}
which completes the proof. 
\end{proof}

According to Proposition \ref{prop:1}, for all linear denoisers, the term $\mathcal{J}\qut{R}$ in \eqref{eq:newmse} differs from the MSE by a constant $c$ not depending on $R$.
Therefore a minimization of $\mathcal{J}\qut{R}$ over all linear $R$ also leads to the LMMSE estimator. We remark that the objective function \eqref{eq:newmse} is not defined based on the ground truth image $\bmx$, and the distribution of the random vector $\bmn$ is not necessarily known or equal to that of $\bmz$. 
In fact, the denoising problem can be reformulated as follows. Given $R$, the quantity $R\bmyh$ is known, and we want to decouple the term $R\bmyh-\bmx$ from
\begin{equation}\label{eq:compensate-1}
   R\bmyh-\bmy=(R\bmyh-\bmx)-\bmn.
\end{equation}
In the loss $\mathcal{J}\qut{R}$, this is implemented by adding the scaled auxiliary vector $\bmz/\alpha$ to both sides of \eqref{eq:compensate-1}. The vector $\bmz/\alpha$ compensates the noise $\bmn$ in the sense that, when taking the expectation of $\qutn{R\bmyh - \bmy+\bmz/\alpha}^2$ over $\bmn$ and $\bmz$, the $\bmn$-related term $\bbE\qut{\qutan{ R{\bmn}, \bmn }\mid \bmx}$ is canceled by the $\bmz$-related term $\bbE\qut{\qutan{ R{\bmz}, \bmz }\mid \bmx}$ (i.e., the last equality in \eqref{eq:noise_compensate}), which leads to \eqref{eq:prop1-eq0}.

\subsection{Optimal partially linear denoiser}\label{subsect:opld}
Though linear denoisers have nice properties that link \eqref{eq:newmse} to the MSE, they may not be the best denoisers for imaging data. Nonlinearity is unavoidable in order to achieve good denoising quality and preserve fine structures such as edges in the image. 
To this end, we consider a more general class of denoisers that are partially linear, as expressed in \eqref{eq:part-lin}, where $g$ is nonlinear and $\bme$ is not necessarily zero. 

In a special case where the denoiser outputs exactly the clean image, i.e. $R\qut{\bmyh} := \bmx$, $R$ is also partially linear with $L$ and $\bme$ being both zero. However, such denoisers do not exist in most settings. Nevertheless, for good denoisers one could still expect the residual term $\bme$ to be small.

The equivalence of the two objective functions stated in Proposition \ref{prop:1} does not hold for nonzero $\bme$. It is therefore of interest to know how well $\eqref{eq:newmse}$ approximates the MSE in this general case. The next proposition quantifies the approximation error in the presence of a small nonzero residual term $\bme$.

\begin{proposition}\label{prop:2}
    If $\bmn$ and $\bmz$ satisfy Assumption (A1) - (A4) and $R$ satisfies \eqref{eq:part-lin}, then there is some constant $c$ not depending on $R$, such that 
    \begin{equation}\label{eq:part-lin-mse}
        \mathcal{J}\qut{R}
        = 
    \bbE\qut{ \qutn{{R(\bmyh) - \bmx}}^2} - 2 \bbE\qut{\qutan{\bme, {\bmn - \bmz/\alpha}}} + c
    \end{equation}
    Additionally, if the variance $\bbE\qut{\qutn{\bme}^2} \leq \epsilon^2$ for $\epsilon>0$, then the approximation error 

\begin{equation}\label{eq:bound1}
        \begin{split}
        {\rm Err} := & \left| %
        \mathcal{J}\qut{R}-
        \bbE(\qutn{{R(\bmyh) - \bmx}}^2) - c \right| 
        \\
        \leq  
        & 
        2 \epsilon \sqrt{ \bbE(\qutn{ \bmn - \bmz/\alpha }^2) }
    \end{split}
\end{equation}
If furthermore the residual $\bme\qut{\bmx, \bmnh}$ is Lipschitz continuous with respect to $\bmnh$ 
i.e., $\|\bme\qut{\bmx, \bmnh_1 } - \bme\qut{\bmx, \bmnh_2}\| \leq K \| \bmnh_1 - \bmnh_2 \|$ for some constant $K$, then 
\begin{equation}\label{eq:boudn2}
      {\rm Err} \leq 2\epsilon \sqrt{ \bbE(\qutn{ \bmn }^2) } + 2K \bbE\qut{\qutn{\bmz}^2}.
    \end{equation}
\end{proposition}
\begin{proof}
By the decomposition of $R$ in \eqref{eq:part-lin},
\begin{equation}\label{eq:prop2-eq1}
    \begin{split}
    R\qut{\bmyh} - \qut{\bmy - \bmz/\alpha} 
     = & \qut{g\qut{\bmx} + L \bmnh + \bme - \bmx} \\ & - \qut{\bmn - \bmz/\alpha},
\end{split}
\end{equation}
where $\hat{\bmn}$ is defined in \eqref{eq:nhat_yhat}. For function $g\qut{\cdot}$, from assumption (A1) and (A2) it is clear that
$\bbE\qut{\qutan{g\qut{\bmx}, {\bmn - \bmz/\alpha}} \mid \bmx} = 0$. Therefore,
\begin{equation}\label{eq:prop2-eq2-1}
     \bbE\qut{ \qutan{g\qut{\bmx}-\bmx, {\bmn - \bmz/\alpha}}} = 0.
\end{equation}
Moreover, as $L$ is linear, with a similar argument to the proof of \eqref{eq:prop1_eq2}, we can show that
\begin{equation}\label{eq:prop2-eq2-2}
     \bbE\qut{ \qutan{L\bmnh, {\bmn - \bmz/\alpha}}} = 0. 
\end{equation}
According to Equation \eqref{eq:prop2-eq1}, \eqref{eq:prop2-eq2-1} and \eqref{eq:prop2-eq2-2}, if we let $c:= \bbE\qut{\|\bmn - \bmz/\alpha\|^2}$, then Equation \eqref{eq:part-lin-mse} holds.

Next, we consider the bound of $\bbE\qut{\qutan{\bme, {\bmn - \bmz/\alpha}}}$.
Assuming that $\bbE(\qutn{\bme}^2) \leq \epsilon^2$, 
a straightforward application of Cauchy-Schwarz inequalities leads to 
\begin{equation}\label{eq:CS-ineq}
    \begin{split}
    \left| \bbE \qut{ \qutan{\bme, {\bmn - \bmz/\alpha}} } \right| 
    & \leq \sqrt{ \bbE\qut{\qutn{\bme}^2}} 
    \sqrt{ \bbE\qut{\qutn{ \bmn - \bmz/\alpha }^2} } \\
    & = \epsilon 
    \sqrt{ \bbE\qut{\qutn{ \bmn - \bmz/\alpha }^2} },
    \end{split}
\end{equation}
and therefore Equation \eqref{eq:bound1} holds.

Furthermore, assume that $\bme:= \bme\qut{\bmx, \bmn + \alpha \bmz}$ is Lipschitz continues w.r.t. $\bmnh:=\bmn + \alpha \bmz$ with Lipschitz constant $K$.
Since $\bmz$ has zero mean conditioned on $\bmx$ and $\bmn$, the conditional expectation
\[
\begin{split}
\bbE \qut{ \qutan{\bme, \bmz/\alpha} \mid \bmx, \bmn }
 = 
& \bbE \qut{ \qutan{\bme_0, \bmz/\alpha} \mid \bmx, \bmn } \\
& + \bbE \qut{ \qutan{\bme - \bme_0, \bmz/\alpha} \mid \bmx, \bmn }\\
\leq & 
0 + 
\bbE \qut{ K \qutan{\bmz, \bmz} \mid \bmx, \bmn },
\end{split}
\]
where $\bme_0 := \bme\qut{\bmx, \bmn + \bm{0}}$. So $\bbE \qut{ \qutan{\bme, \bmz/\alpha}}\leq K \bbE\qut{\qutn{\bmz}^2}$, and because of symmetry we have $\bbE \qut{ \qutan{\bme, \bmz/\alpha}}\geq - K \bbE\qut{\qutn{\bmz}^2}$.  Finally, 
\[
\begin{split}
\left| \bbE \qut{ \qutan{\bme, {\bmn - \bmz/\alpha}} } \right| \leq & \left| \bbE \qut{ \qutan{\bme, \bmn} } \right| + \left| \bbE \qut{ \qutan{\bme, {\bmz/\alpha}} } \right|\\
\leq & \epsilon \sqrt{ \bbE\qut{\qutn{ \bmn }^2} } + K \bbE\qut{\qutn{\bmz}^2}
\end{split}
\]
and the inequality \eqref{eq:boudn2} follows. 
\end{proof}

From the bounds given in Proposition \ref{prop:2}, if the variance of $\bme$ is small, then the error \eqref{eq:newmse} provides a good estimate to the MSE. Note that making a Lipschitz continuity assumption on $\bme$, the error bound in \eqref{eq:bound1} can be further improved to \eqref{eq:boudn2} which is independent of $\alpha$ (in contrast to \eqref{eq:bound1}  where the bound has a $1/\alpha$ factor). In practice, this Lipschitz continuity assumption is not restrictive as we expect that the denoiser $R$ is stable with respect to small perturbations in $\bmyh$ and therefore has a small Lipschitz constant $K$ for $\bme$.
\begin{remark}[Estimated noise variance]\label{remark2}
If Assumption (A4) does not hold, then the bound \eqref{eq:bound1} also depends on the linear operator $L$. Consider for instance ${\rm Cov}\qut{\bmz, \bmz \mid \bmx} \!=\! (1+\beta) {\rm Cov}\qut{\bmn, \bmn \mid \bmx}$ where $\beta\!>\!-1$. Then Equation \eqref{eq:part-lin-mse} becomes 
$\mathcal{J}\qut{R} = \bbE\qut{ \qutn{{R(\bmyh) - \bmx}}^2} - 2 \bbE\qut{\qutan{\bme, {\bmn - \bmz/\alpha}}} + c + 2 \beta \bbE\qut{\qutan{\bmz, L\bmz}}
$, the proof of which is given in Appendix A. The minimization of $J\qut{R}$ therefore favors denoisers with small $\beta \bbE\qut{\qutan{\bmz, L\bmz}}$. Such a property can be used to construct criteria for whether the noise variance is underestimated ($\beta\!<\!0$) or overestimated ($\beta\!>\!0$), e.g., by checking the value of $\bbE\qut{\qutan{\bmz, L\bmz}}$ for the minimizer of $\mathcal{J}$. Examples will be given in Section \ref{sect:denoisingexp}. 
\end{remark}

Next, we take a closer look at the errors of the proposed denoisers. We first introduce the set $\mathcal{R}_\epsilon$ of partially linear denoisers. Second, the distance between the minimizers of $\mathcal{J}$ and $\mathcal{J}_{\rm mse}$ over $\mathcal{R}_\epsilon$ is estimated.
Then we verify the convergence of the proposed denoisers to the best denoisers for corrupted images generated from a single ground truth image. Finally, focusing on a special subset of images that consists of constant patches, we demonstrate the partial linearity of the optimal denoisers with an example.

For $\epsilon>0$, let $\mathcal{R}_\epsilon$ be the set of denoiser satisfying \eqref{eq:part-lin} with variance of $\bme$ less than or equal to $\epsilon^2$, i.e.,
\[
\begin{split}
\mathcal{R}_\epsilon := & \left\{ R \mid \exists g, L, \bme,\ \text{such that}\  R\qut{\bmyh} = g\qut{\bmx} + L \bmnh + \bme, \
\right.\\ 
&\left. L \ \text{is linear, and} \ \bbE\qut{\qutn{\bme}^2} \leq \epsilon^2 \right\}.    
\end{split}
\]

\begin{lemma}[Convexity]\label{lm:convex}
For any $\epsilon \geq 0$, $\mathcal{R}_\epsilon$ is a convex set. 
\end{lemma}
\begin{proof}
If $R_1, R_2 \in \mathcal{R}_\epsilon$, then there exist $g_1, g_2$, linear operators $L_1, L_2$, and residual terms $\bme_1, \bme_2$, such that 
\[
R_i\qut{\bmyh} = g_i\qut{\bmx} + L_i \bmnh + \bme_i, \quad 
\bbE\qut{\qutn{\bme_i}^2} \leq \epsilon^2, \
i \in \{1,2\}.
\]
For any $\lambda_1 \in [0,1]$ and $\lambda_2 := 1 - \lambda_1$, 
\[
\begin{split}
\lambda_1 R_1\qut{\bmyh} + \lambda_2 R_2\qut{\bmyh} = & \qut{\sum_{i=1}^2 \lambda_i g_i }\qut{\bmx}
+ \qut{\sum_{i=1}^2 \lambda_i L_i } \qut{\bmnh} \\ & + \sum_{i=1}^2 \lambda_i \bme_i.     
\end{split}
\]
It is easy to see that $\sum_{i=1}^2 \lambda_i L_i $ is also linear and 
$ \bbE\qut{\|\sum_{i=1}^2 \lambda_i \bme_i\|}^2 \leq \epsilon^2$. Therefore $\lambda R_1+\lambda_2 R_2 \in \mathcal{R}_\epsilon$ which implies that $\mathcal{R}_\epsilon$ is convex.
\end{proof}

We are interested in the best denoisers in $\mathcal{R}_\epsilon$ with as small MSE as possible. 
More precisely, for small $\delta\geq0$, we define the MMSE denoiser, denoted by $R^{\epsilon,\delta}$, as one of the denoisers in $\mathcal{R}_\epsilon$ that satisfies
\begin{equation}\label{eq:rmmse}
\mathcal{J}_{\rm mse}\qut{R^{\epsilon,\delta}} \leq \inf_{R \in \mathcal{R}_\epsilon} \mathcal{J}_{\rm mse}\qut{R} + \delta,
\end{equation}
where $\mathcal{J}_{\rm mse}$ are defined in \eqref{eq:mse-new}. Note that $R^{\epsilon,\delta}$ exists for arbitrarily small positive $\delta$. 
Based on the approximation properties stated in Proposition \ref{prop:2}, we propose a denoiser   
 $\hat{R}^{\epsilon,\delta} \in \mathcal{R}_\epsilon$ satisfying
\begin{equation}\label{eq:rhatmmse} \mathcal{J}\qut{\hat{R}^{\epsilon,\delta}} \leq \inf_{R \in \mathcal{R}_\epsilon} \mathcal{J}\qut{R} + \delta
\end{equation}
as an alternative version of $R^{\epsilon,\delta}$. %
In light of the theoretical bound \eqref{eq:bound1},
the distance between $R^{\epsilon,\delta}$ and $\hat{R}^{\epsilon,\delta}$ has an upper bound described in the following proposition. 

\begin{restatable}{proposition}{propositiond}\label{prop:4}
    Let the denoisers $R^{\epsilon,\delta}$ and $\hat{R}^{\epsilon,\delta}$ be given by \eqref{eq:rmmse} and \eqref{eq:rhatmmse} respectively. If $\bmn$ and $\bmz$ satisfy Assumption (A1) - (A4) and $0 < \delta < 1$, it holds that
    \[
    \begin{split}
    & \bbE\qut{\qutn{R^{\epsilon,\delta}\qut{\bmyh} - \hat{R}^{\epsilon,\delta}\qut{\bmyh}}^2} \\
    \leq & \frac{2 \epsilon}{1-\sqrt{\delta}} \sqrt{ \bbE(\qutn{ \bmn - \bmz/\alpha }^2) } + \frac{\sqrt{\delta}}{1-\sqrt{\delta}},
    \end{split}
    \]
    where the expectation is taken over $\bmx$, $\bmn$ and $\bmz$.
\end{restatable}
The proof is given in Appendix A.
The bound in Proposition \ref{prop:4} converges to $2 \epsilon \sqrt{ \bbE(\qutn{ \bmn - \bmz/\alpha }^2) }$ as $\delta$ tends to $0$, and it converges to $0$ as $\epsilon$ and $\delta$ converge to zero.

Building on Proposition \ref{prop:4}, we further analyze the convergence of the proposed solutions for a special case where all realizations of the corrupted images $\bmyh$ are generated from the same clean image. The next corollary shows that, in this setting, the approximation error is at most of the same order as $\epsilon^2$, regardless of the noise distributions or the noise variance. The proof of this corollary is given in Appendix A. 
\begin{restatable}{corollary}{propositiondb}\label{prop:5}
    Assume that the conditions in Proposition \ref{prop:4} hold. If $\bmx$ follows a delta distribution, i.e., the probability $P\qut{\bmx = x} = 1$ and $P\qut{\bmx \neq x} = 0$ for some image $x$, then 
\begin{equation}\label{eq:coro6-1}
    \begin{split}
    \bbE\qut{\qutn{R^{\epsilon,\delta}\qut{\bmyh} - \hat{R}^{\epsilon,\delta}\qut{\bmyh}}^2} 
    \leq \frac{\epsilon^2 + 2\sqrt{\delta}}{1-\sqrt{\delta}}.
    \end{split}    
\end{equation}
\end{restatable}

Proposition \ref{prop:4} and Corollary \ref{prop:5} illustrate that the minimizers of $\mathcal{J}$ are good estimates of the best partially denoiser when $\epsilon$ is small. In principle, $\bme$ should be much smaller than the noise itself, as this is a necessary condition for the denoiser being robust to different realizations of $\bmnh$. We can provide some more insight into the size of the residual term for high quality denoisers when applied to a special subset of images.
Smooth regions are one of the most important components in natural images, and they often account for a large fraction of the pixels. Here we study the basic case with constant images and pixelwise independent noise. We give an example to demonstrate the scale of the residual term of the minimum MSE denoiser for constant image patches.

\begin{example}[Constant patches]
Assume that $\bmx$ models constant patches of size $21\!\times\!21$ pixels. The pixel values are uniformly distributed in $[0,\lambda]$ where $\lambda>0$ is a constant. The noise is pixelwise independent conditioned on $\bmx$. For pixel $i$, $\bmyh_i \sim  {\rm Pois}\qut{\bmx_i}$ where ${\rm Pois}$ is the Poisson distribution.
The optimal denoised images $R_0\qut{\bmyh}$, where $R_0$ minimizes the cost $\mathcal{J}_{\rm mse}\qut{\cdot}$, are therefore constant images. In this example, given $R_0$ and $\bmx$, the quantities $g\qut{\bmx}$ and $L$ (defined in \eqref{eq:part-lin}) are computed using Remark \ref{remark1}. Consequently, $g\qut{\bmx}$, $L\bmnh$ and $\bme$ are constant images. The values of $R_0\qut{\bmyh}$ are plotted versus $g\qut{\bmx}+L\bmnh$ (for different ground truth values and for different $\lambda$) on the last two rows of Fig. \ref{fig:poisson_linearity}, where the results are computed based on a deep CNN approximation to $R_0$. Each yellow dot in the plot represents a realization of the noise, and the blue line represents $\bme=0$. The plots suggest that $\bme$ has small variances and hence the set $\mathcal{R}_\epsilon$ maintains good approximations to $R_0$ for small $\epsilon$. 
\end{example}
\begin{figure}[ht!]
    \centering
    \setlength{\tabcolsep}{2pt}
    \begin{tabular}{c|c|c}
    \includegraphics[width=0.3\linewidth, trim=0 0 20 0, clip]{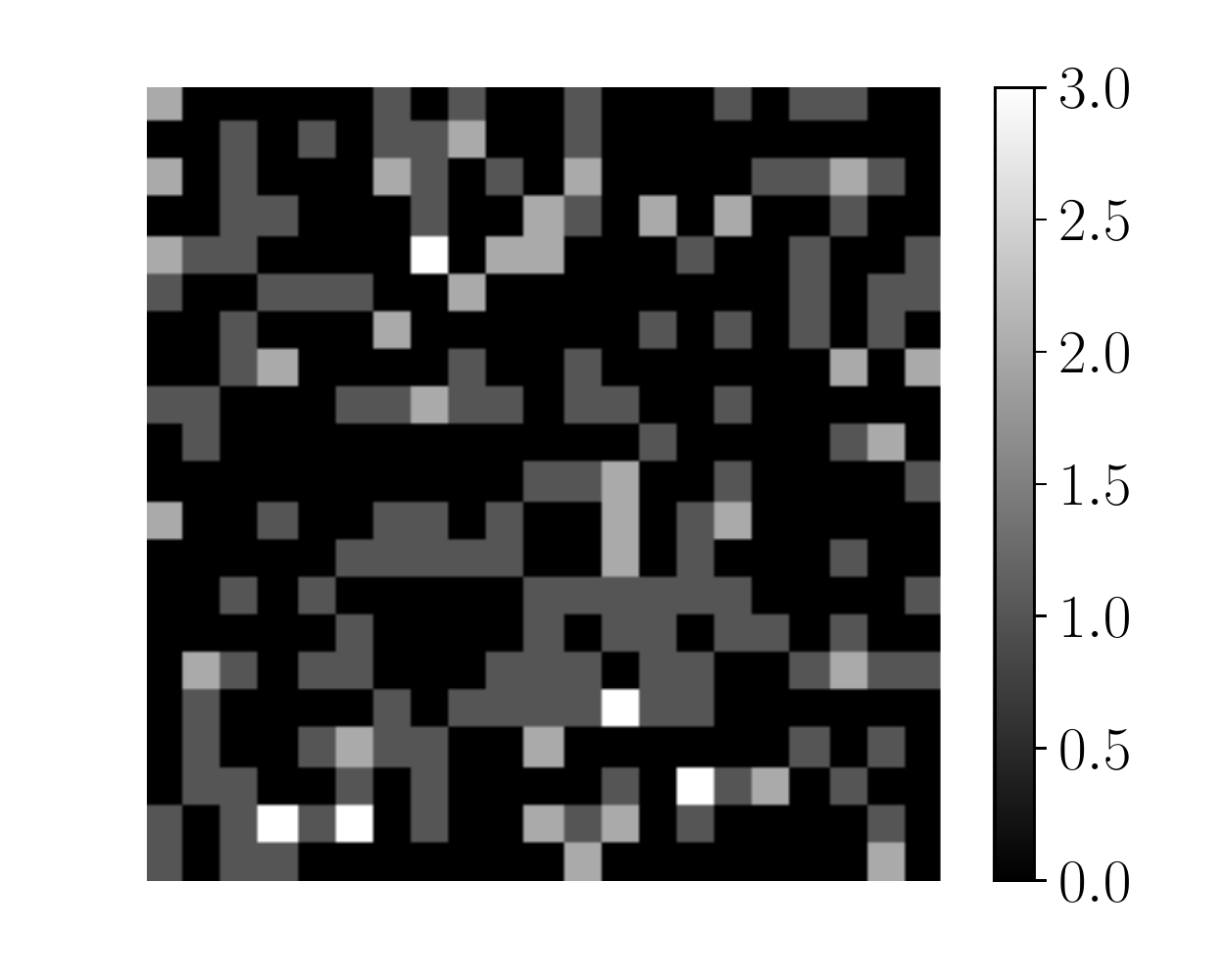}
    &
    \includegraphics[width=0.3\linewidth, trim=0 0 20 0, clip]{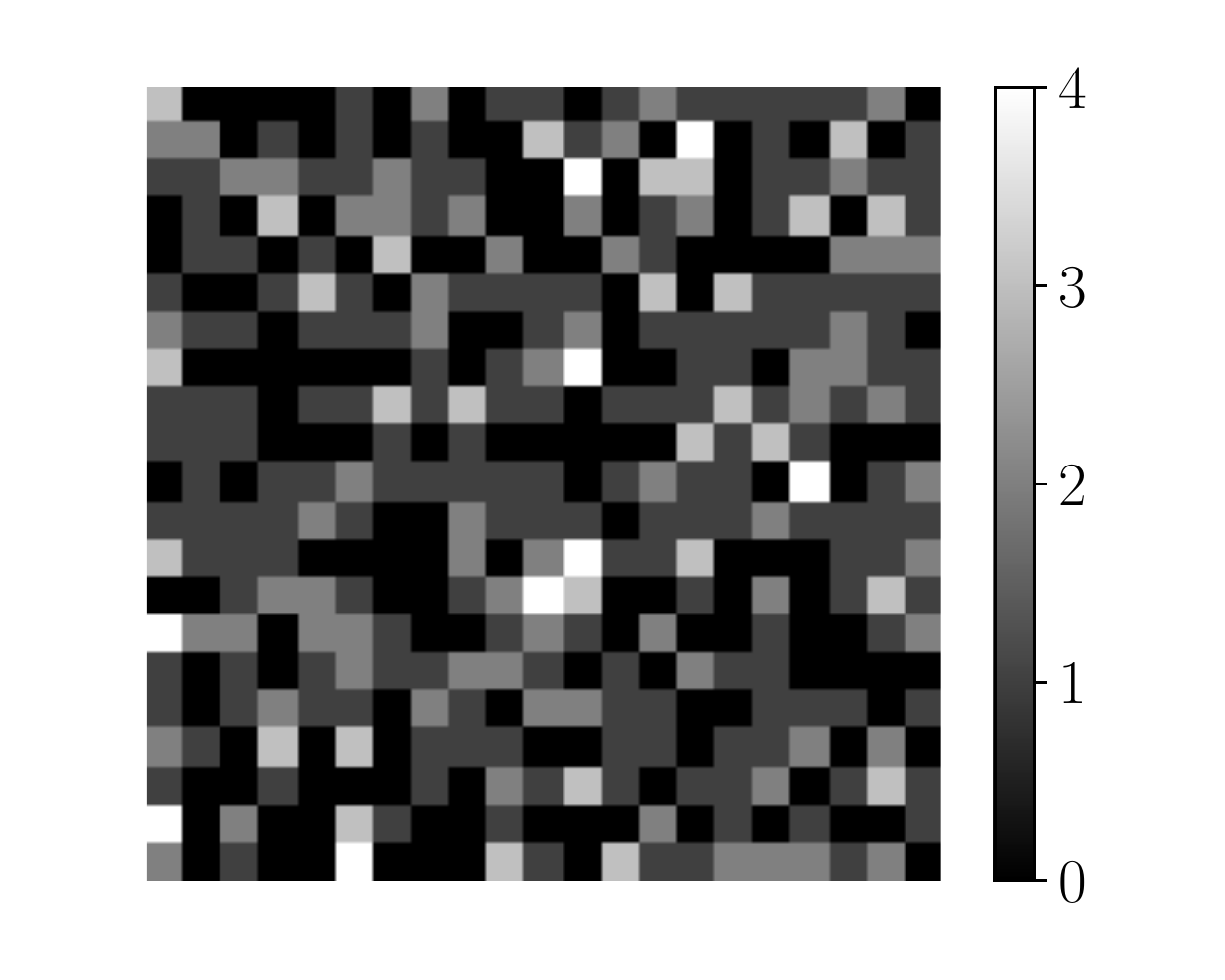}
    &
    \includegraphics[width=0.3\linewidth, trim=0 0 20 0, clip]{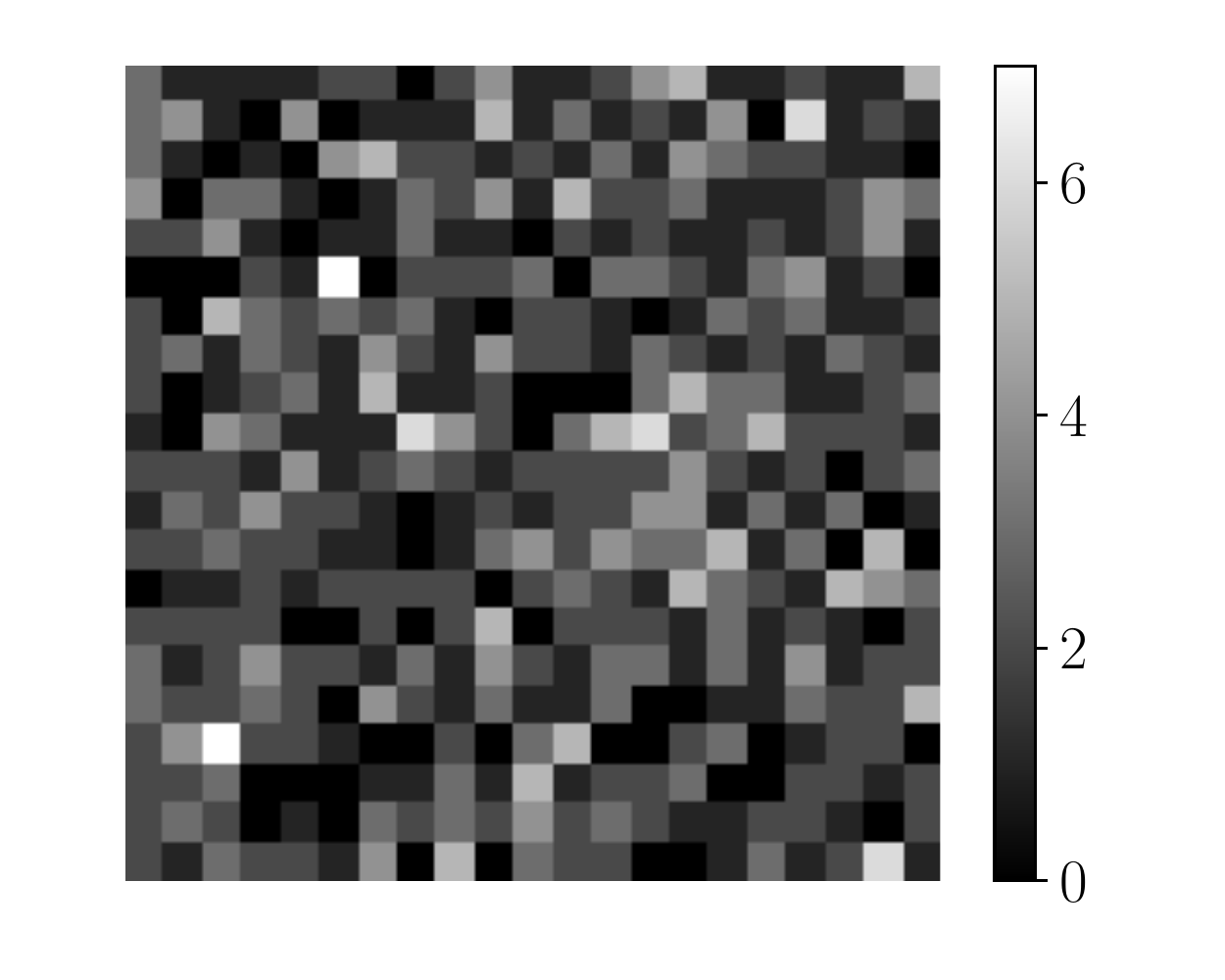}
         \\
    \includegraphics[width=0.3\linewidth, trim=10 0 -20 0, clip]{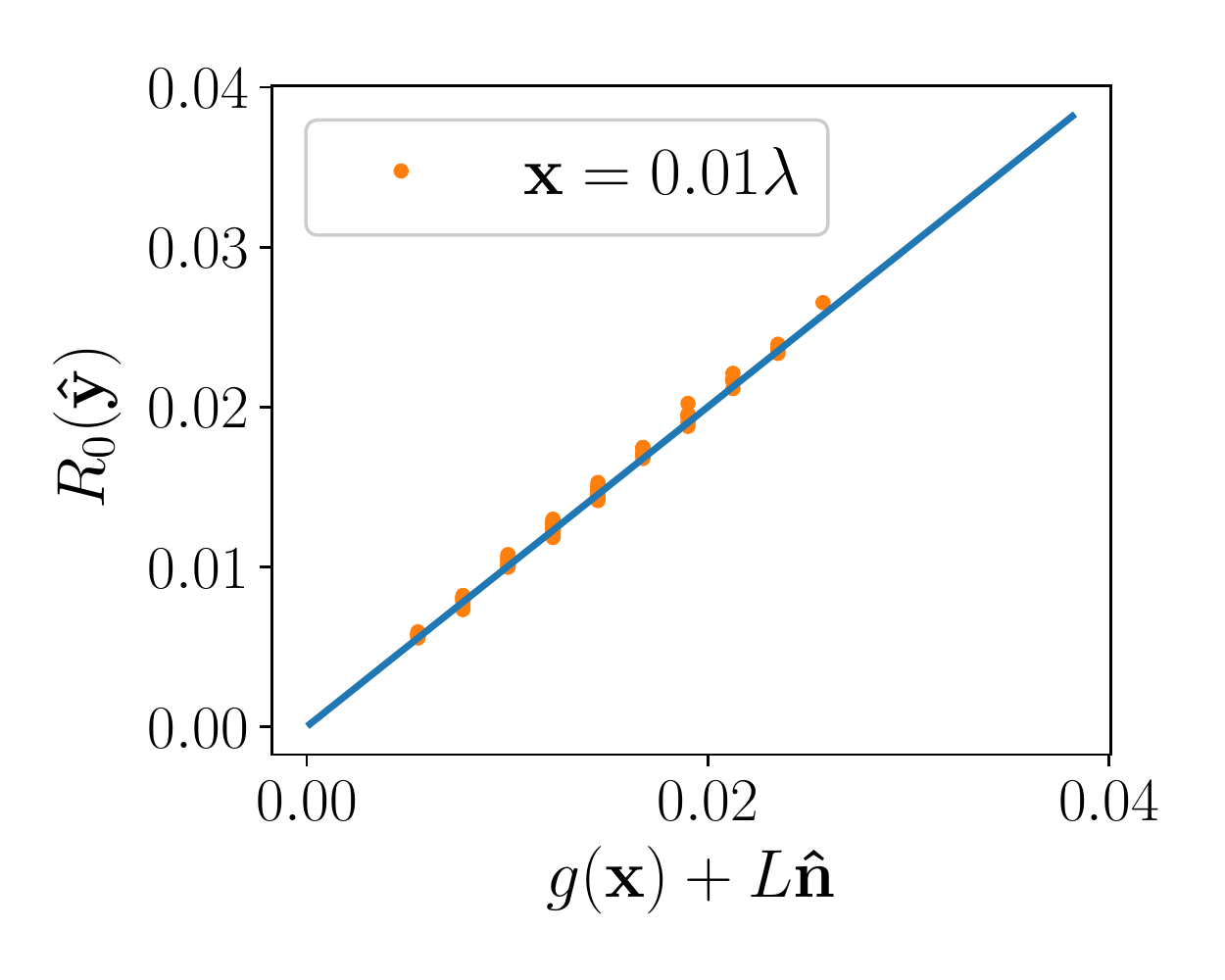} & 
    \includegraphics[width=0.3\linewidth, trim=10 0 -20 0, clip]{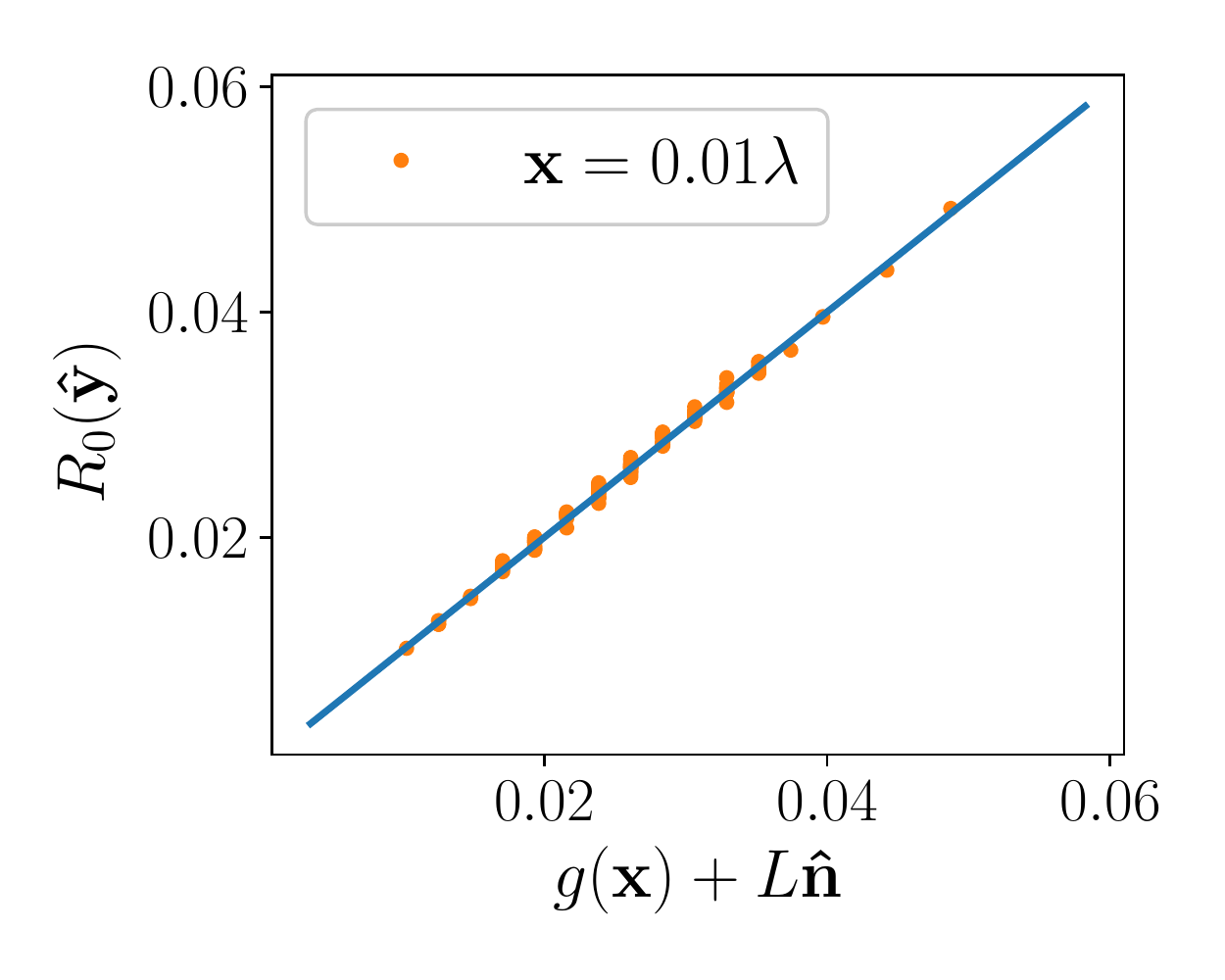} & 
    \includegraphics[width=0.3\linewidth, trim=10 0 -20 0, clip]{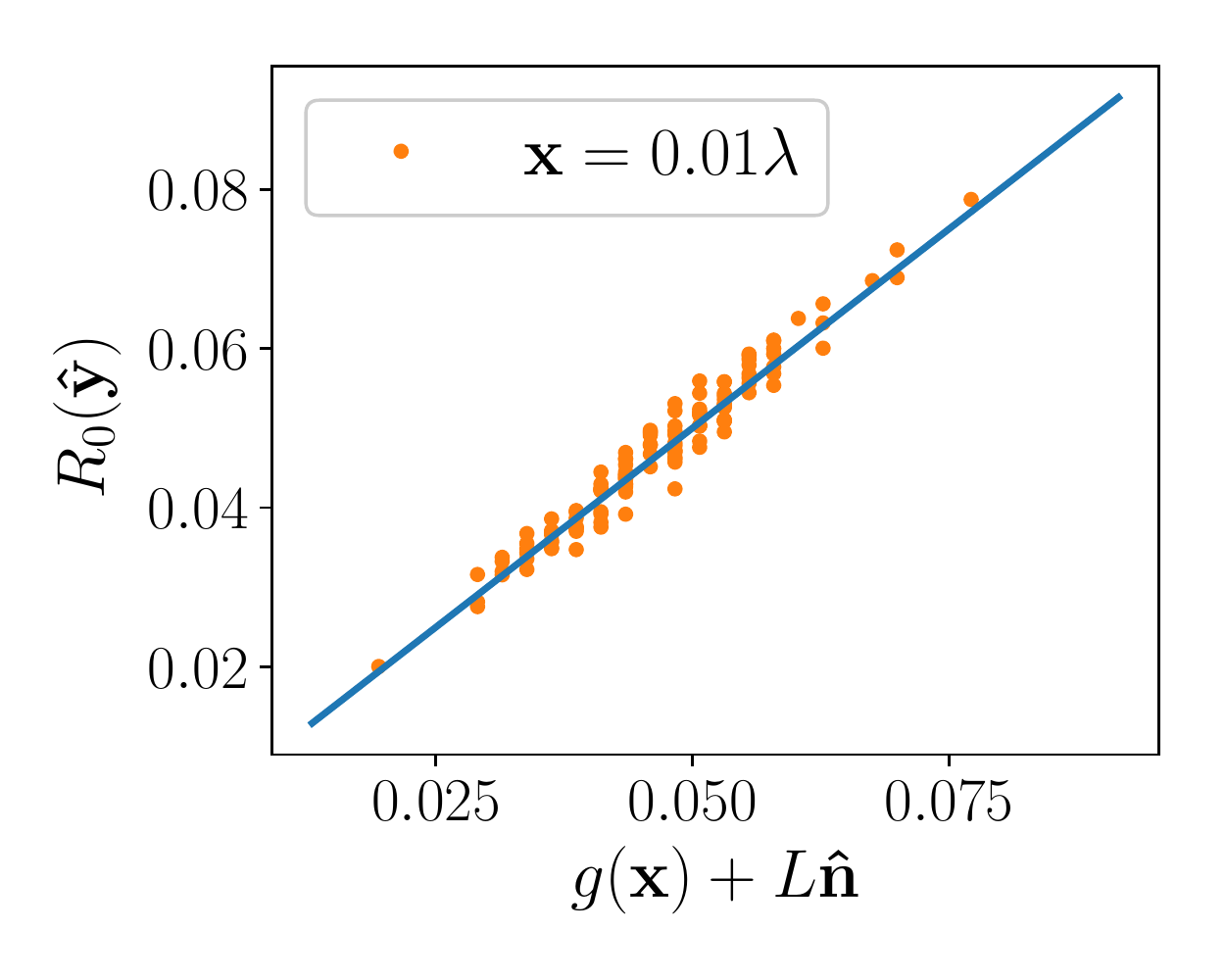} \\
    \includegraphics[width=0.3\linewidth, trim=10 0 -20 0, clip]{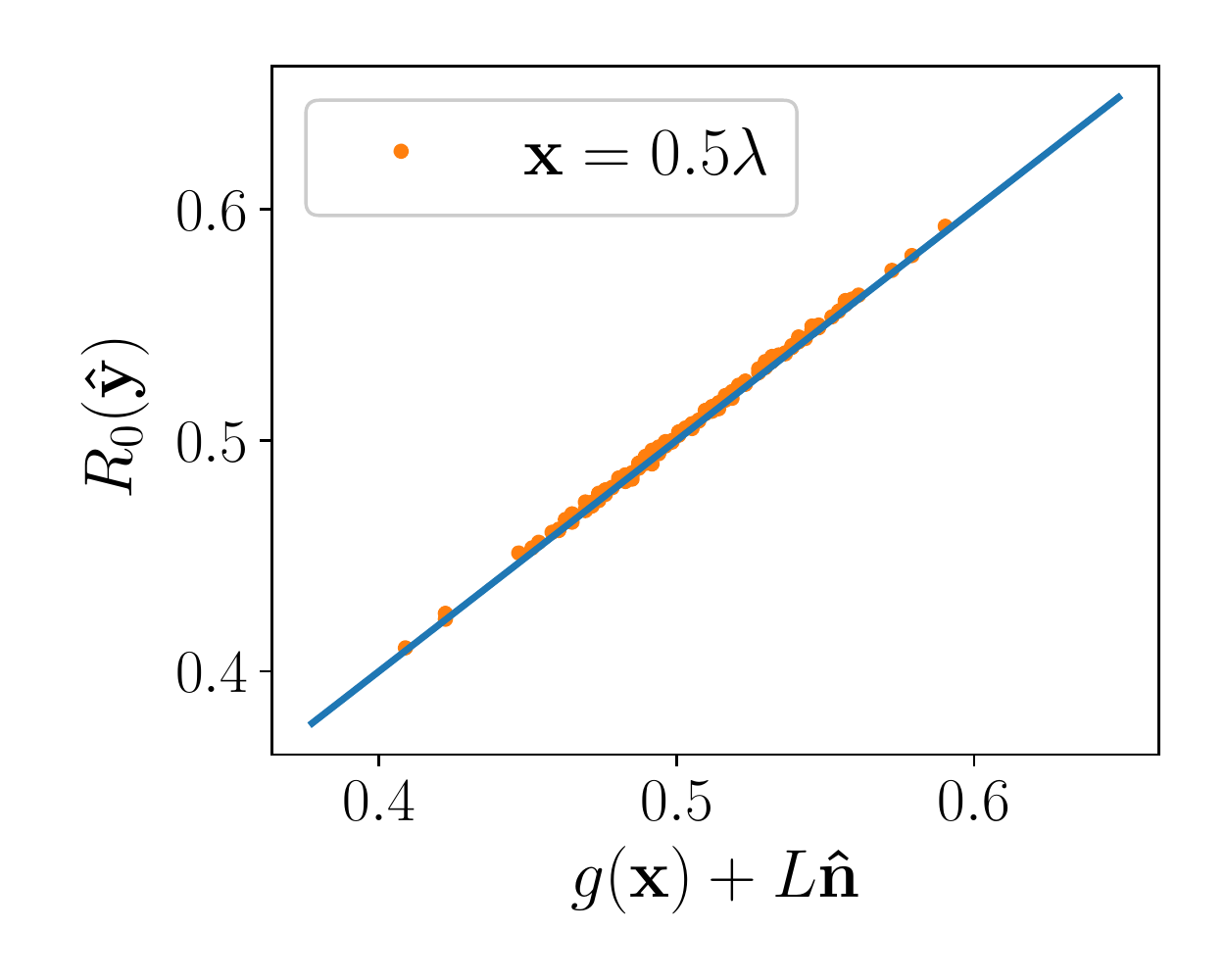} & 
    \includegraphics[width=0.3\linewidth, trim=10 0 -20 0, clip]{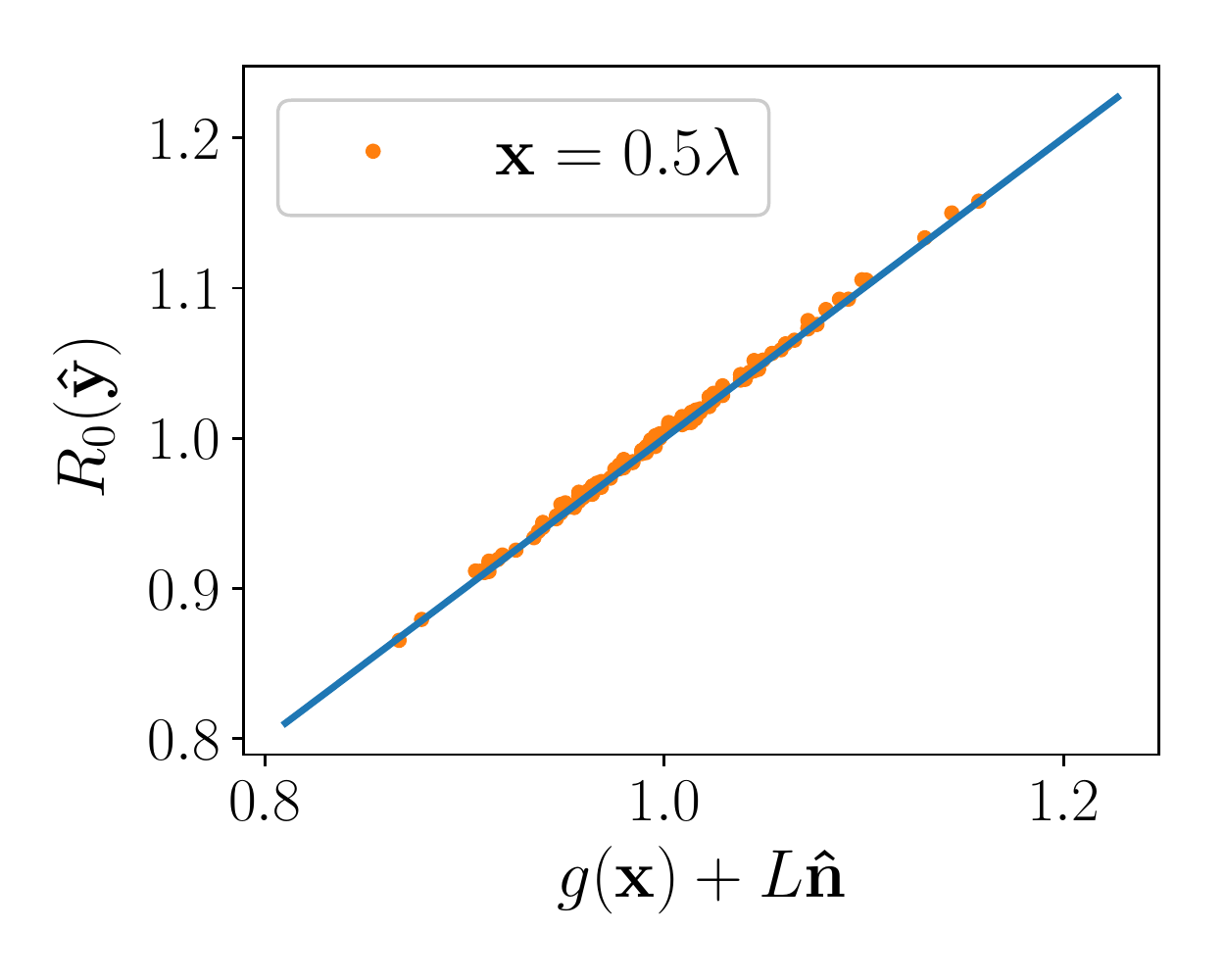} & 
    \includegraphics[width=0.3\linewidth, trim=10 0 -20 0, clip]{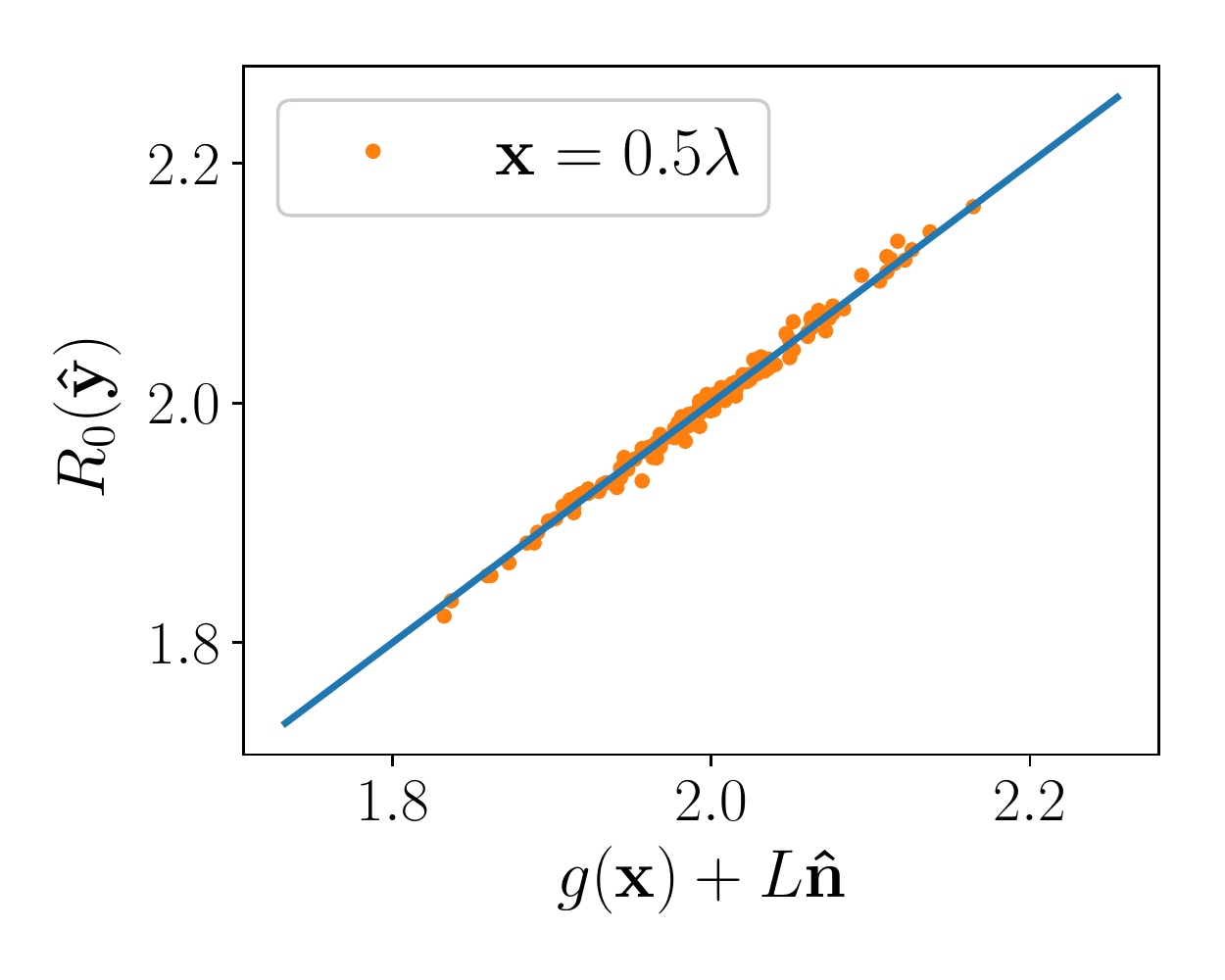} \\
    \footnotesize{(a). $\lambda=1$} & \footnotesize{(b). $\lambda=2$} & \footnotesize{(c). $\lambda=4$}
    \end{tabular}
    \caption{Reconstructions for constant image patches corrupted by Poisson noise with parameter $\lambda=1$, $2$ and $4$ (i.e., the column (a)-(c) respectively).
    Top row: \textit{samples of corrupted patches} (the ground truth is $\bmx = 0.5\lambda$). 
    Second and third rows: \textit{the optimal reconstructions} $R_0(\bmyh)$ plotted against $g(\bmx) + L\bmnh$ (represented by the yellow dots, each of which is generated with a realization of noise $\bmnh$). The blue line represents $\bme\!=\!0$, i.e., $(g(\bmx) + L\bmnh, g(\bmx) + L\bmnh)$.
    }
    \label{fig:poisson_linearity}
\end{figure}

\subsection{Learning a denoiser from noisy samples.}

For learning the denoiser we consider parametrized denoising models $R_\theta$, e.g., deep CNNs, that are parameterized by $\theta$. Suppose that a set $\left\{ y_i \right\}$ of realizations of $\bmy$, e.g., the set of noisy images that one wants to denoise,  is given. %
The noise in these noisy images satisfies Assumption (A1). 
Associated with each $y_i$, we define $z_i$ as a realization of the auxiliary random vector $\bmz$ that satisfies (A2) - (A4). In (A2) and (A4) only the first two moments of $\bmz$, rather than its distribution, are specified. Therefore without loss of generality, one can randomly generate $z_i$ from normal distributions. 
With the auxiliary vector $z_i$, the samples of $\bmyh$, denoted by $\hat{y}_i$, are computed according to Equation \eqref{eq:nhat_yhat}.
Based on the objective function \eqref{eq:newmse}, we minimize the empirical loss function 
\begin{equation}\label{eq:empiricalloss}
    \mathcal{L}\qut{\theta}:=
\sum_i \qutn{R_\theta \qut{\hat{y}_i} - \qut{y_i - z_i/\alpha}}^2
\end{equation}
under the condition that $R_\theta \in \mathcal{R}_\epsilon$. 
This condition can be implemented with the partial linearity constraint described below.

\vspace{0.2cm}
\textit{Partial linearity constraint.} To preserve the partial linearity during training, restriction on the variance of the residual term $\bme$ in \eqref{eq:part-lin} is required. Unfortunately, a direct evaluation of $\bme$ is not feasible due to 1) the fact that the formulation of $\bme$ depends on $g\qut{\bmx}$ and $L$ which are unknown, and 2) the condition that only one noisy observation for each image is given. However, since $g\qut{\bmx}$ and $L$ depend only on the image $\bmx$ and $L$ is linear, we can remove these two terms by perturbing $\bmyh$. One simple way of %
implementing
this is to find $\bm{q}_1$ and $\bm{q}_2$, two perturbed versions of $\bmyh$, satisfying 
$\bmyh = \tau \bm{q}_1 + (1-\tau) \bm{q}_2$ where $\tau \in [0,1]$. Let the residual terms associated with $R\qut{\bm{q}_1}$ and $R\qut{\bm{q}_2}$ be denoted by $\bme_1$ and $\bme_2$ respectively. Then, from Equation \eqref{eq:part-lin} we have 
\begin{equation}\label{eq:linearC}
\begin{split}
    & R\qut{\bmyh} - \tau R\qut{\bm{q}_1} - (1-\tau) R\qut{\bm{q}_2} \\ 
    = &\bme - \tau \bme_1 - (1-\tau) \bme_2.
\end{split}
\end{equation}
By the assumption that $\bme, \bme_1$ and $\bme_2$ are small, based on the representation of the residual terms in \eqref{eq:linearC}, one can therefore penalize $\psi\qut{R\qut{\bmyh} - \tau R\qut{\bm{q}_1} - (1-\tau) R\qut{\bm{q}_2}}$ for some metric $\psi$. %

\subsection{Learning image deblurring without ground truth images}\label{subsect:deblurring_1}
The proposed partially linear denoisers can be generalized to learning image reconstructors for other linear inverse problems such as image deblurring, in particular when the ground truth images are missing. We repeat the 
observation model \eqref{eq:generalproblem} here
\begin{equation}\label{eq:deblur_problem}
    \bmy = A \bmx + \bmn
\end{equation}
where $A$ is a linear forward operator (such as the Radon transform for CT reconstruction, a convolution operator for image deblurring, etc), and $\bmn$ models the zero-mean noise. 
For many inverse imaging problems, the naive reconstruction $A^\dagger \bmy$, is based on a direct inversion of the linear operator $A$ given by the pseudo-inverse operator $A^\dagger$. Due to ill-conditioning of $A$ such inversion is typically unstable and therefore not accurate in the presence of noise. If $A$, however, has full column rank, then the resulting artifacts in the reconstruction could be remedied by denoising $A^\dagger \bmy$. In fact, it is easy to see that $\bmn^\dagger:=A^\dagger \bmy - \bmx = A^\dagger \bmn$ has zero mean. The covariance of $\bmn^\dagger$ can be estimated as long as the covariance of $\bmn$ is known. Therefore a straightforward application of the partially linear denoisers to the noisy reconstruction $A^\dagger \bmy$ leads to an estimate of $\bmx$. 
Various approaches have been proposed for denoising $A^\dagger \bmy$ in a data-driven post-processing setup (see e.g., \cite{jin2017deep, han2018framing}).
These learned post-processing techniques require noisy-clean image pairs as training data.
Another class of methods, in contrast, leverage denoisers to construct regularization for \eqref{eq:deblur_problem}. For instance, the Regularization by Denoising (RED) methods \cite{romano2017little, hong2019acceleration} minimize a variational objective function with a regularization term derived from the denoisers. The Plug-and-Play Prior framework \cite{venkatakrishnan2013plug, sreehari2016plug} is based on variable splitting algorithms for the Maximum a Posteriori (MAP) optimization problem with the proximal mapping associated with the regularization term being replaced by the denoisers. These denoiser-based regularization approaches require predefined denoisers. 
However, our framework does not require any noisy-clean pair or any predefined denoiser. In particular, our proposed partially linear denoiser allows training an estimator from noisy data alone.

Image deblurring is a special case of \eqref{eq:deblur_problem} where $A$ represents a convolutional kernel. In practice, $A$ is governed by different 
imaging factors including motion and camera focus. In the context of the proposed partial linear denoiser, one way to solve the deblurring problem is to train a single model $R_\theta$ that maps the measurement $\bmy$ directly to $\bmx$. Here, knowledge of $A$ is indirectly encoded in $R_\theta$. 
Assuming that the training data contains noisy samples of $\bmyh$ and the blurring kernel $A$ but no ground truth images,
then similar to \eqref{eq:newmse} one can minimize the loss function 
\begin{equation}\label{eq:deblurring_loss}
    \bbE\qut{ \qutn{ A{R_\theta}\qut{\bmyh} - \qut{\bmy - \bmz/\alpha}}^2 },
\end{equation}
and in this case, $A{R_\theta}$ acts as a partially linear denoiser. At test time, the deblurred image can be directly computed as ${R_\theta}\qut{\bmyh}$ without knowing the operator $A$. 
These considerations also draw connections of the partial linear denoiser to the deep image prior approach \cite{ulyanov2018deep}. There, an implicit regularizer is introduced, based on a convolutional neural network $R_\theta$ and a sole data-fitting loss function is minimized in the training. Early stopping is applied to prevent the estimator from over-fitting to the noise. 
With the partial linearity structure, however, we establish a connection between the cost \eqref{eq:deblurring_loss} and the MSE of the noise free measurement $A\bmx$ which aims to get $A{R_\theta}\qut{\bmyh}$ close to $A\bmx$ rather than its noisy versions.

\section{Experiments}\label{sect:exp}
In this section, we report experimental results to demonstrate the efficiency of the proposed approach for different denoising tasks and for deblurring\footnote{The code will be made available at  \textit{\scriptsize{https://github.com/RK621/Unsupervised-Restoration-PLD}}.}. Firstly, we start by comparing the partial linearity of classical denoisers including total variation (TV) denoising \cite{rudin1992nonlinear}, BM3D \cite{dabov2007image}, as well as CNN based denoisers (see Subsection \ref{subsec:pl} for details). Secondly, we evaluate our method on denoising problems with synthetic noise (Subsection \ref{subsect:exp_gaussian}, \ref{subsect:exp_poission}) and analyse its robustness towards errors in the estimate of the noise variance (Subsection \ref{subsect:env}). Thirdly, a numerical study on the stability of our approach with respect to varying noise levels is given and the importance of the partially linear constraint is investigated (Subsection \ref{subsect:exp_stability}). Fourthly, the proposed approach is used to denoise real microscopy images (Subsection \ref{subsect:microscopydenoising}).
Finally, we apply our approach to learn an image deblurring model, using a set of single noisy and blurry observations of the images as training data. The details of the learning methods and the results for the deblurring experiment are presented in Subsection \ref{subsect:blind_deblurring}.

\subsection{Partial linearity of denoisers}\label{subsec:pl}
In this test, we investigate the partial linearity of some existing standard denoising approaches, including the TV approach \cite{rudin1992nonlinear}, BM3D \cite{dabov2007image} and DnCNN \cite{zhang2017beyond}. 
For the convenience of the readers, the formula of the partially linear denoiser %
\eqref{eq:newmse}
is repeated here 
\begin{equation}\label{eq:exp_decomposition}
R\qut{\bmyh} = g\qut{\bmx} + L \bmnh + \bme,
\end{equation}
where $R$ is the underlying denoiser. In particular, the DnCNN is trained by minimizing the standard empirical MSE loss \eqref{eq:mse-new}
with a training set of $400$ images with ground truth \cite{zhang2017beyond}.
For a given image $\bmx$, we compute the decomposition \eqref{eq:exp_decomposition} using Remark \ref{remark1}. Specifically, $g(\bmx)\!=\! \bbE\qut{R\qut{\bmyh}\mid\bmx}$ and $L\!=\!\min_{L_0 \in \mathcal{L}} \bbE\qut{\qutn{R\qut{\hat{\bmy}}-g\qut{\bmx}-L_0\hat{\bmn}}^2}$ where the expectation is taken over $\hat{\bmn}$, and $\mathcal{L}$ denotes the set of linear operators. In this example, we use a standard test image called parrot (cf. top middle of Fig. \ref{fig:diagram}) as the ground truth image, i.e., a realization of $\bmx$. The noise $\bmnh$ is i.i.d. Gaussian noise with zero mean and standard deviation $25$ (corresponding to 256 gray levels). The expectations are evaluated using $20000$ random realizations of the pair $\qut{\bmyh, R\qut{\bmyh}}$.

\begin{table}[ht]
    \centering
    \setlength{\tabcolsep}{4pt}
    \caption{The PSNR for $R$ (first row), the variance of $\bme$ averaged over all pixels (second row), and the PSNR for modified denoisers $\hat{R}:=g(\bmx) + L \bmnh$ (third row) on the image parrot.}
    \begin{tabular}{|c|ccc|}
        \toprule
        Denoiser & TV & BM3D & DnCNN \\
        \hline
        \hline
        PSNR (dB) & 27.62 & 28.87 & 29.47 \\
        \hline
        $\epsilon^2/m$ &  $8.751 \times 10^{-5}$  &$4.875\times 10^{-5}$ &  $4.281 \times 10^{-5}$ \\ 
        \hline
        Modified PSNR (dB)  & 27.82 & 29.00 & 29.60 \\
        \bottomrule
    \end{tabular}
    \label{tab:epsilon}
\end{table}
For the three denoisers, we report the variance of $\bme$ averaged over the image pixels of the parrot image (i.e., $\epsilon^2/m$ where $m$ denotes the number of pixels in $\bme$) in Table \ref{tab:epsilon}.
A comparison of the accuracy of the methods, measured in PSNR, is also given in the table.
All figures reported in the table are averaged over $8000$ independent runs with different realizations of $\bmyh$. 
Based on the table, the DnCNN achieves the best denoising quality, and it outperforms BM3D by around $0.5$ dB and the TV approach by around $1.8$ dB. All three methods have $\epsilon^2/m$ less than $10^{-4}$, and the value for the CNN denoiser is about half of that of the TV method. 

\begin{figure}[ht!]
    \centering
    \setlength{\tabcolsep}{2pt}
    \begin{tabular}{ccc}
    \includegraphics[width=0.3\linewidth, trim=20 20 20 20, clip]{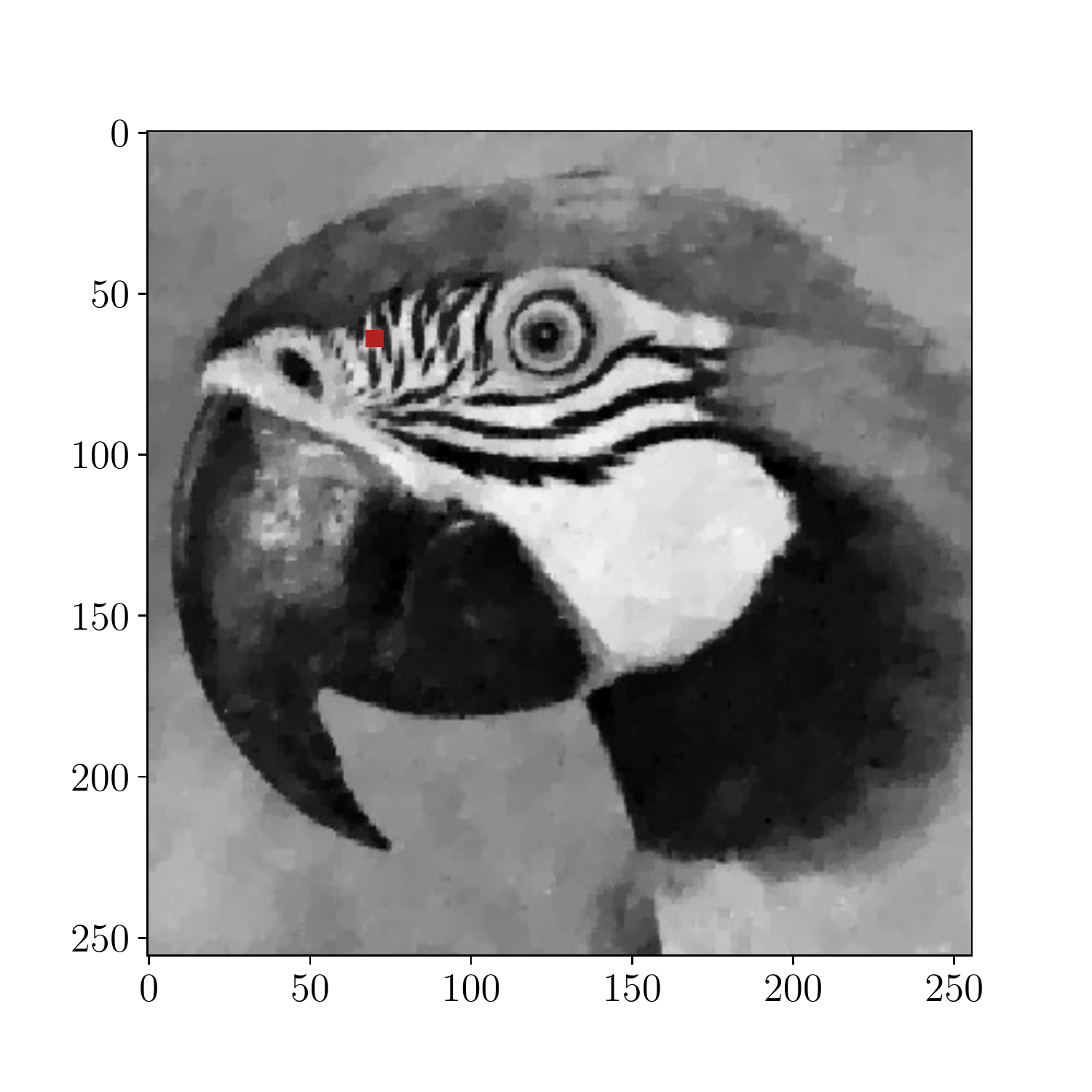}
         &
    \includegraphics[width=0.3\linewidth, trim=20 20 20 20, clip]{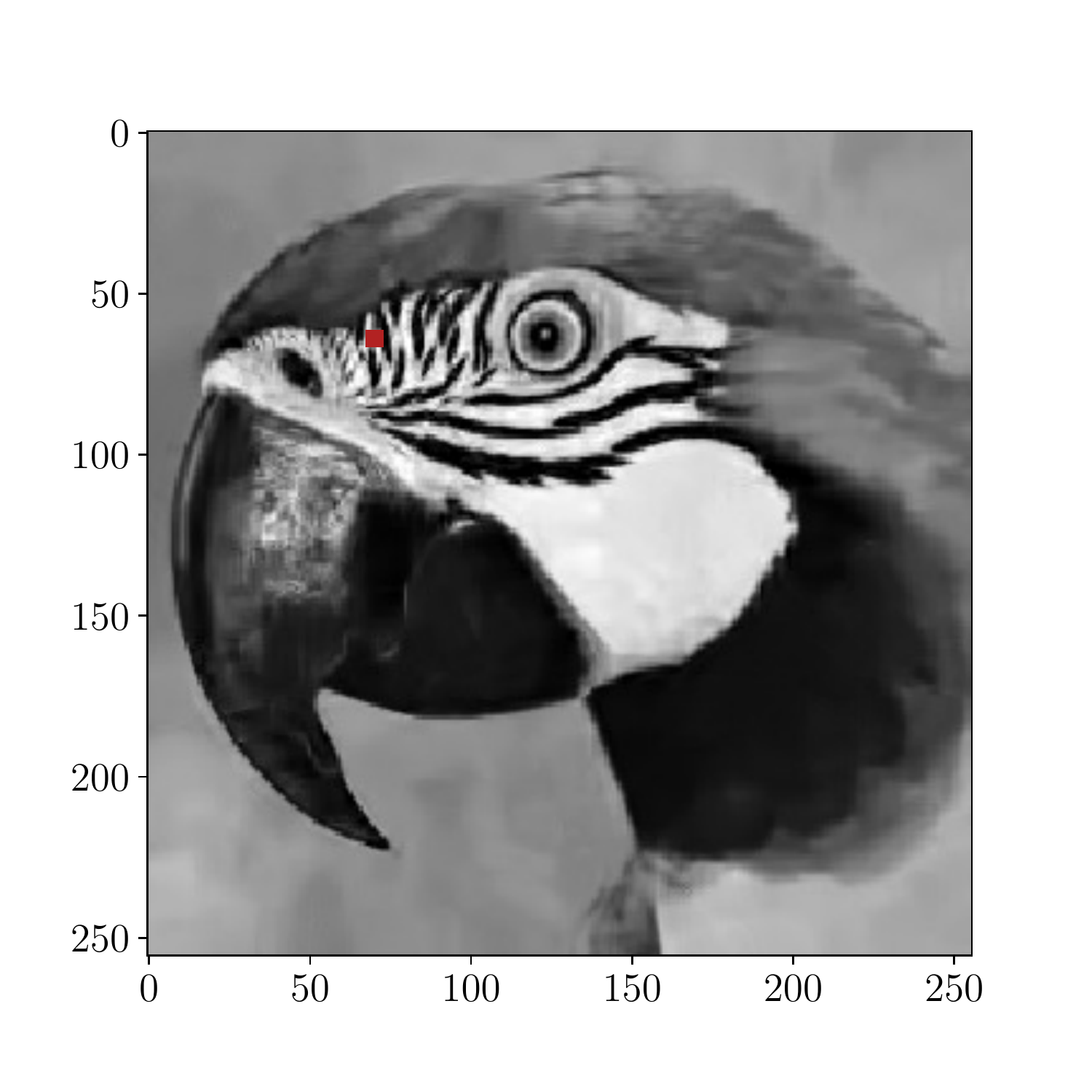}
    &
    \includegraphics[width=0.3\linewidth, trim=20 20 20 20, clip]{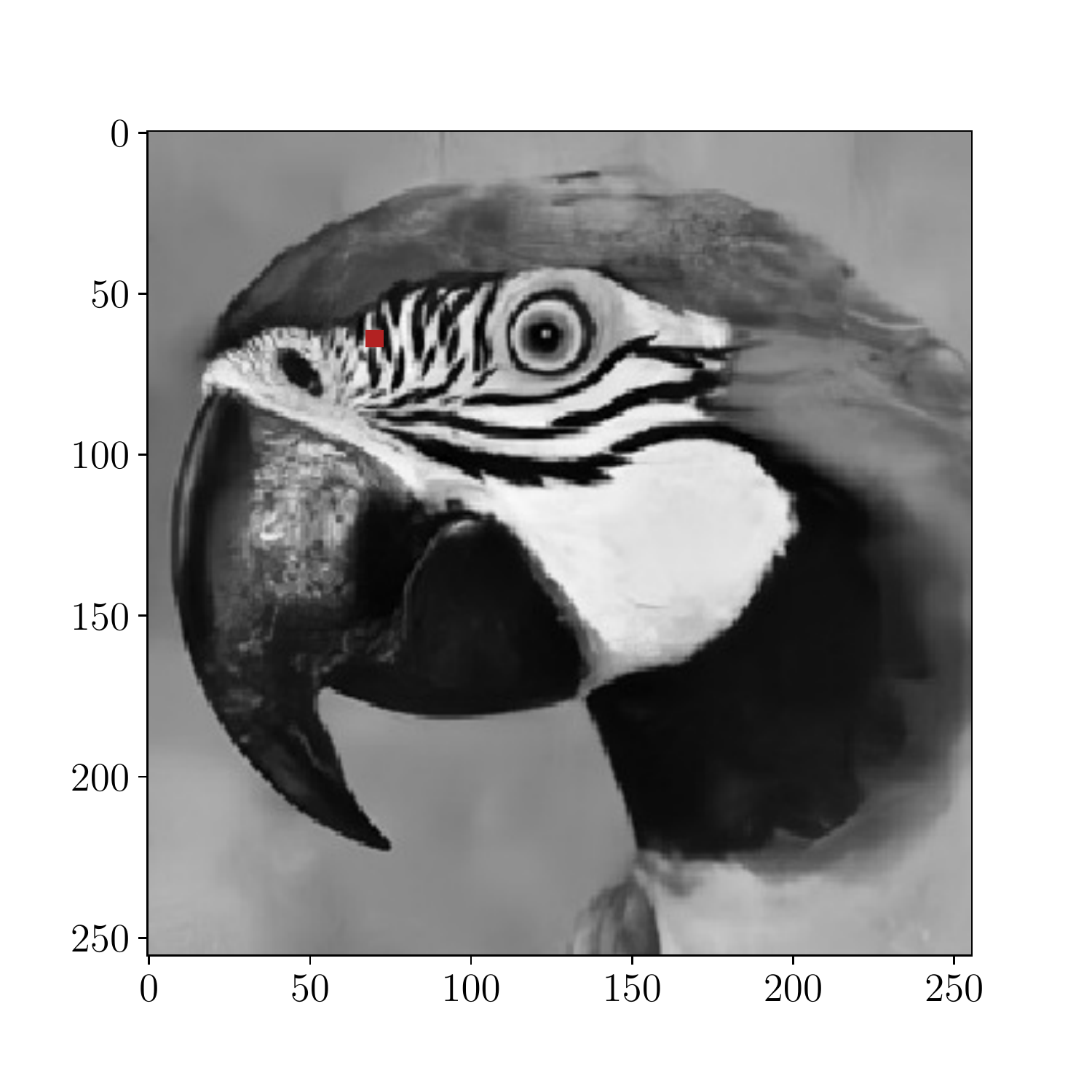}
         \\
    \includegraphics[width=0.3\linewidth, trim=20 20 20 20, clip]{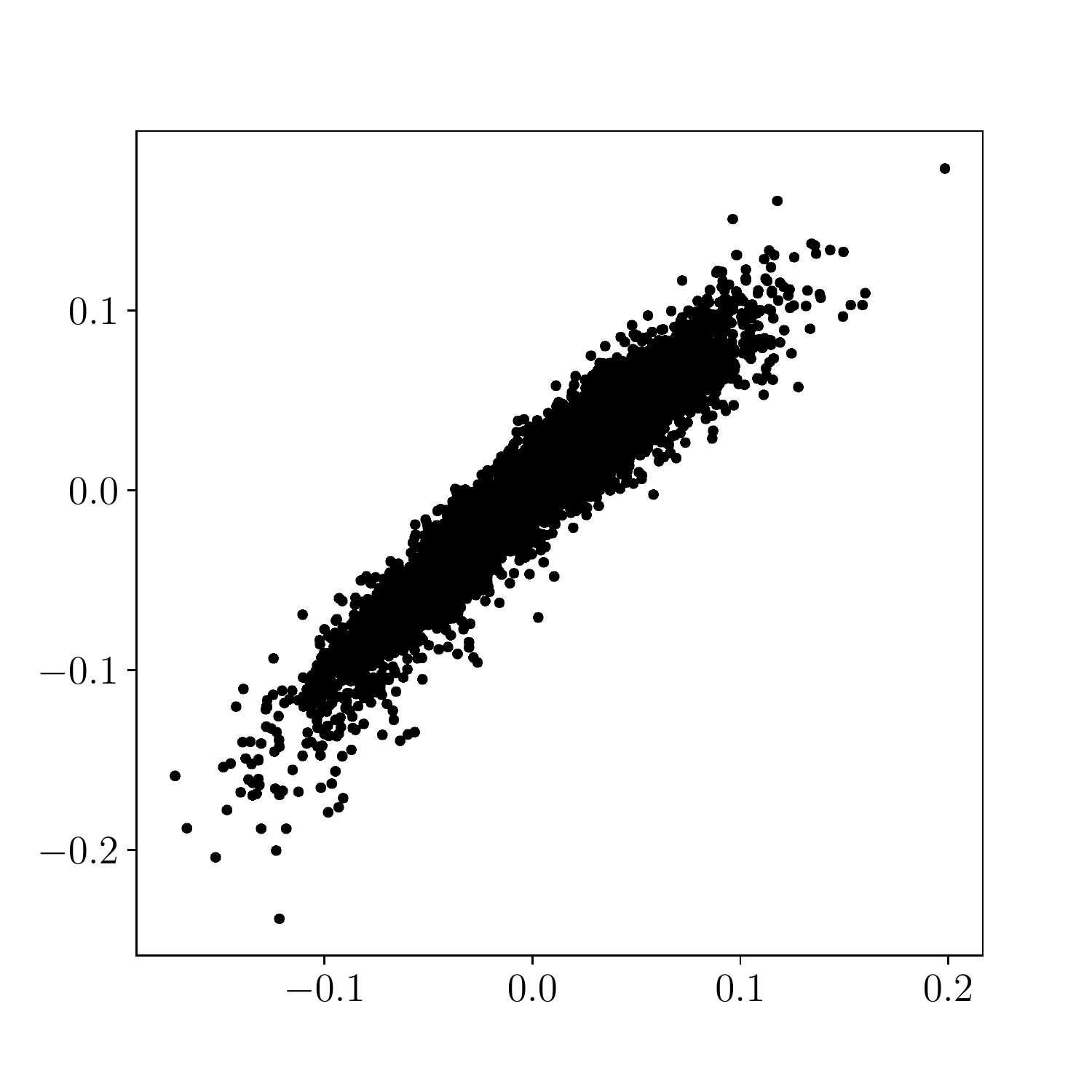} & 
    \includegraphics[width=0.3\linewidth, trim=20 20 20 20, clip]{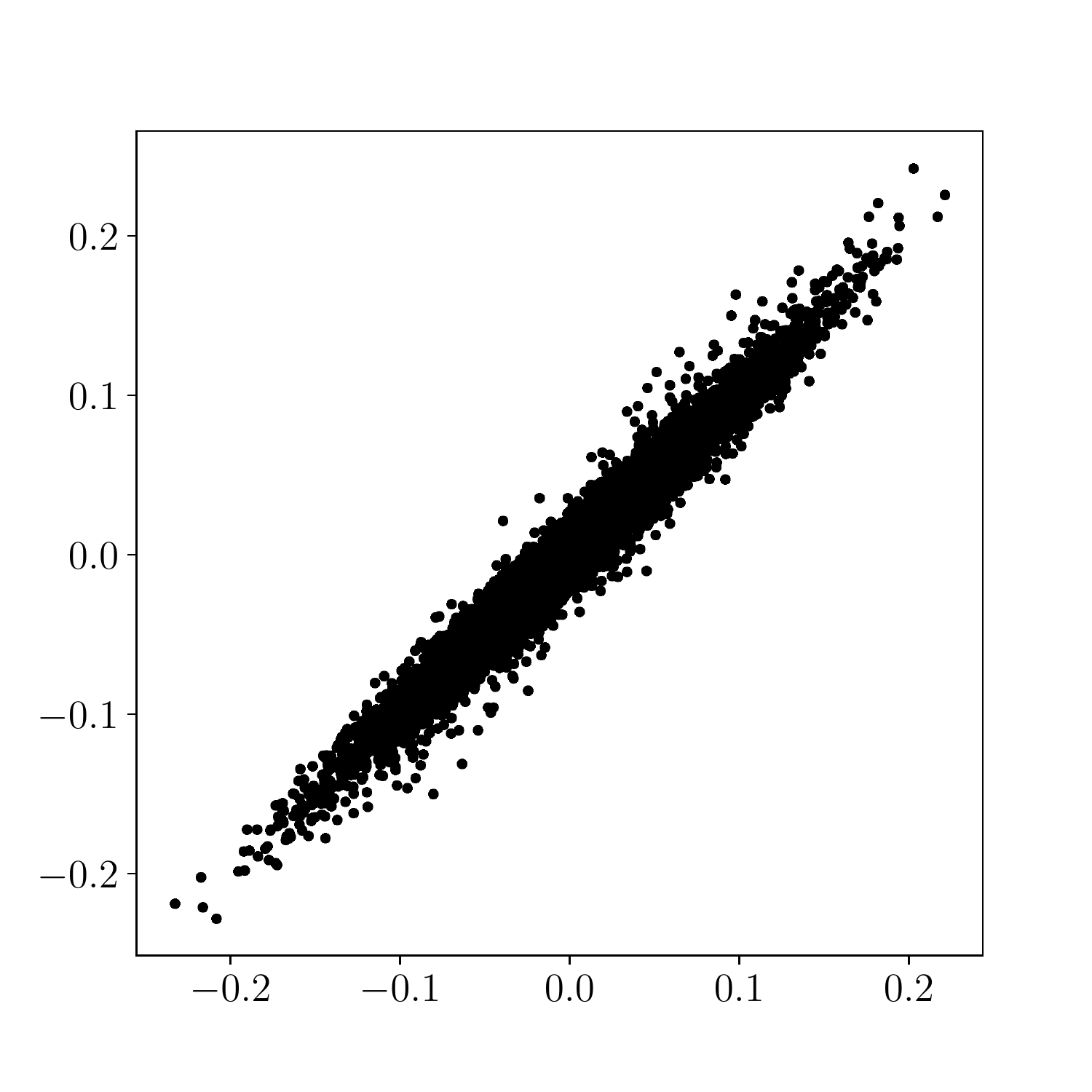} & 
    \includegraphics[width=0.3\linewidth, trim=20 20 20 20, clip]{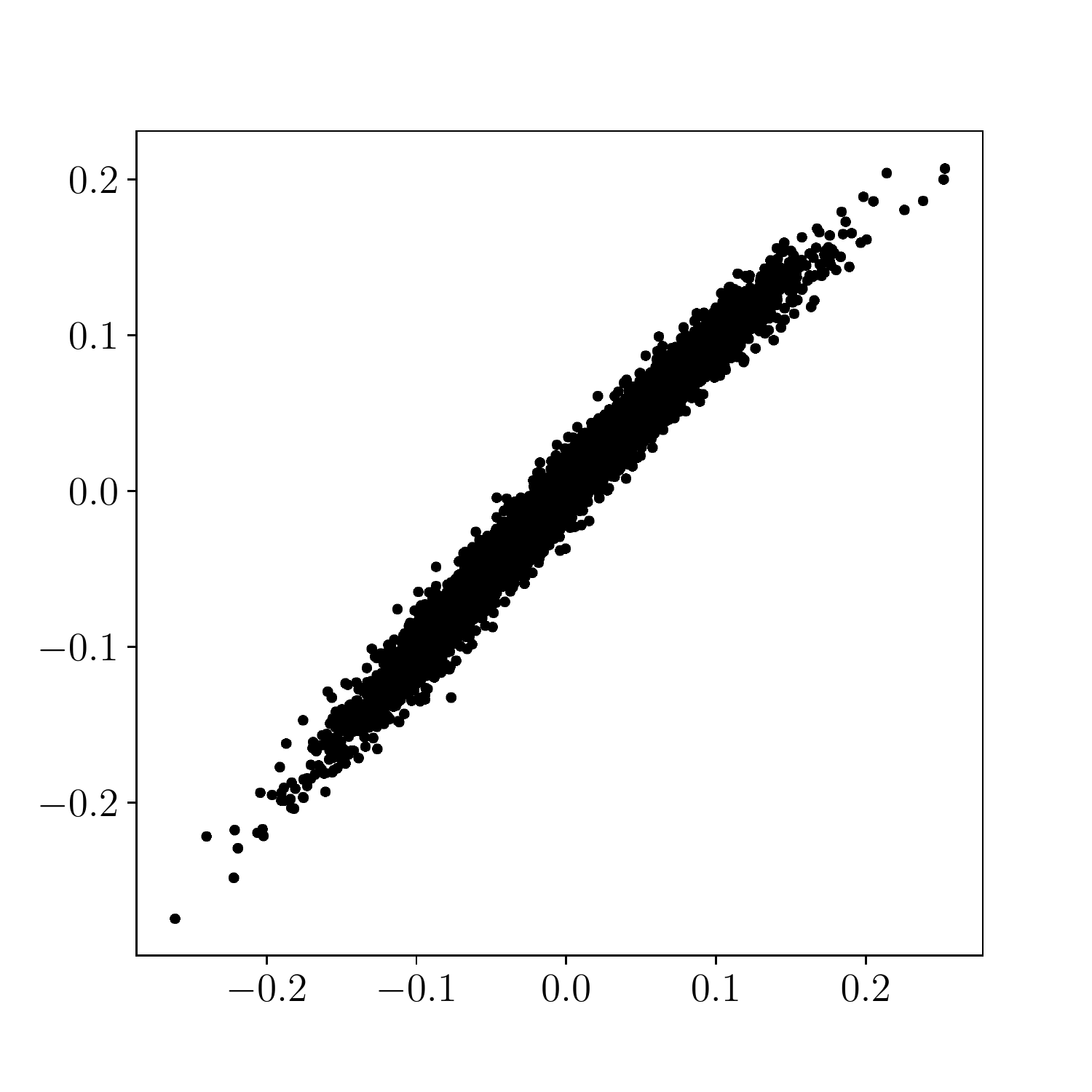} \\
    \footnotesize{(a). TV} & \footnotesize{(b). BM3D} & \footnotesize{(c). DnCNN} 
    \end{tabular}
    \caption{The partial linearity of three denoisers. Top row: denoised images by TV method, BM3D, and DnCNN respectively. Bottom row: the values of $[R\qut{\bmy}-g\qut{\bmx}]_i$ plotted against $ [L \bmnh]_i$  where $i$ is the pixel indicated by the red dot on the top row. Each plot contains $8000$ dots associated with different realizations of $\bmnh$.  
    }
    \label{fig:partial_linearity}
\end{figure}

Given a denoiser $R$, in order to study the partial linearity we consider the pair
\begin{equation}\label{eq:linearity_pairs}
    \qut{ [L \bmnh]_i,   [R\qut{\bmyh} -g\qut{\bmx}]_i }
\end{equation}
where $i$ denotes a fixed pixel (localized by the red dot on the first row of Fig. \ref{fig:partial_linearity}). 
To visualize the linearity for the three denoisers, the pair \eqref{eq:linearity_pairs} is plotted on the second row of Fig. \ref{fig:partial_linearity} under $8000$ realizations of the noise $\bmnh$. Each point in the plot corresponds to one realization. As shown in the figure, all the denoisers have certain degrees of partial linearity at the pixel $i$. For the TV denoiser, the residual $\bme$ has a relatively larger variance (also illustrated by Table \ref{tab:epsilon}) and there are some outliers for big $|[L \bmnh]_i|$. 

\subsection{Denoising experiments}\label{sect:denoisingexp}
We consider two types of synthetic noises (Gaussian noise and Poisson noise), and the performance of our method is compared with recent denoising methods that are not trained on ground truth images, such as SURE \cite{soltanayev2018training}, Noise2Self \cite{batson2019noise2self} and Noisier2Noise \cite{moran2020noisier2noise}, along with the classic BM3D approach \cite{dabov2007image}.
Throughout the denoising experiments, we use the same network architecture DnCNN \cite{zhang2017beyond} for the fully supervised baseline (we will call it DnCNN in the subsequent), SURE, Noise2Self, Noisier2Noise and our approaches. 
The experimental setups for the five methods are the same, except that the clean images are used to train the DnCNN  while they  are unseen by the latter four methods in the learning phase. All models are trained on a benchmark denoising dataset \cite{zhang2017beyond} consisting of $400$ training images of size $180\!\times\!180$. In the training phase, we feed the CNNs with image patches of size $40\!\times\!40$ and set a batch size of $128$. Augmentations such as random flipping and random cropping are applied to the patches. 
In the inference phase, the inputs to the networks are the whole noisy images. In particular, in our approach we do not include the auxiliary vector $\bmz$ at this stage, and the denoised images are the outputs of the network without any post-processing. 
We evaluate the denoising quality on two different test image sets, namely the BSD68 (containing $68$ images) \cite{martin2001database} and the $12$ wildly used images (Set12) \cite{dabov2007image}.

\textit{Training.} To train our denoising model, we minimize the loss function \eqref{eq:empiricalloss} using the Adam optimizer \cite{kingma2014adam}. The minimization process consists of two stages.
In the first stage, we fix $\alpha\!=\!1$ and minimize the loss function \eqref{eq:empiricalloss} without constraints. In the second stage, we randomly choose $\alpha \!\in\! [0.1,0.5]$ for each of the noisy samples, and additionally, in order to control the variance of $\bme$ defined in \eqref{eq:part-lin}, we impose a partially linear constraint \eqref{eq:linearC}. The implementation details of the latter will be given in the next paragraph. Both stages contain $2\!\times\!10^5$ optimization steps with an initial learning rate $0.001$, which then drops to $0.0001$ and $0.00005$ at the $6\!\times\! 10^4$th step and $1.2 \!\times\! 10^5$th step, respectively. 
Throughout the experiments, $z_i$, i.e., the samples of the auxiliary random vector $\bmz$ in \eqref{eq:empiricalloss}, are generated from Gaussian distributions.

\textit{Partially linear constraint.} We implement the partially linear constraint by perturbing the noisy images $\hat{y}_i$ (i.e., samples of the noisy image $\bmyh$) and by following the formula \eqref{eq:linearC}. 
Specifically, for each  $\hat{y}_i$ let $\beta_i^\qut{1}$ and $\beta_i^\qut{2}$ be random numbers uniformly distributed in $[1,1.5]$, and let $q_i$ be a perturbation vector randomly generated from the same distribution as $z_i$. With $q_i$ we construct two perturbed versions of $\hat{y}_i$ as $q^{(1)}_i \!:=\! \hat{y}_i - \beta_i^\qut{1} q_i$ and $q^{(2)}_i \!:=\! \hat{y}_i + \beta_i^\qut{2} q_i$ respectively. Then we have the linear relationship
\[
\hat{y}_i = \tau_1 q^{(1)}_i + \tau_2 q^{(2)}_i
\]
where $\tau_1 := \beta_i^\qut{2} / \qut{\beta_i^\qut{1} + \beta_i^\qut{2}}$ and $\tau_2 := 1 - \tau_1$. 
In practice, we let $q_i$ be independent of $z_i$. 
Since the noise levels of $q^{(1)}_i$ and $q^{(2)}_i$ depend on the pixel values of $q_i$, we modify $q_i$ to be sparse to avoid raising the noise level too much. To do this, we randomly select $1/25$ of the pixels, the pairwise distances of which are at least $4$ pixels. The remaining pixels of $q_i$ are set to zero, i.e., no perturbations are applied to these pixels. To avoid having outliers at some individual pixels caused by the perturbations, we also clip $q_i$ such that the pixel values $q_i^\qut{1}$ and $q_i^\qut{2}$ fit in the range $\left[1.2 a - 0.2 b, 1.2b - 0.2a\right]$ where $a\!:=\!\min_j [\hat{y}_i]_j$ and $b\!:=\!\max_j [\hat{y}_i]_j$. Having $q_i^\qut{1}$ and $q_i^\qut{2}$, based on Property \eqref{eq:linearC} we add the following penalty term to the loss function
\begin{equation}\label{eq:loss_Lc}
\begin{split}
\mathcal{L}_c\qut{\theta} := & \sum_{i} \left\| M \left( \vphantom{\frac{ \beta_{3-k} }{\beta_1+\beta_2}} R_\theta\qut{\hat{y}_i}  - \sum_{k=1}^2
\tau_k R_\theta\qut{q_{i}^{(k)}}\right)\right\|^2, 
\end{split}
\end{equation}
where $\theta$ are the network parameters, $M$ is a diagonal matrix.
The diagonal entries of $M$ are set as 
\[
M_{jj} := 
1/\left[\left| q^{(1)}_i - q^{(2)}_i  \right|
+ 0.1 \sigma\right]_j,
\]
if $[q_i]_j \!\neq\! 0$, otherwise $M_{jj} \!:=\! 0$, where $\sigma \!>\! 0$ is the square root of the largest pixel-wise variance of the noise $\bmn$. Note that having $M$ in the loss \eqref{eq:loss_Lc} means that we penalize the nonlinearity at the perturbed pixels only. 
The term $0.1 \sigma$ prevents division by very small numbers.  In summary, the loss function is $\mathcal{L} + \gamma \mathcal{L}_c$ where $\mathcal{L}$ is defined in \eqref{eq:empiricalloss}, and $\gamma$ is a hyperparameter which can be tuned based on the quality of the denoised images. 

\subsubsection{Gaussian noise}\label{subsect:exp_gaussian}
In this experiment, we test the denoisers for restoring images corrupted by Gaussian white noise. The training sets for our unsupervised approach are the noisy images.
We consider two different levels of noise with standard deviation $\sigma \!=\! 25$ and $\sigma \!=\! 50$ 
(corresponding to 256 gray levels),
respectively. 
Associated with the two noise levels, two denoisers are trained using the proposed method, and the ground truth images are unseen during the training phases. In both cases, the parameter $\gamma$ for the partially linear constraint term \eqref{eq:loss_Lc} is set to $4$. 

The test results for BSD68 \cite{martin2001database} are reported in Table \ref{tab:bsd68_gau}, where we call our method DPLD (\textit{deep partially linear denoiser}). 
The denoising quality is measured by the peak signal-to-noise ratio (PSNR) and the structural similarity (SSIM) index. We compare our denoiser with BM3D \cite{dabov2007image}, the self-supervised method Noise2Self \cite{batson2019noise2self}, SURE \cite{soltanayev2018training}, Noisier2Noise \cite{moran2020noisier2noise} as well as the fully-supervised DnCNN \cite{zhang2017beyond}.  It is worth mentioning that the last denoiser DnCNN, in contrast to the other five, requires the noisy-clean image pairs for training. 

\begin{table}[ht]
    \centering
    \caption{Denoising quality, measured by PSNR (dB) and SSIM, for BSD68 \cite{martin2001database} corrupted by Gaussian Noise.}
        \begin{tabular}{|c|cc||cc|}
        \hline
        Noise Level & \multicolumn{2}{c||}{$\sigma=25$} & \multicolumn{2}{c|}{$\sigma=50$} \\
        \hline
        Measure & PSNR & SSIM & PSNR & SSIM \\
        \hline
        
        BM3D\cite{dabov2007image} & $28.58$ & $0.8861$ & $25.66$ & $0.8041$ \\
        Noise2Self\cite{batson2019noise2self} & $27.48$ & $0.8588$ & $25.15$ & $0.7818$ \\
       SURE\cite{soltanayev2018training} &  $28.99$ & $0.8961$ & $25.88$ & $0.8118$ %
       \\
       Noisier2Noise\cite{moran2020noisier2noise} & $28.96$ & $0.8951$ & $25.96$ & $0.8125$
        \\
        DPLD & $\bm{29.08}$ & $\bm{0.8961}$ & $\bm{26.13}$ & $\bm{0.8196}$ \\
        \hline
        DnCNN\cite{zhang2017beyond} & $29.22$ & $0.9017$ & $26.24$ & $0.8265$ \\
        \hline
        \end{tabular}
    \label{tab:bsd68_gau}
\end{table}

As shown in Table \ref{tab:bsd68_gau}, the fully supervised denoiser DnCNN achieves the best accuracy. This is not surprising as it learns from the ground truth images which are not provided for the other ones. Our method is the best among the denoisers trained without the ground truth. It outperforms the Noise2Self, SURE and Noisier2Noise by $1.6$ dB, $0.09$ dB and $0.12$ dB respectively for noise level $\sigma \!=\! 25$, and outperforms them by $0.98$ dB, $0.25$ dB and $0.17$ dB respectively in the $\sigma \!=\! 50$ case. Compared to the DnCNN, the PSNR values of DPLD are lower by $0.14$ dB and $0.11$ dB for noise levels $\sigma \!=\! 25$ and $\sigma \!=\! 50$ respectively.

\begin{table*}[ht]
\centering
\caption{Denoising quality (in dB) for the $12$ wildly used image \cite{dabov2007image} and Gaussian noise}\label{tab:set12_gau}
\makebox[\textwidth][c]{
\begin{tabular}{|c|c|c|cccccccccccc|}
\hline
& {Denoiser} & 
{Average} &{C. man} &  {House} &  {Peppers} &  {Starfish} &  {Monarch}  &  {Airplane}  &  {Parrot}  &  {Lena} &  {Barbara} &  {Boat} &  {Man} &  {Couple}  \\
\hline
\hline
\multirow{4}{*}[-1ex]{\rotatebox{90}{$\sigma  = 25$}} 
& BM3D & $29.97$ & $29.39$ & $32.98$ & $30.18$ & $28.61$ & $29.3$ & $28.43$ & $28.83$ & $32.05$ & $\bm{30.61}$ & $29.86$ & $29.64$ & $29.72$ \\
& Noise2Self & $28.81$ & $27.95$ & $32.22$ & $29.51$ & $28.12$ & $28.64$ & $26.7$ & $27.58$ & $31.52$ & $26.41$ & $29.06$ & $29.04$ & $28.94$  \\ 
& SURE & $30.12$ & $29.75$ & $32.65$ & $30.51$ & $29.15$ & $30.08$ & $\bm{28.97}$ & $\bm{29.32}$ & $32.08$ & $29.28$ & $29.93$ & $29.94$ & $29.82$ %
\\
& Noisier2Noise & $30.12$ & $29.69$ & $32.74$ & $30.52$ & $29.13$ & $30.10$ & $28.95$ & $\bm{29.32}$ & $32.10$ & $29.22$ & $29.92$ & $29.94$ & $29.79$ 
\\
& DPLD & $\bm{30.28}$ & $\bm{29.84}$ & $\bm{33.04}$ & $\bm{30.69}$ & $2\bm{9.26}$ & $\bm{30.21}$ & $\bm{28.97}$ & ${29.30}$ & $\bm{32.33}$ & $29.66$ & $\bm{30.10}$ & $\bm{30.02}$ & $\bm{29.97}$  \\ 
\cline{2-15}
& DnCNN & $30.44$ & $30.08$ & $33.13$ & $30.8$ & $29.44$ & $30.39$ & $29.12$ & $29.48$ & $32.43$ & $29.96$ & $30.21$ & $30.12$ & $30.12$  \\ 
\hline
\hline
\multirow{4}{*}[-1ex]{\rotatebox{90}{$\sigma  =50$}} 
& BM3D & $26.71$ & $26.36$ & $29.75$ & $26.69$ & $24.99$ & $25.9$ & $25.22$ & $25.74$ & $28.84$ & \bm{$26.98$} & $26.76$ & $26.84$ & $26.49$ \\
& Noise2Self & $26.14$ & $25.60$ & $29.20$ & $26.36$ & $24.67$ & $25.68$ & $24.53$ & $25.29$ & $28.55$ & $24.56$ & $26.45$ & $26.64$ & $26.10$  \\ 
& SURE & $26.62$ & $26.45$ & $29.25$ & $26.74$ & $25.20 $ & $26.23$ & $25.50$ & $26.10$ & $28.70$ & $25.01$ & $26.80$ & $26.99$ & $26.43$ %
\\
& Noisier2Noise & $26.79$ & $26.61$ & $29.65$ & $26.98$ & $25.25$ & $26.49$ & $25.62$ & $26.21$ & $28.94$ & $25.13$ & $26.92$ & $27.08$ & $26.56$ 
\\
& DPLD & $\bm{27.05}$ & \bm{$26.89$} & \bm{$30.01$} & \bm{$27.22$} & \bm{$25.45$} & \bm{$26.72$} & \bm{$25.76$} & \bm{$26.36$} & \bm{$29.24$} & $25.70$ & \bm{$27.13$} & \bm{$27.23$} & \bm{$26.83$}  \\ 
\cline{2-15}
& DnCNN & $27.19$ & $27.03$ & $30.10$ & $27.36$ & $25.55$ & $26.87$ & $25.89$ & $26.45$ & $29.29$ & $26.26$ & $27.22$ & $27.27$ & $26.94$  \\ 
\hline
\hline
\end{tabular}}
\end{table*}

The comparison of the denoisers on the $12$ wildly used images \cite{dabov2007image} is given in Table \ref{tab:set12_gau}. For noise level $\sigma \!=\! 50$, the DPLD reaches the best PSNR for all images among the five denoisers that do not consume ground truth data, except for the images \textit{Parrot} and \textit{Barbara}. It is interesting to note that BM3D performs better than the fully supervised method DnCNN on the image \textit{Barbara}. 
On average, for $\sigma \!=\! 50$ it outperforms Noise2Self by $0.91$ dB  and SURE by $0.43$ dB respectively, and it falls behind the DnCNN by $0.14$ dB.

\def\widthfive{0.15\linewidth}
\def\getdemoa#1{
\begin{tikzpicture}[spy using outlines]
   \node{\includegraphics[width=\widthfive]{#1}};
    \spy [blue, magnification=7, width=\widthfive,height=0.8*\widthfive, line width=20] on (-0.25*\widthfive, -0.06*\widthfive) in node at (0,-0.94*\widthfive);
    \spy [red, magnification=7, width=\widthfive,height=0.8*\widthfive, line width=2] on (0.26*\widthfive,0.07*\widthfive) in node at (0,-1.78*\widthfive);
\end{tikzpicture}
}
\begin{figure*}[h!]
    \centering
    \setlength{\tabcolsep}{0.5pt}
    \begin{tabular}{cccccc}
        \getdemoa{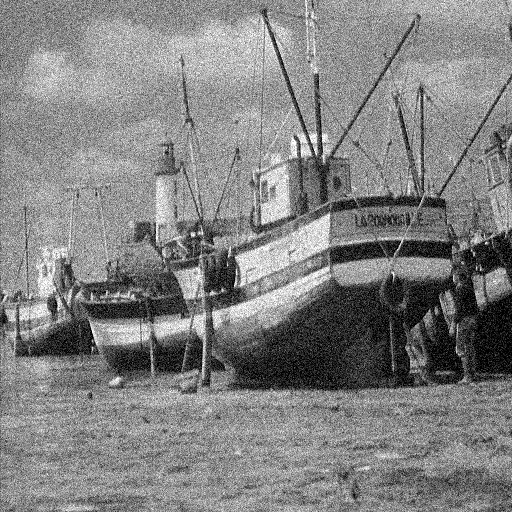} 
        &
        \getdemoa{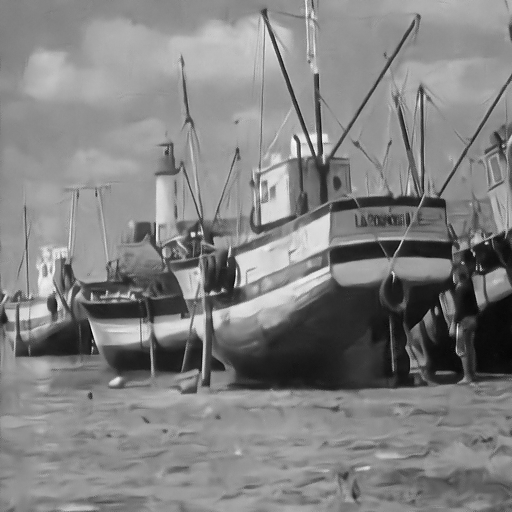}
        &
        \getdemoa{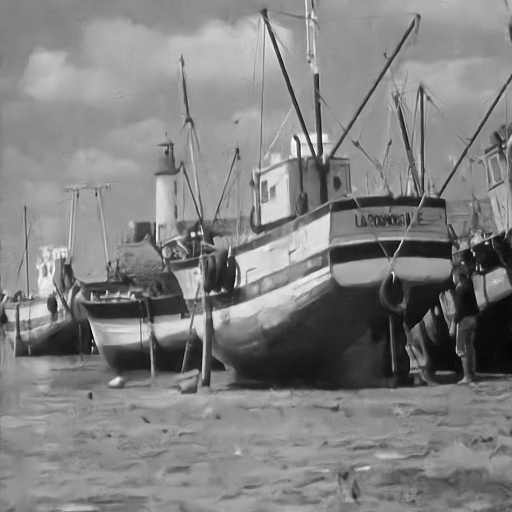} %
        &
        \getdemoa{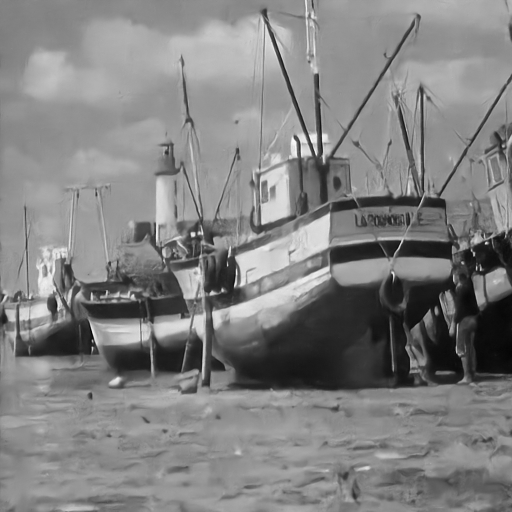}
        &
        \getdemoa{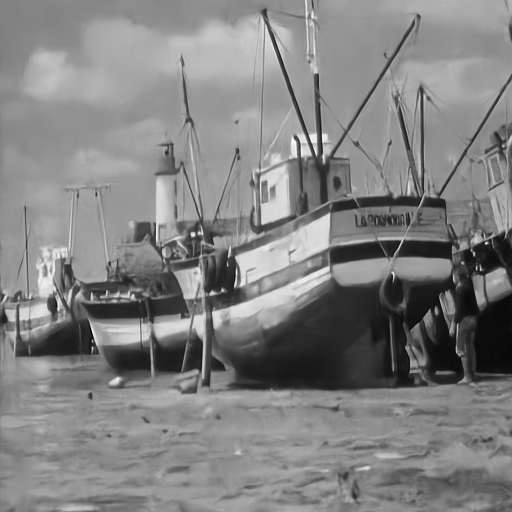}
        &
        \getdemoa{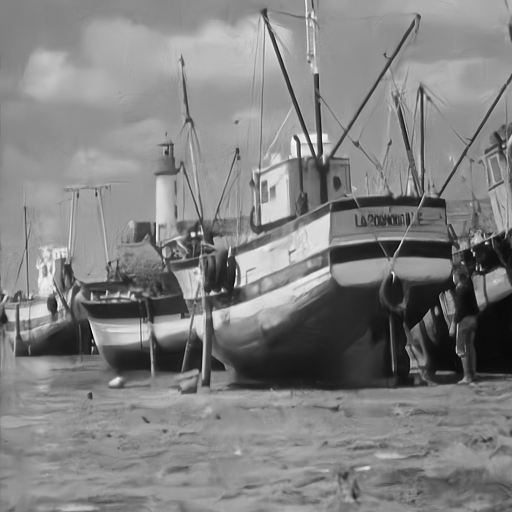}
        \\
        \footnotesize{Noisy} & \footnotesize{Noise2Self} & \footnotesize{SURE} & \footnotesize{Noisier2Noise} & \footnotesize{DPLD} & \footnotesize{DnCNN} \\
    \end{tabular}
    \caption{Denoised results for the image Boat (with Gaussian noise {\small $\sigma \!=\! 25$}). The last two rows are enlarged views of the indicated regions.}
    \label{fig:denoised_gau}
\end{figure*}

Fig. \ref{fig:denoised_gau} displays the denoising results for the image "Boat". It can be seen that, though the Noise2Self, SURE, Noisier2Noise, and DPLD do not see the clean images or have any explicit smoothness constraints during the training stages, they yield denoised images with smooth regions (Cf. the 2nd to 5th columns of Fig. \ref{fig:denoised_gau}, respectively). 
The resulted image from DPLD looks close to that of DnCNN visually.
The output of Noise2Self has some relatively blurry edges compared to DPLD and DnCNN, e.g., at the letters displayed in the last row of Fig. \ref{fig:denoised_gau}.

\begin{figure}[h!]
    \centering
    \setlength{\tabcolsep}{1pt}
    \begin{tabular}{cc}
        \includegraphics[width=0.5\linewidth, trim=0 7 0 10, clip]{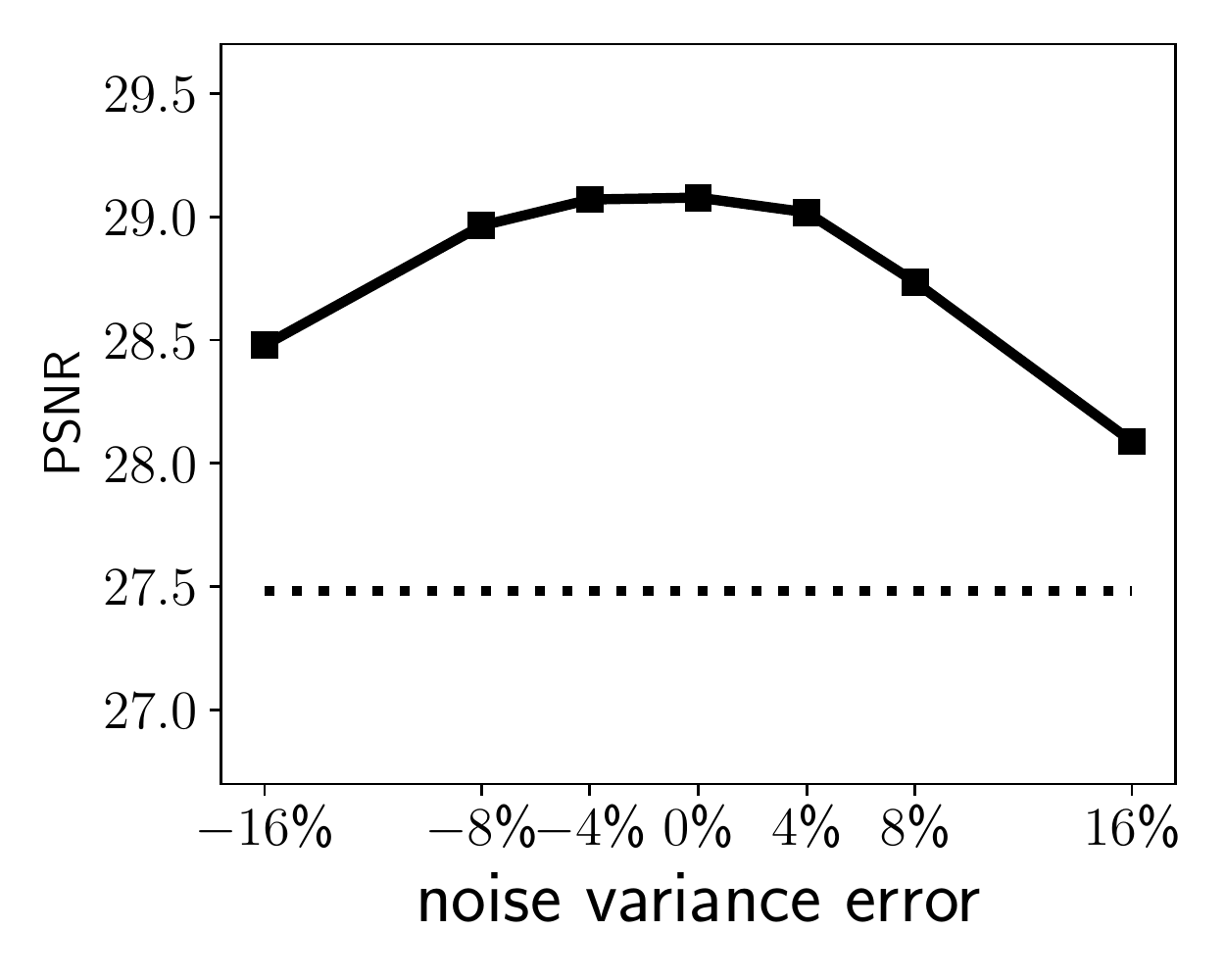}
        &
        \includegraphics[width=0.5\linewidth,trim=0 7 0 10, clip]{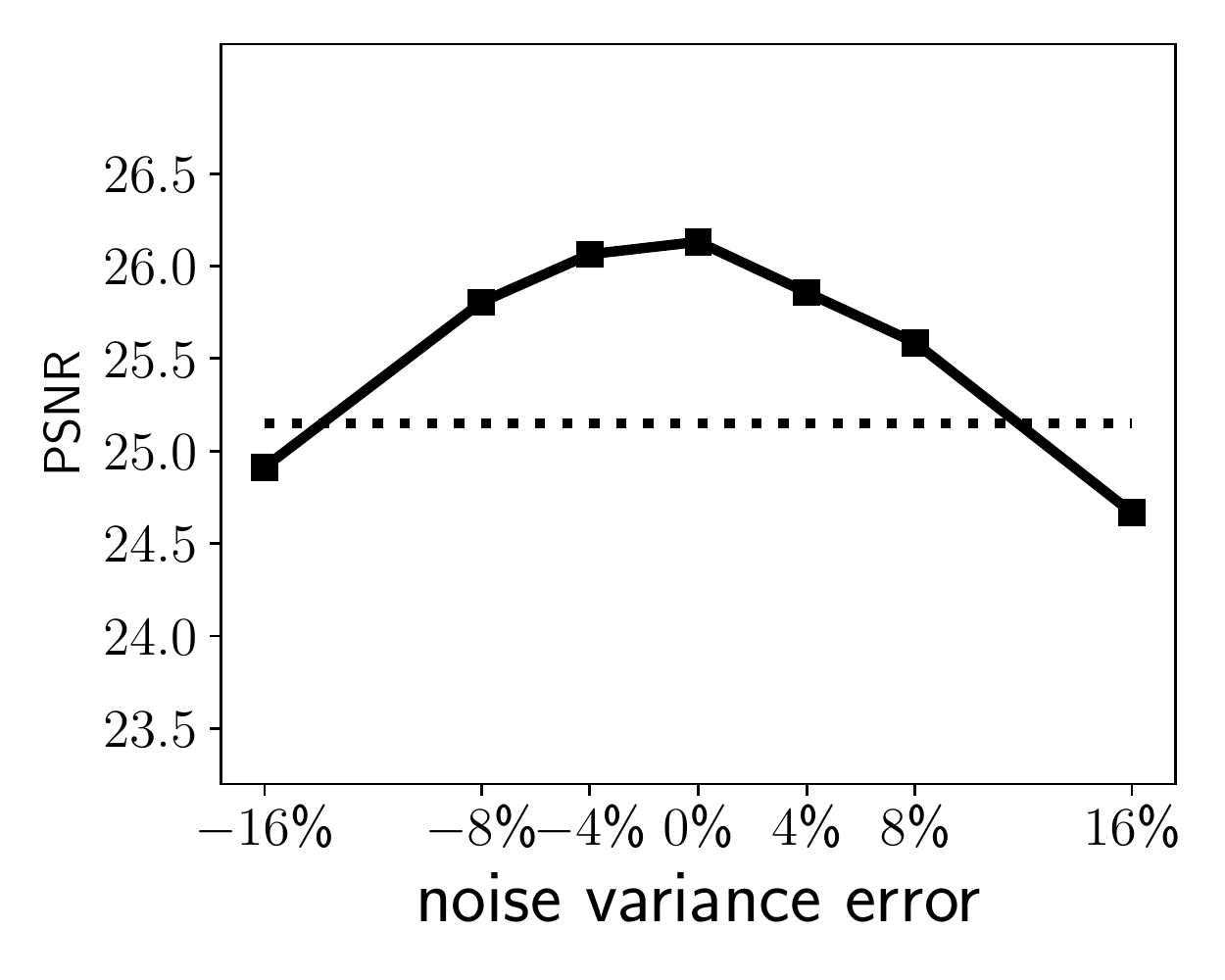}\\
        \includegraphics[width=0.5\linewidth,trim=0 7 0 10, clip]{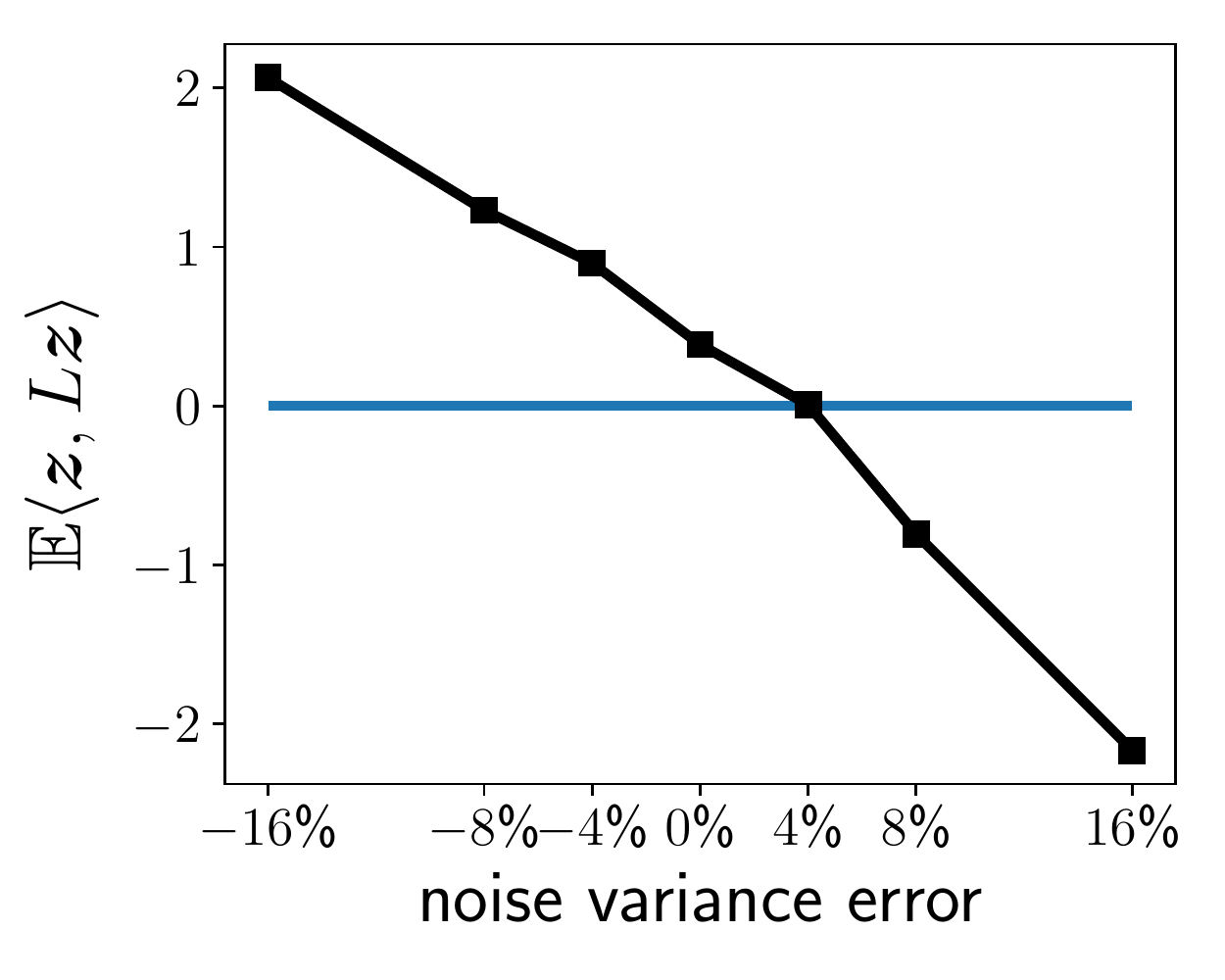}& 
        \includegraphics[width=0.5\linewidth,trim=0 7 0 10, clip]{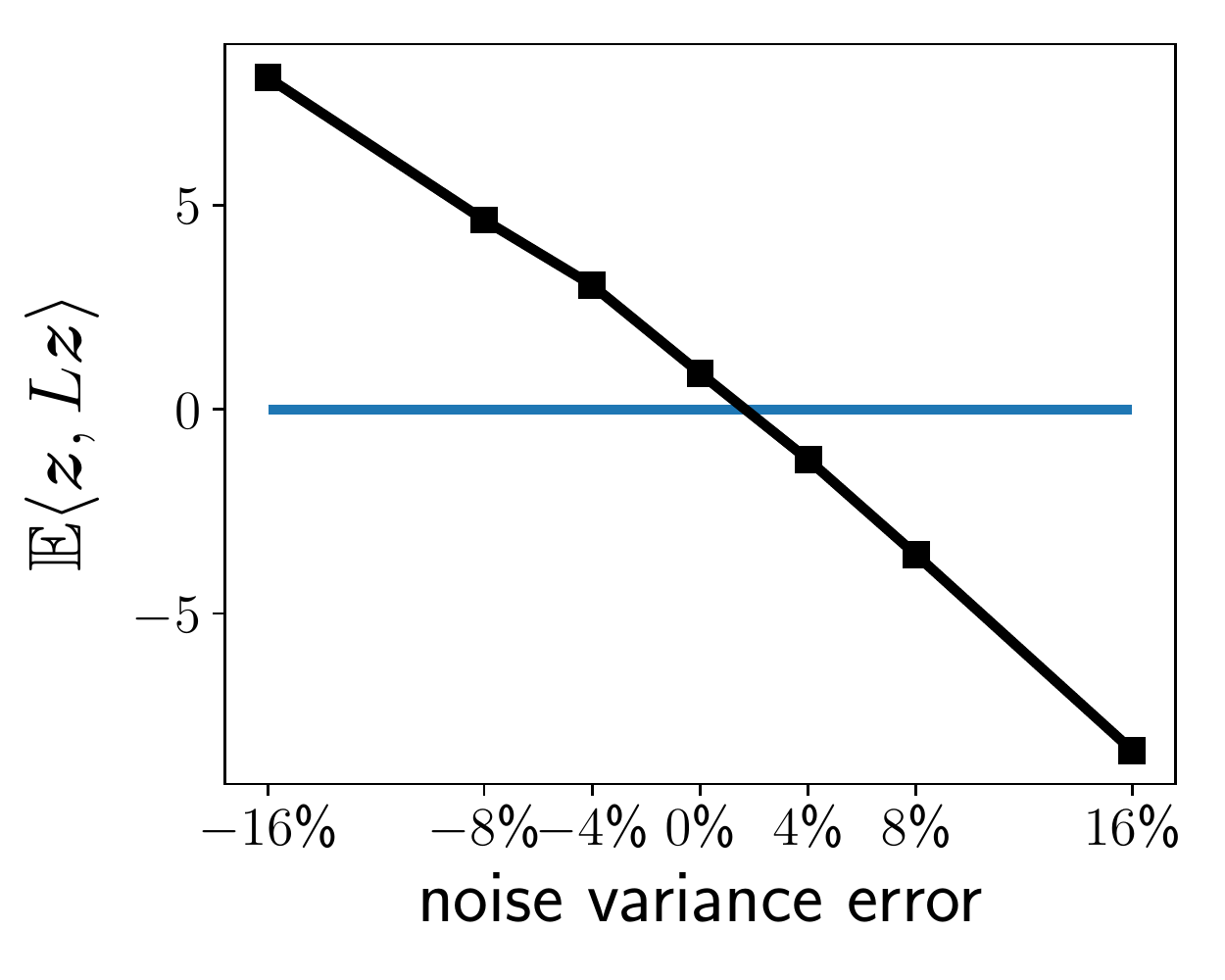} \\
        \footnotesize{ (a). $\sigma=25$} & \footnotesize{ (b). $\sigma=50$} 
    \end{tabular}
    \caption{The influence of inaccurate noise variance during training. Top row: PSNR on the test set BSD68 \cite{martin2001database}. The dotted horizontal lines indicate results of Noise2Self \cite{batson2019noise2self}. Bottom row: the mean of $\qutan{\bmz, L\bmz}$ of the learned denoisers.}\label{fig:noisevarianceerror}
\end{figure}

\subsubsection{Training with estimated noise variance}\label{subsect:env}
We report additional results for the scenarios where only an inaccurate estimated noise variance (ENV) is available (for generating $\bmz$) during training. The results are obtained under the same settings as Subsection \ref{subsect:exp_gaussian} except that an ENV is used. 
\begin{enumerate}
    \item[-] On the top row of Fig. \ref{fig:noisevarianceerror}, the PSNR values of the test results on BSD68 \cite{martin2001database}, for both noise levels $\sigma\!=\!25$ and $\sigma\!=\!50$, are plotted against the relative errors in the ENV $(1+\beta)\sigma^2$ (where $\beta\!=\!-16\%, -8\%, \cdots, 16\%$). 
For both noise levels, the errors in the ENV lead to a decrease in PSNR. The reduction is less significant when the noise variance is underestimated than when it is overestimated. The drop in PSNR is small when the error in the ENV is small ($<\!0.1$ dB for an ENV $4$\% less than its true value), and our method maintains an accuracy significantly better than Noise2Self \cite{batson2019noise2self} for small errors in the ENV.
\item[-]  We observe a correlation between the errors in the ENV and the term $\bbE\qut{\qutan{\bmz, L\bmz}}$ (see Fig. \ref{fig:noisevarianceerror}, bottom row). Here $L$ is computed based on Remark \ref{remark1}, with a constant ground truth image $\bmx$ (pixel values: $0.5$, size: $40\!\times\!40$) and generated $\bmn$ following the same distribution as $\bmz$ (variance: $(1+\beta)\sigma^2$). In the plots,  $\bbE\qut{\qutan{\bmz, L\bmz}}$ is a small positive number when $\beta\!=\!0$, and it increases (resp. decreases) as the ENV decreases (resp. increases). The is due to the extra term $2 \beta \bbE\qut{\qutan{\bmz, L\bmz}}$ in the objective function $\mathcal{J}\qut{\cdot}$ as shown in Remark \ref{remark2}.
\end{enumerate}
Overall, the performance is stable to small errors in ENV, and importantly, by analyzing the structure of the learned denoiser one could investigate the error in the noise variance. The quantity $\bbE\qut{\qutan{\bmz, L\bmz}}$ is helpful for choosing a better denoising model when the noise variance is unknown.

\subsubsection{Poisson noise}\label{subsect:exp_poission}
We evaluate our method on three different levels of Poisson noise, with parameter $\lambda\!=\!60$, $\lambda\!=\!30$ and $\lambda\!=\!15$ respectively. The training settings for the denoisers are the same as the ones for Gaussian noise, except that the variances of the auxiliary vector samples $z_i$ are computed differently. As the noise $\bmn$ is image dependent and not identically distributed for all image pixels, its variances are not known without having the ground truth image. Note that however our method does not require a precise model on the noise distribution, but only an estimate of the noise variance. 
At pixel $k$, the noisy value satisfies $\lambda [\bmy]_k \sim  {\rm Pois}\qut{\lambda [\bmx]_k}$ where $\lambda$ is a known constant. So the noise variance ${\rm var}\qut{[\bmn_k]\!\mid\![\bmx]_k}\!=\! {\rm var}\qut{[\bmy]_k\!\mid \! [\bmx]_k}\!=\! \frac{1}{\lambda} [\bmx]_k \!=\! \frac{1}{\lambda} \bbE\qut{[\bmy]_k\!\mid\! [\bmx]_k}$. This implies that the sample value of $\frac{1}{\lambda}[\bmy]_k$ provides an unbiased estimate of the noise variance.
The entries of the auxiliary vector are then generated as $[\bmz]_k \!=\! {\bm r}_k\sqrt{[\bmy]_k/\lambda}$ where ${\bm r}_k$ is the standard Gaussian random variable, hence ${\rm var}\qut{ [\bmz]_k \! \mid \! [\bmx]_k} \!=\! \frac{1}{\lambda} \bbE\qut{[\bmy]_k\!\mid\! [\bmx]_k} \!=\! {\rm var}\qut{ [\bmn]_k \! \mid \! [\bmx]_k}$.
The partially linear constraint parameter $\gamma$ for \eqref{eq:loss_Lc} is tuned manually for the three noise levels.
In this subsection, unless specified otherwise, the results are obtained by setting $\gamma = 4, 16, 64$ for $\lambda = 60, 30, 15$, respectively. 

\begin{table}[ht]
    \centering
    \caption{Denoising quality, measured by PSNR (dB) and SSIM, for BSD68 \cite{martin2001database} corrupted by Poisson Noise.}
    \label{tab:bsd68_poi}
\begin{tabular}{|c|cc|cc||cc|}
\hline
\multirow{2}{*}{\makecell{Noise \\ Level}} & \multicolumn{2}{c|}{Noise2Self\cite{batson2019noise2self}} & \multicolumn{2}{c||}{DPLD} & \multicolumn{2}{c|}{DnCNN\cite{zhang2017beyond}} \\
\cline{2-7}
& PSNR & SSIM & PSNR & SSIM & PSNR & SSIM \\
\hline
$\lambda = 60$ & $27.78$ & $0.8660$ & $\bm{29.28}$ & $\bm{0.9018}$ & $29.43$ & $0.9081$ \\ 
$\lambda = 30$ & $26.56$ & $0.8285$ & $\bm{27.65}$ & $\bm{0.8625}$ & $27.86$ & $0.8743$ \\ 
$\lambda = 15$ & $25.42$ & $0.7902$ & $\bm{26.16}$ & $\bm{0.8210}$ & $26.39$ & $0.8348$ \\
\hline
\end{tabular}
\end{table}

For comparison, we train denoisers with Noise2Self \cite{batson2019noise2self} and DnCNN \cite{zhang2017beyond} on the same training set (whereas the ground truth images are available only for DnCNN). 
The denoising results on the test set BSD68 are reported in Table \ref{tab:bsd68_poi}. 
Trained on the ground truth images, the DnCNN has the highest average PSNR for all three noise levels. The proposed method DPLD outperforms Noise2Self by $1.5$ dB, $1.09$ dB and $0.74$ dB for the cases with $\lambda = 60, 30, 15$ respectively. On the other hand, it losses $0.15$ dB, $0.21$ dB and $0.23$ dB when compared to DnCNN. It should be noted that when the noise level decreases, the PSNR gap between DPLD and DnCNN gets smaller. In contrast, the gap between Noise2Self and DnCNN becomes larger as the noise becomes smaller. This may be due to the fact that the Noise2Self approach can not learn identity mapping. For a given pixel, the denoiser can not see its observed value and has to infer its value from the information of its neighboring pixels. 
If knowledge about the noise distribution is available, then the results can be improved by reusing the noisy images in the inference phase \cite{krull2019probabilistic}.  
Our denoiser uses all information of the noisy image, and, as shown in Proposition \ref{prop:2}, the gap between the DPLD and the best denoiser in $\mathcal{R}_\epsilon$ tends to zero as the noise goes to zero. 

\begin{table*}[ht]
\centering
\caption{Denoising quality (in dB) for the $12$ wildly used images \cite{dabov2007image} and Poisson noise}\label{tab:set12_poi}
\makebox[\textwidth][c]{\begin{tabular}{|c|c|c|cccccccccccc|}
\hline
& {Denoiser} & {Average} & {C. man} &  {House} &  {Peppers} &  {Starfish} &  {Monarch}  &  {Airplane}  &  {Parrot}  &  {Lena} &  {Barbara} &  {Boat} &  {Man} &  {Couple}  \\
\hline
\hline
\multirow{3}{*}[-0.4ex]{\rotatebox{90}{$\lambda = 60$}} 
& Noise2Self & $29.15$  & $28.32$ & $32.13$ & $29.79$ & $27.94$ & $28.99$ & $26.63$ & $27.98$ & $31.74$ & $28.60$ & $29.34$ & $29.20$ & $29.18$  \\ 
& DPLD & $\bm{30.34}$ & $30.05$ & $32.89$ & $30.82$ & $28.93$ & $30.54$ & $28.64$ & $29.68$ & $32.56$ & $29.79$ & $30.09$ & $30.02$ & $30.07$  \\ 
\cline{2-15}
& DnCNN & $30.49$ & $30.29$ & $33.02$ & $30.95$ & $29.07$ & $30.69$ & $28.79$ & $29.82$ & $32.67$ & $30.03$ & $30.22$ & $30.12$ & $30.21$  \\ 
\hline
\hline
\multirow{3}{*}[-0.4ex]{\rotatebox{90}{$\lambda = 30$}} 
& Noise2Self & $27.82$ & $27.36$ & $30.74$ & $28.27$ & $26.31$ & $27.72$ & $25.61$ & $26.95$ & $30.36$ & $26.79$ & $27.91$ & $28.02$ & $27.82$  \\ 
& DPLD & $\bm{28.68}$ & $28.43$ & $31.44$ & $29.06$ & $27.08$ & $28.74$ & $26.96$ & $28.14$ & $31.00$ & $27.74$ & $28.48$ & $28.59$ & $28.47$  \\ 
\cline{2-15}
& DnCNN & $28.87$ & $28.79$ & $31.67$ & $29.23$ & $27.16$ & $28.89$ & $27.11$ & $28.30$ & $31.14$ & $28.20$ & $28.60$ & $28.71$ & $28.67$  \\ 
\hline
\hline
\multirow{3}{*}[-0.4ex]{\rotatebox{90}{$\lambda = 15$}} 
& Noise2Self & $26.26$ & $26.01$ & $28.81$ & $26.60$ & $24.62$ & $25.99$ & $24.23$ & $25.58$ & $28.69$ & $24.94$ & $26.48$ & $26.78$ & $26.36$  \\ 
& DPLD & $\bm{26.99}$ & $27.00$ & $29.76$ & $27.27$ & $25.22$ & $26.88$ & $25.26$ & $26.63$ & $29.39$ & $25.31$ & $27.02$ & $27.25$ & $26.87$ \\ 
\cline{2-15}
& DnCNN & $27.28$ & $27.32$ & $30.13$ & $27.51$ & $25.42$ & $27.10$ & $25.47$ & $26.85$ & $29.62$ & $26.30$ & $27.20$ & $27.37$ & $27.11$  \\ 
\hline
\hline
\end{tabular}}
\end{table*}

Table \ref{tab:set12_poi} lists the PSNR of denoising outputs for the $12$ wildly used images \cite{dabov2007image}. Similar to the Gaussian noise cases, the proposed DPLD has higher PSNR values than Noise2Self on all images. The DPLD outperforms Noise2Self by more than $0.7$ dB in average PSNR, and falls behind DnCNN by less than $0.3$ dB.

\def\widthfive{0.18\linewidth}
\def\getdemo#1{
\begin{tikzpicture}[spy using outlines]
   \node{\includegraphics[width=\widthfive]{#1}};
    \spy [blue, magnification=5, width=\widthfive,height=0.8*\widthfive, line width=20] on (-0.3*\widthfive, -0.27*\widthfive) in node at (0,-.94*\widthfive);
    \spy [red, magnification=5, width=\widthfive,height=0.8*\widthfive, line width=2] on (0.2*\widthfive,0.02*\widthfive) in node at (0,-1.78*\widthfive);
\end{tikzpicture}
}
\begin{figure*}[h!]
    \centering
    \setlength{\tabcolsep}{1pt}
    \begin{tabular}{ccccc}
        \getdemo{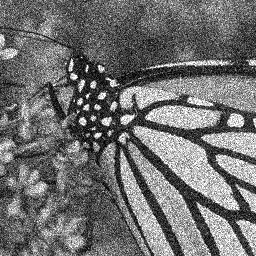}
        &
        \getdemo{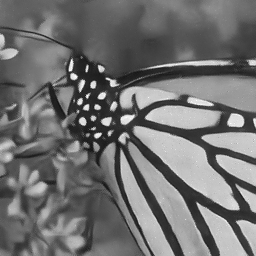}
        &
        \getdemo{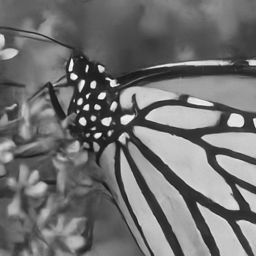}
        &
        \getdemo{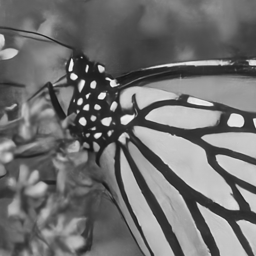}
        & 
        \getdemo{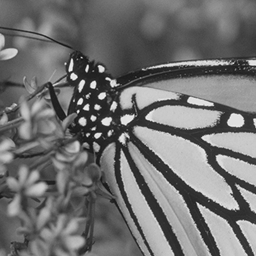} 
        \\
        \footnotesize{Noisy} & \footnotesize{Noise2Self} & \footnotesize{DPLD} & \footnotesize{DnCNN} & \footnotesize{Ground truth} \\
    \end{tabular}
    \caption{Quality comparison for different methods for Poisson noise ($\lambda\!=\!30$). The last two rows are enlarged views of the indicated regions.}
    \label{fig:poi_result}
\end{figure*}

Fig. \ref{fig:poi_result} displays an example of the denoised images. This example shows that Noise2Self, DPLD and DnCNN are capable of recovering the details of the image, though the first two are not exposed to the detailed structures of the images during training. Compared to DPLD and DnCNN, the denoised image of Noise2Self is less smooth. Also, Noise2Self tends to remove the sharp points of a jagged edge (Cf. the third row of Fig. \ref{fig:poi_result}), since it relies on the data of surrounding pixels when restoring the pixels at the sharp point and therefore 
may encourage more regularized shapes of objects.
\begin{figure*}[h!]
  \centering
  \setlength{\tabcolsep}{3pt}
  \begin{tabular}{c|ccc}
       \includegraphics[trim=35 30 30 30, clip, width=0.2\linewidth]{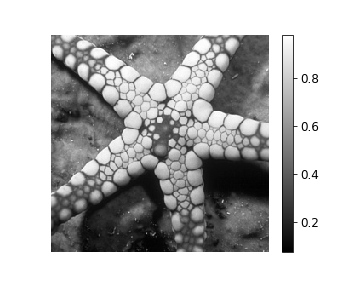}
       &
       \includegraphics[trim=35 30 30 30, clip, width=0.2\linewidth]{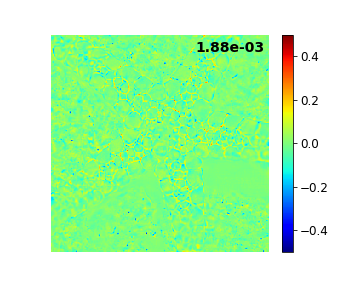}
       & 
       \includegraphics[trim=35 30 30 30, clip, width=0.2\linewidth]{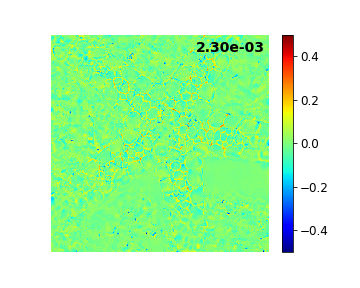}
       &
       \includegraphics[trim=35 30 30 30, clip, width=0.2\linewidth]{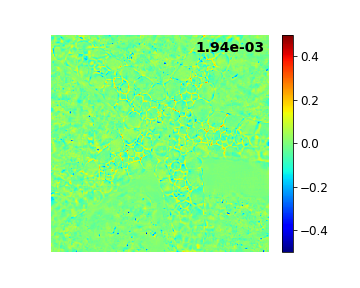}
       \\
       \addlinespace[-1.0ex]
       \footnotesize{Ground truth $\bmx$} & \footnotesize{DnCNN error} & \footnotesize{Noise2Self error} & \footnotesize{Ours error} 
       \\
       \addlinespace[0.5ex]
       \includegraphics[trim=35 30 30 30, clip, width=0.2\linewidth]{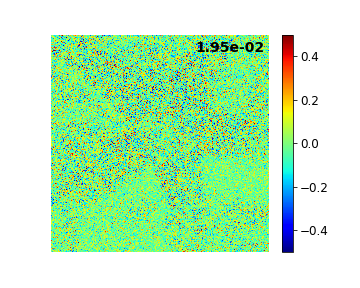}
       &
       \includegraphics[trim=35 30 30 30, clip, width=0.2\linewidth]{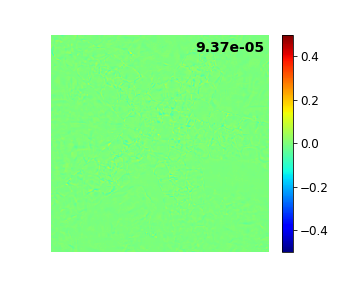}
       & 
       \includegraphics[trim=35 30 30 30, clip, width=0.2\linewidth]{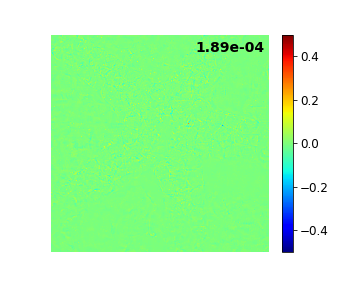}
       &
       \includegraphics[trim=35 30 30 30, clip, width=0.2\linewidth]{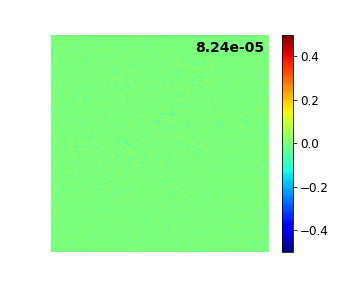}
       \\
       \addlinespace[-1.0ex]
       \footnotesize{Noise $\bmn$} & \footnotesize{DnCNN residual $\bme$} & \footnotesize{Noise2Self residual $\bme$} & \footnotesize{Ours residual $\bme$}
  \end{tabular}
  \caption{Residual terms of DnCNN, Noise2Self and our method.
  The error images displayed on the first row are computed by subtracting the ground truth from the denoised images. The number on the top right corner of the images is the mean square of pixel values.}
  \label{fig:poisson_residual}
\end{figure*}

Finally, the residual term $\bme$ (defined in \eqref{eq:part-lin-pure}) of three different methods, on the test image \textit{Starfish}, are presented in 
Fig. \ref{fig:poisson_residual}. 
For each method, in order to compute $\bme$ we first compute $g\qut{\bmx}$ and $L$ with Remark \ref{remark1} and using its denoised outputs for 1,600,000 realizations of noisy image $\bmy$. This experiment is carried out under the setting of Poisson noise with $\lambda\!=\!30$. As can be seen from Fig. \ref{fig:poisson_residual}, the error (i.e., the difference between the ground truth and the denoised image) is an order of magnitude smaller than the noise, while the residual term is an order of magnitude smaller than the error. The variance of $\bme$ of the proposed denoiser is slightly smaller than the fully supervised denoiser DnCNN, and Noise2Self has a variance of $\bme$ about two times larger than the other two. 

\subsection{Stability with respect to the partial linear constraint and noise levels}\label{subsect:exp_stability}
In this experiment we study the stability of the proposed denoiser against two parameters, $\gamma$ and the noise levels. Here we consider Poisson noise with fixed $\lambda\!=\!30$ for training, under the same experimental setting as Subsection \ref{subsect:exp_poission}. 

First we study the stability concerning $\gamma$.
Fig. \ref{fig:stability_gamma} shows the PSNR and SSIM, evaluated on the test set BSD68, for different choices of $\gamma$ used during training. The horizontal lines in Fig. \ref{fig:stability_gamma} show the PSNR and SSIM of the other two methods DnCNN and Noise2Self (N2S), which do not depend on $\gamma$. 
According to the figure, the learned denoiser has very low PSNR values, when $\gamma$ is very small (e.g., $10^{-4}$, which means that less important is put on the constraint). This implies that the partial linearity is crucial for finding a high quality denoiser when minimizing \eqref{eq:empiricalloss}. As $\gamma$ becomes larger, such constraint starts to take effect and substantially improve the denoising performance. The peak of PSNR is achieved at around $2^4$. For $\gamma$ in $[2^2,2^8]$, we observe a variation of PSNR less than $0.2$ dB, which is relatively small compared to the gap between N2S and DnCNN, before the denoising quality degrades at very large $\gamma$ values (e.g., $2^{14}$).

Fig. \ref{fig:stability_lambda} compares the robustness of the learned denoisers in terms of different noise levels. All denoisers are trained in the Poisson noise setting (again with $\lambda\!=\!30$), and in our approach the optimal $\gamma$ (i.e. $2^4$) is applied. It can be seen that all methods suffer from a degradation in the denoising quality when $\lambda$ in testing is below $30$. This is reasonable because in general smaller $\lambda$ implies higher noise and more difficulty. Besides, the DPLD and Noise2Self, which use only noisy images in training, are more stable than DnCNN when $\lambda$ is around $20$. 
The PSNR of DnCNN decreases quickly as $\lambda$ decreases from $30$, and it is worse than DPLD when $\lambda \leq 24$. However, it is interesting to note that the DPLD outperforms the (fully supervised) DnCNN in the lower noise cases (corresponding to big $\lambda$), even though neither the ground truth nor less noisy samples are provided for the training. This suggests that auxiliary random vector approach together with the partially linear constraint brings extra robustness compared to its counterpart that learns from a set noisy-clean image pairs with fixed noise levels. 

\begin{figure}[h!]
  \centering
  \includegraphics[width=0.8\linewidth, trim=0 15 0 12, clip]{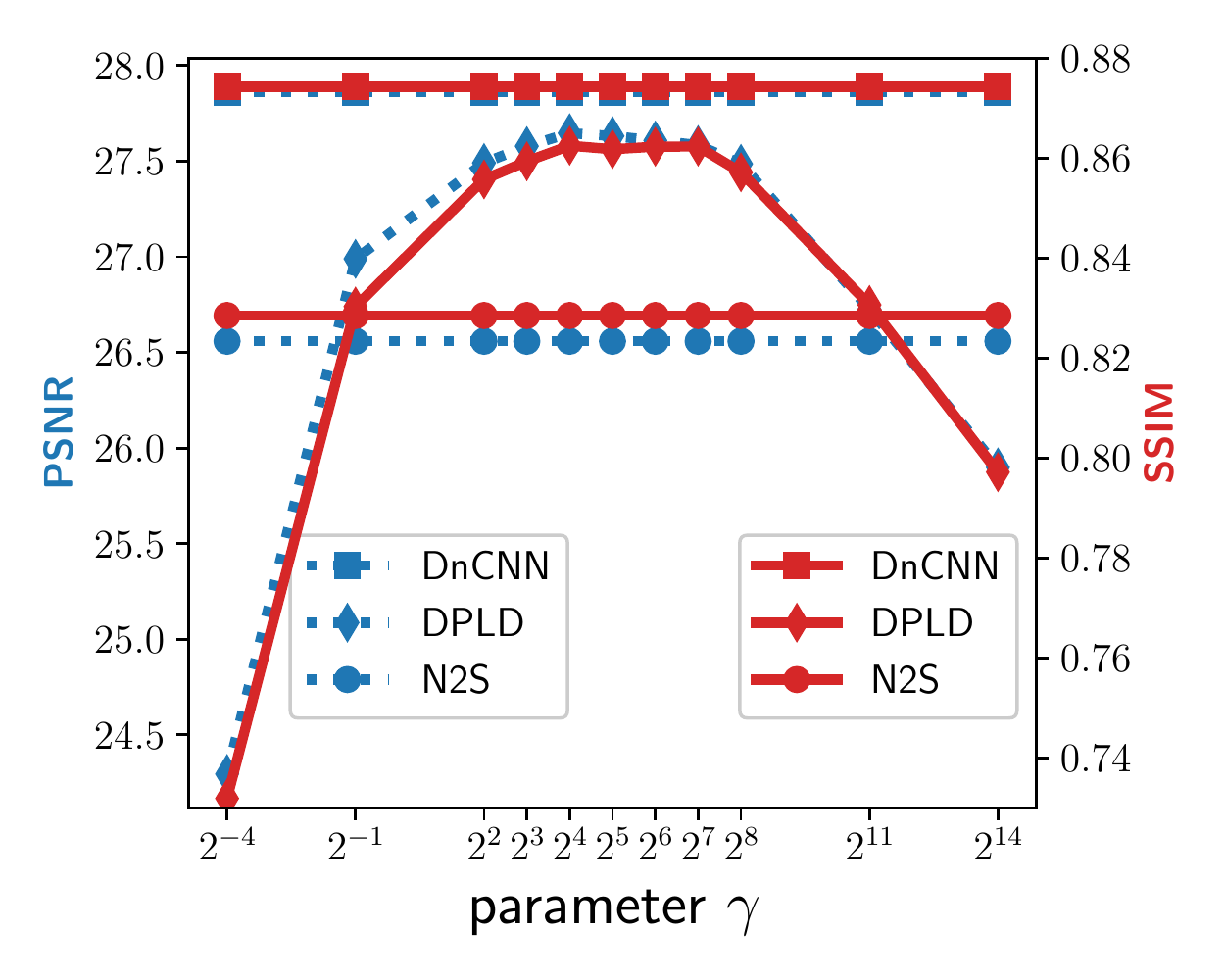}
  \caption{PSNR (in blue) and SSIM (in red) plotted against the partially linear constraint parameter $\gamma$ (better viewed in color). The  horizontal lines represent results of the DnCNN and Noise2Self for comparison.}
  \label{fig:stability_gamma}
\end{figure}
\begin{figure} 
  \centering
  \includegraphics[width=0.7\linewidth, trim=0 15 0 12, clip]{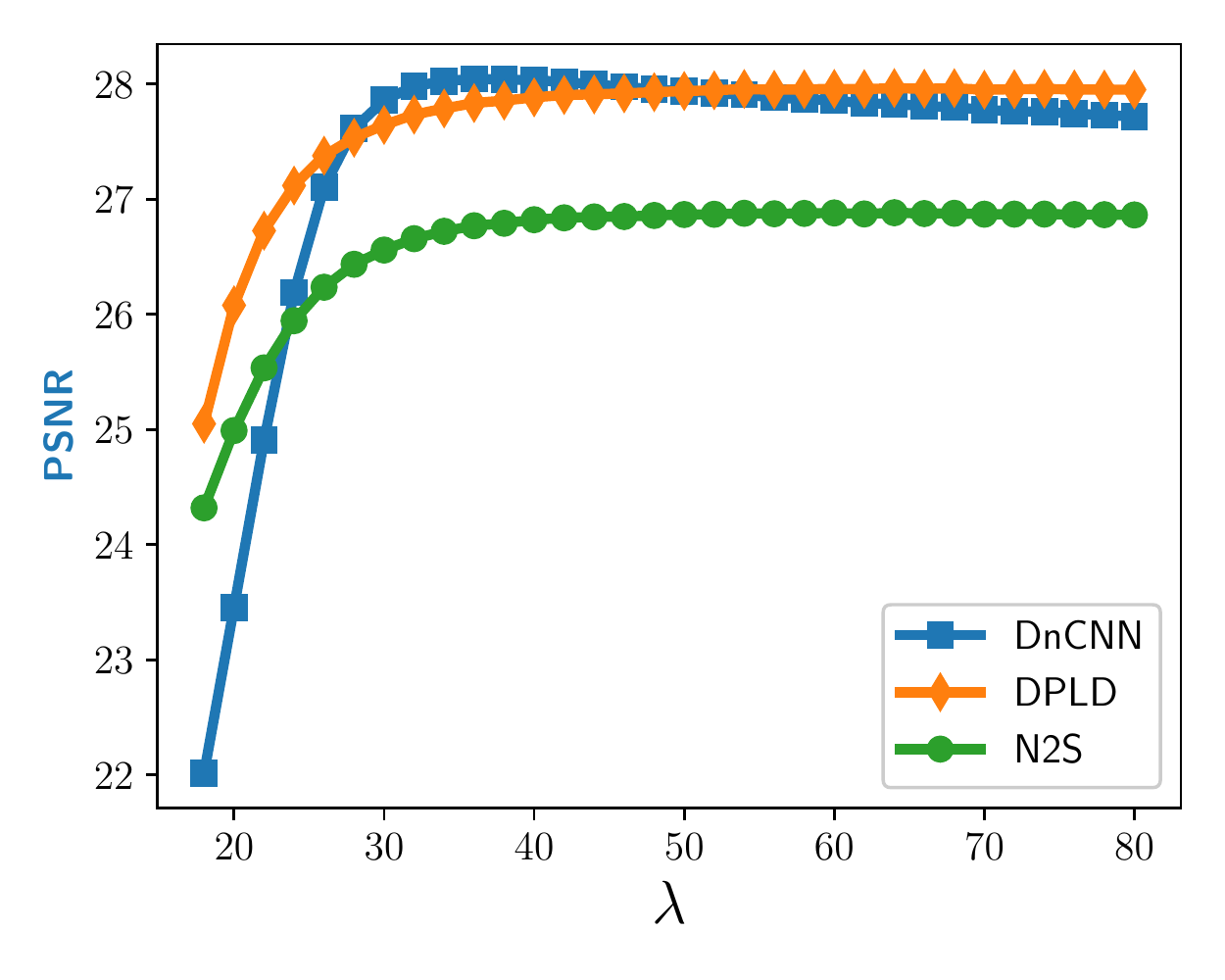}
  \caption{The Robustness of the denoisers with respect to the noise levels.}
  \label{fig:stability_lambda}
\end{figure}

\subsection{Denoising real microscopy images}\label{subsect:microscopydenoising}
We apply the proposed method to denoising real microscopy images. Neither the ground truth images nor their noise variance is known. Specifically, we consider three datasets, namely N2DH-GOWT1, C2DL-MSC, and TOMO110, which are used in \cite{krull2019noise2void}. The first two are fluorescence microscopy images of GFP transfected GOWT1 mouse stem cells and rat mesenchymal stem cells, respectively. The training sets contain $92$ images of size $1024\!\times\!1024$ (N2DH-GOWT1) and $48$ images of size $992\!\times\!832$ (C2DL-MSC). The third one TOMO110 is acquired by a Cryo-transmission electron microscope, and for this one we use one image (size $7676\!\times\! 7420$) for training.

For each of the datasets, we first estimate the noise variance.\\ 1). For N2DH-GOWT1 and C2DL-MSC, we first obtain a rough estimate via computing the average squared difference of neighboring pixels at smooth regions. We then refine the estimated noise variance using the correlation between the linear component of the denoiser and the errors in the noise variance. This is motivated by the observation in Remark \ref{remark2} and the numerical study in Subsection \ref{subsect:env}.\\  2). For TOMO110, since multiple frames/copies of the image are available because of the dose-fractionated acquisition, we estimate the variance by computing the variance of the frames for which the noise is independent. Note that though several frames of the image are available, our approach learns to denoise the sum of all frames. In contrast, the Noise2Noise method learns to denoise only one half of the frames, the noise level for which is therefore higher. 

The details of the noise estimation and training of the networks are given in Appendix B.
The denoising results are displayed in Fig. \ref{realmicro} (with enlarged views for the regions indicated by the blue/red boxes). We compare our method with BM3D \cite{dabov2007image} and Noise2Void \cite{krull2019noise2void}. As shown in the figure, the BM3D method has denoised images that are smooth, but also contains artifacts (e.g., in the cells on the second row of Fig. \ref{realmicro}). The Noise2Void method yields denoising outputs with certain graininess effect at non-constant regions (Cf. the third row and first column of the figure). The denoised images from our approach look smooth, yet it shows slightly more details of the object structures (e.g. the regions highlighted by the red boxes).

\def\widthtr{0.31\linewidth}
\def\getdemodmo#1#2#3#4#5{
\begin{tikzpicture}[spy using outlines]
  \node{\includegraphics[width=\widthtr]{#1}};
    \spy [blue, magnification=3, width=0.46*\widthtr, height=0.46*\widthtr, line width=2] on (#2*\widthtr,#3*\widthtr) in node at (-0.265*\widthtr,-0.75*\widthtr);
    \spy [red, magnification=3, width=0.46*\widthtr, height=0.45*\widthtr, line width=2] on (#4*\widthtr,#5*\widthtr) in node at (0.26*\widthtr,-0.75*\widthtr);
\end{tikzpicture}
}
\def\getdemodmoL#1#2#3#4#5{
\begin{tikzpicture}[spy using outlines]
  \node{\includegraphics[width=\widthtr]{#1}};
    \spy [blue, magnification=5, width=0.46*\widthtr, height=0.46*\widthtr, line width=2] on (#2*\widthtr,#3*\widthtr) in node at (-0.265*\widthtr,-0.75*\widthtr);
    \spy [red, magnification=5, width=0.46*\widthtr, height=0.45*\widthtr, line width=2] on (#4*\widthtr,#5*\widthtr) in node at (0.265*\widthtr,-0.75*\widthtr);
\end{tikzpicture}
}
\begin{figure}
    \centering
    \setlength{\tabcolsep}{0pt}
    \begin{tabular}{cccc}
    \raisebox{.7\height}{\rotatebox{90}{Noisy image}} &
    \getdemodmo{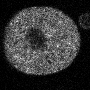}{-0.28}{-0.15}{-0.08}{0.23} &  
    \getdemodmo{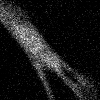}{-0.18}{0.05}{0.22}{-0.21} &
    \getdemodmoL{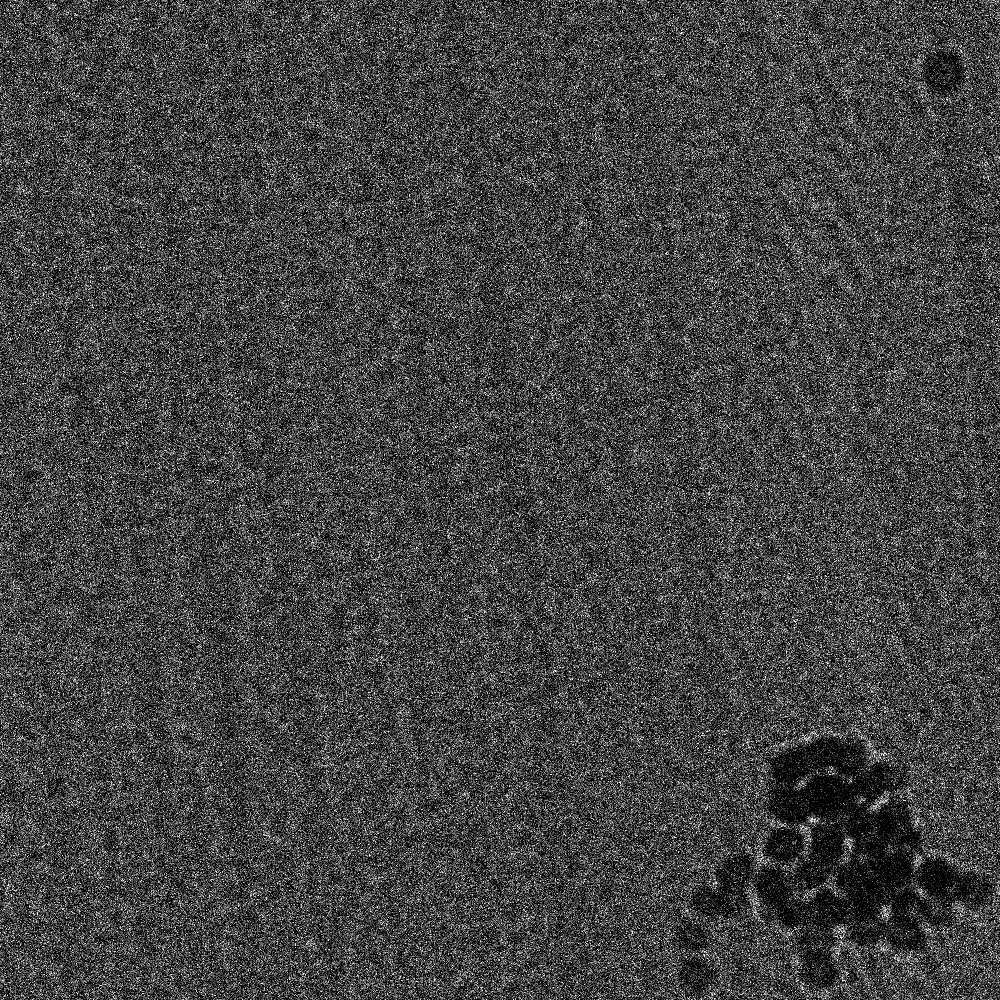}{-0.32}{0.32}{0.33}{-0.25}
    \\
    \raisebox{1.7\height}{\rotatebox{90}{BM3D}} &
    \getdemodmo{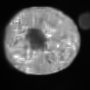}{-0.28}{-0.15}{-0.08}{0.23}
    & \getdemodmo{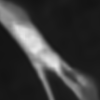}{-0.18}{0.05}{0.22}{-0.21}
    &
    \getdemodmoL{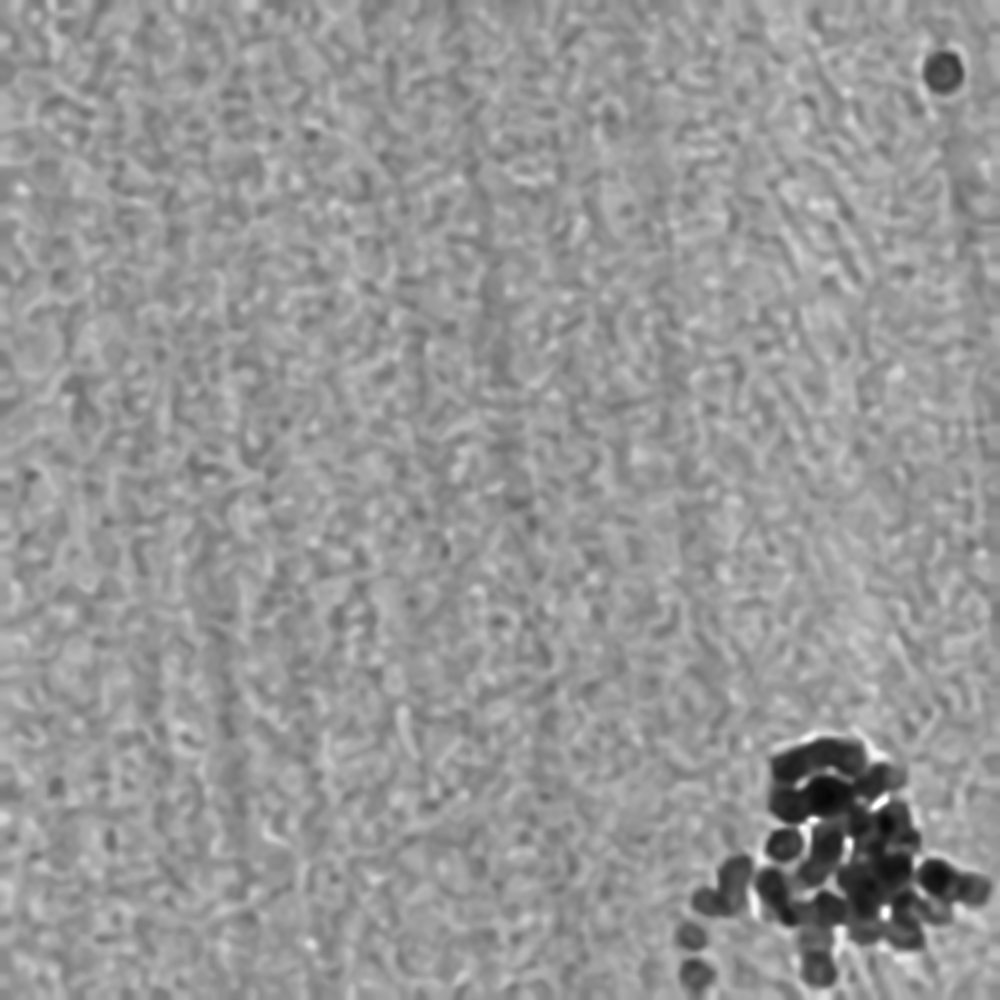}{-0.32}{0.32}{0.33}{-0.25}
    \\
    \raisebox{0.7\height}{\rotatebox{90}{Noise2Void}} &
    \getdemodmo{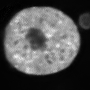}{-0.28}{-0.15}{-0.08}{0.23}&
    \getdemodmo{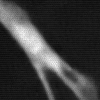}{-0.18}{0.05}{0.22}{-0.21} &
    \getdemodmoL{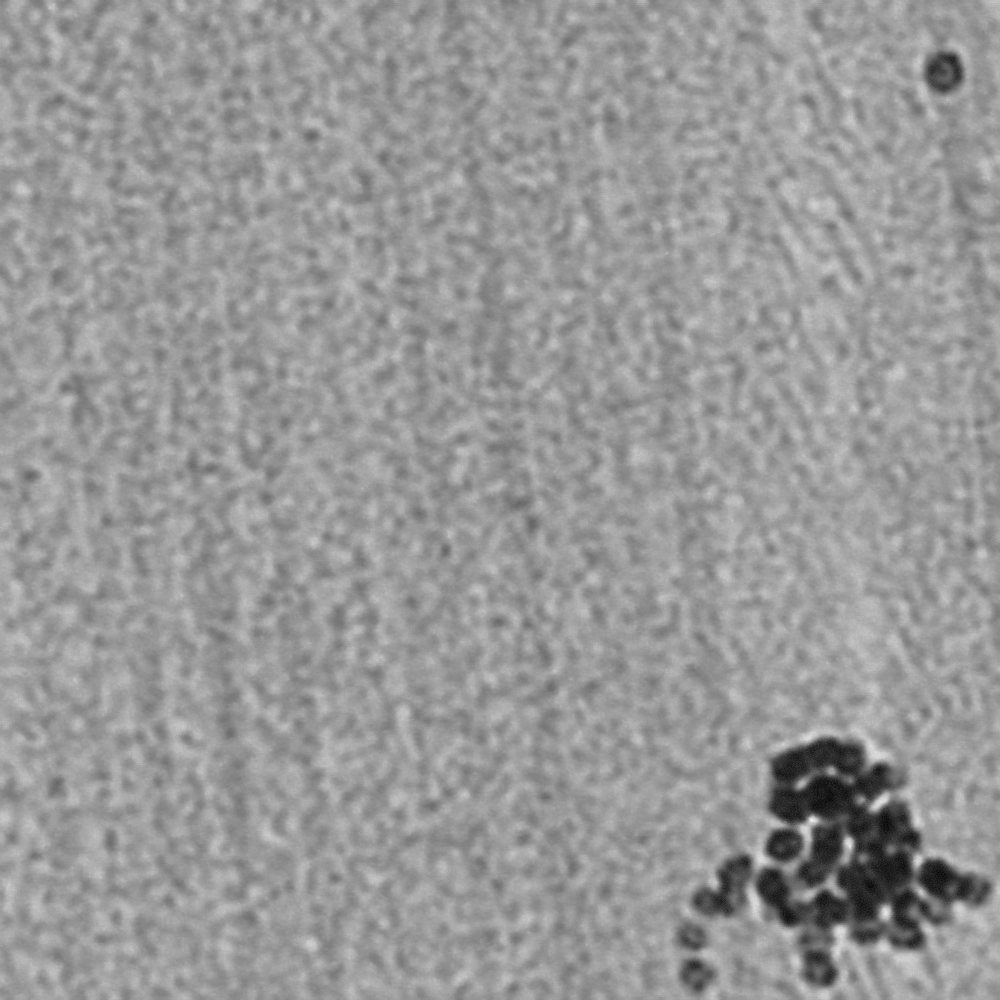}{-0.32}{0.32}{0.33}{-0.25}
    \\
    \raisebox{2.8\height}{\rotatebox{90}{Ours}} &
    \getdemodmo{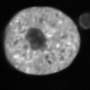}{-0.28}{-0.15}{-0.08}{0.23} &
    \getdemodmo{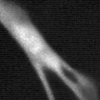}{-0.18}{0.05}{0.22}{-0.21} &
    \getdemodmoL{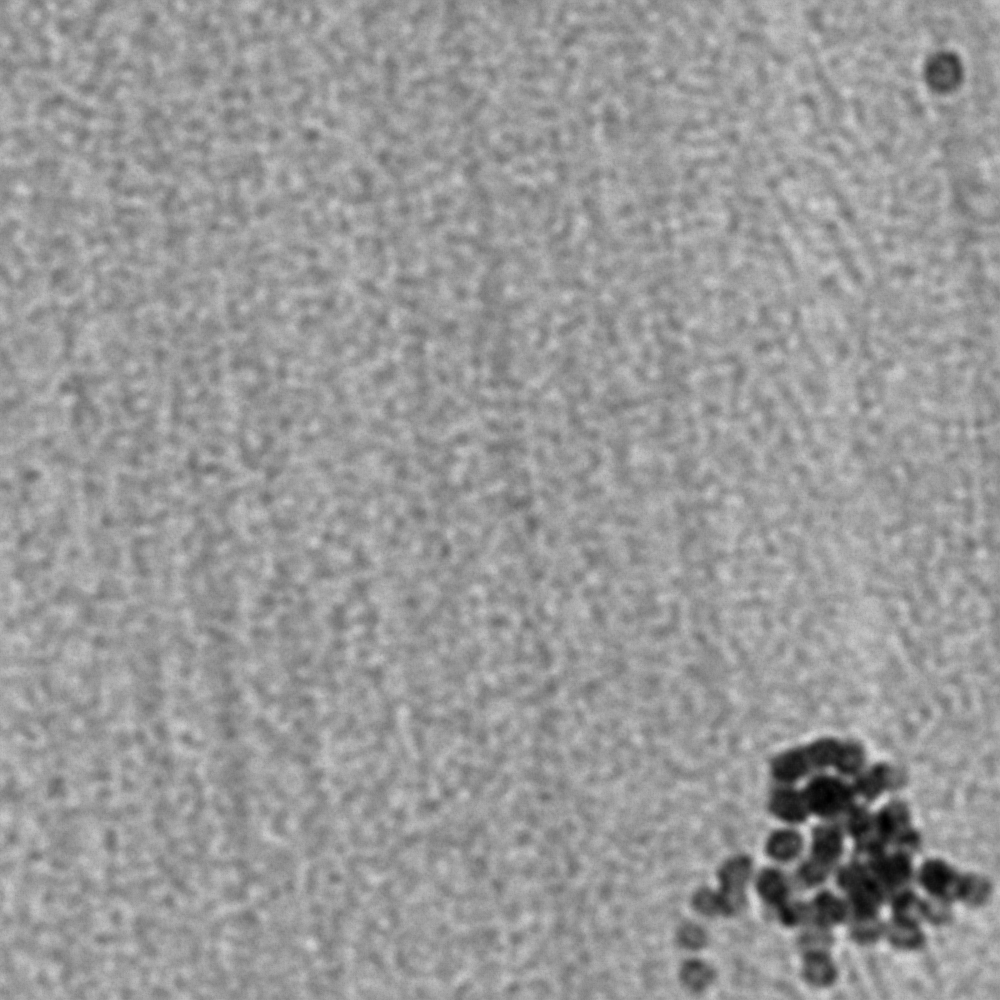}{-0.32}{0.32}{0.33}{-0.25}
    \\
    & {\small N2DH-GOWT1} & {\small C2DL-MSC} & {\small TOMO110}\\
    \end{tabular}
    \caption{Denoising results of three real microscopy datasets (on 3 columns).}\label{realmicro}
\end{figure}

\subsection{Image deblurring}\label{subsect:blind_deblurring}
Following \cite{shen2018deep} and \cite{xia2019training}, we consider the image deblurring problem \eqref{eq:deblur_problem} where the operator $A$ is a convolution operator with a blur kernel arising from random motions. 
The datasets for the image deblurring experiments are face images from the Helen dataset \cite{le2012interactive} and CelebA dataset \cite{liu2015deep}, and we use the same training/validation split as in \cite{xia2019training}. 
This results in a total of $161,800$ images for training, $22,000$ images for validation, and $16,000$ images for testing. 
Throughout this experiment, the noise $\bmn$ is Gaussian white noise with standard deviation $\sigma \!=\! 2$. We compare our approach with the fully supervised deblurring method, as well as an unsupervised method \cite{xia2019training} that uses a pair of noisy and blurry observations per image. To facilitate the comparison, the same U-Net architecture in \cite{xia2019training} for the deblurring networks is used in all three methods. 

We train the deblurring networks $R_\theta$ using the empirical loss of  \eqref{eq:deblurring_loss}, given the random noisy samples $\{y_i\}$ and the operator $A$,
\[
\mathcal{L}_{\rm b} := \sum_i \qutn{ A{R_\theta}\qut{\hat{y}_i} - \qut{y_i - z_i/\alpha}}^2
\]
where $\{z_i\}$ are the auxiliary vectors. As suggested in  \cite{xia2019training}, we also integrate a proxy loss defined as 
\[
\mathcal{L}_{\rm prox} :=  \sum_i \qutn{ R_\theta\qut{y^{\rm prox}_{i}} - x^{\rm prox}_{i} }^2
\]
in which 
$x^{\rm prox}_{i} := R_\theta\qut{y_{i}}$,  $y^{\rm prox}_{i}:= A^{\rm prox} x^{\rm prox}_{i} + n^{\rm prox}_i$,
and $A^{\rm prox}$, $n^{\rm prox}_i$ are randomly generated blur kernels and noise respectively. In our setting $n^{\rm prox}_i$ is generated from the same distribution as $\bmz$. So this loss function is the MSE loss defined on the synthetic training pair $(x^{\rm prox}_{i}, y^{\rm prox}_{i})$ where $x^{\rm prox}_{i}$, a deblurred version of $y_i$, is treated as the ground truth. The proxy loss has been found useful in improving the deblurring accuracy. 
Finally, we apply the partially linear constraint in \eqref{eq:loss_Lc} with the denoiser being $A R_\theta\qut{y_{i}}$ in this context. In summary, the training loss function is $\mathcal{L}_{\rm b} + \gamma_{\rm prox}\mathcal{L}_{\rm prox} + \gamma \mathcal{L}_{c}$ where $\gamma_{\rm prox}$ is a hyperparameter.

The deblurring model is trained using the Adam optimizer \cite{kingma2014adam} with a batch size of $128$ and $360,000$ optimization steps, starting from an initial learning rate of $10^{-3}$. The learning rate is then decreased to $3\times 10^{-4}$, $10^{-4}$, $5\times 10^{-5}$ at the $310,000^{\rm th}$,
$340,000^{\rm th}$, $350,000^{\rm th}$ step, respectively. 
In the first $100,000$ optimization steps, we let $\gamma_{\rm prox} = 0$ and $\gamma = 0$, and after that, $\gamma_{\rm prox} = 1/16$, $\gamma = 1/16$. In this experiment, we let $\alpha$ be randomly selected from $[0.1,0.2]$ for each training sample. 

\begin{table}[ht]
    \centering
    \setlength{\tabcolsep}{4pt}
    \caption{Deblurring results. For the data used in the training step, O and C stand for observations and clean images, respectively. The mark $\dagger$ indicates training without the proxy loss.}
    \label{tab:deblur}
    \begin{tabular}{|c|c|cc|cc|}
    \hline
    \multirow{2}{*}{Method} & \multirow{2}{*}{Data} & \multicolumn{2}{c|}{Helen \cite{le2012interactive}} & \multicolumn{2}{c|}{CelebA \cite{liu2015deep}} \\
    \cline{3-6}
    & & PSNR & SSIM & PSNR & SSIM \\
    \hline
    Xu et al. \cite{xu2013unnatural} & -- & 
    20.11 & 0.711 & 18.93 & 0.685 \\ 
    Zuo et al. \cite{zuo2016learning} & O/C Pairs & 22.24 & 0.763 & 20.53 & 0.750\\
    Tao et al. \cite{tao2018scale} & O/C Pairs & 22.86 & 0.762 & 24.11 & 0.862 \\
    Shen et al. \cite{shen2018deep} & O/C Pairs & 25.99 & 0.871 & 25.05 & 0.879 \\ 
    Kupyn et al. \cite{kupyn2018deblurgan} &  O/C Pairs & 23.63 & 0.781 & 22.45 & 0.729 \\
    \hline
    Supervised baseline
    & O/C Pairs & 26.13 & 0.886 & 25.20 & 0.892 \\
    \hline
    Xia et al. \cite{xia2019training}$\dagger$ & O/O Pairs & 25.47 & 0.867 & 24.64 & 0.873 \\
    Ours$\dagger$ & Unpaired O & 25.65 & 0.867 & 24.78 & 0.873\\
    Xia et al. \cite{xia2019training} & O/O Pairs & 25.95 & 0.878 & \textbf{25.09} & 0.885 \\
     Ours & \begin{tabular}{c}
    Unpaired O
     \end{tabular}
      & \textbf{26.00} & \textbf{0.879} & \textbf{25.09} & \textbf{0.886} \\
     \hline
    \end{tabular}
\end{table}
Table \ref{tab:deblur} shows the PSNR and SSIM of the deblurring results for test sets of Helen \cite{le2012interactive} and CelebA \cite{liu2015deep} respectively. In the table, the supervised baseline is based on the same U-net architecture as in our approach and trained on the (noisy) observation/clean image pairs. It achieves the highest PSNR and SSIM scores among the methods being compared. 
Our method reaches competitive PSNR values of 25.65 dB (Helen) and 24.78 dB (CelebA) \textit{without} the proxy loss, and adding the proxy loss results in a further improvement of $0.35$ dB and $0.31$ dB respectively.
As an extension of the Noise2Noise \cite{lehtinen2018noise2noise}, the approach proposed by Xia et al. \cite{xia2019training}, in contrast, relies on paired noisy observations for each image but without the ground truth. For the results reported in the table, the operator $A$ is known during training for both our approach and the method of Xia et al. \cite{xia2019training}. However, our method requires only unpaired observations, and it leads to a deblurring quality comparable to Xia et al. \cite{xia2019training} and not far away from the fully supervised baseline (with a PSNR gap smaller than $0.15$ dB). Furthermore, our method outperforms the existing supervised methods Zuo et al. \cite{zuo2016learning},
Tao et al. \cite{tao2018scale},
Shen et al. \cite{shen2018deep},
Kupyn et al. \cite{kupyn2018deblurgan}, as well as the unsupervised approach Xu et al. \cite{xu2013unnatural}.

A comparison of the deblurred images from different methods is shown in Fig. \ref{fig:deblur}, where the ground truth and blurry images are taken from the test sets of Helen. 
We underline that no blur kernels are provided to these methods during test time. 
Given no ground truth images for training, the deblurred results of Xia et al. \cite{xia2019training} (Cf. the $4^{\rm th}$ column in the figure) and our approach (Cf. the $5^{\rm th}$ column) are able to capture most of the details recovered by the fully supervised baseline (Cf. the $3^{\rm rd}$ column), though the latter gives slightly sharper images. This demonstrates the capacity of the proposed partially linear denoisers for solving the deblurring problem based on only one corrupted observation per image. Though in this experiment the unsupervised methods require knowledge about the blur kernels during the training phase, they can be generalized to the cases where the blur kernels are unknown for training. This can be done, e.g., by jointly learning the clean images and blur kernels \cite{xia2019training}.

\def\widthfive{0.18\linewidth}
\def\getdemod#1#2#3{
\begin{tikzpicture}[spy using outlines]
  \node{\includegraphics[width=\widthfive]{#1}};
    \spy [red, magnification=4, width=\widthfive, height=.618*\widthfive, line width=2] on (#2*\widthfive,#3*\widthfive) in node at (0,-0.83*\widthfive);
\end{tikzpicture}
}
\begin{figure*}[h!]
    \centering
    \setlength{\tabcolsep}{1pt}
    \begin{tabular}{ccccc}
        \getdemod{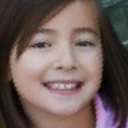}{-0.3}{0.1}
        &
        \getdemod{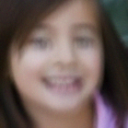}{-0.3}{0.1}
        & 
        \getdemod{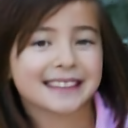}{-0.3}{0.1}
        & 
        \getdemod{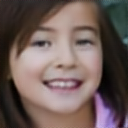}{-0.3}{0.1}
        &
        \getdemod{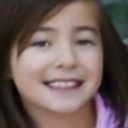}{-0.3}{0.1}
        \\
        Ground truth & Blurry image & Fully supervised & Xia et al. \cite{xia2019training} & Ours \\
    \end{tabular}
    \caption{Face deblurring results. The images are taken from the test sets of Helen \cite{le2012interactive}.}
    \label{fig:deblur}
\end{figure*}

\section{Conclusion}
In this paper, we propose a class of structured nonlinear denoisers. We show that such denoisers, equipped with a partial linearity property, can be trained when we do not have access to any ground truth images, nor to the exact model of the noise. In practice, one only needs to know the noise variance conditioned on the images. 
Based on the partial linearity structure, we proposed an auxiliary random vector approach, which establishes a direct connection between our loss function and the MSE, and allows end-to-end training of the denoising models. 
The approach outperforms other ground-truth-free denoising approaches such as the SURE based learning method for denoising, having a denoising quality close to that of the fully supervised baseline. The approach also offers new opportunities for learning to solve other image restoration tasks from single corrupted observations of the images. The experimental results show that, when generalized to an image deblurring problem, our approach achieves the state of the art for unsupervised deblurring. 

One disadvantage of our method is that it does not work for non-zero mean noise, such as impulse noises for which the corrupted pixels are replaced with random numbers. In such cases, the auxiliary random vector approach no longer provides good estimates to the MSE. Our method requires estimates of the noise variance. Various noise variance estimation methods from noisy images exist (see e.g., \cite{liu2006noise,foi2008practical, acito2011signal, liu2014practical}). Alternatively, in our future work, we are interested in learning-based or automated noise variance extraction methods with the potential of being integrated in the proposed denoiser within an end-to-end framework.

\bibliographystyle{plain}
\bibliography{references}

\IEEEpeerreviewmaketitle

\clearpage
\appendices
\setcounter{page}{1} %

\twocolumn[
\begin{center}
{\Huge
Supplementary Materials:\\
Unsupervised Image Restoration Using Partially Linear Denoisers}\\
\vspace{0.4cm}
{\large Rihuan Ke  and Carola-Bibiane Sch\"onlieb}
\vspace{0.2cm}
\end{center}
]

    \section{Additional proofs for Section \ref{sec:method}}

\noindent\textbf{Remark} \ref{remark2}. Assume that the conditions of Proposition \ref{prop:2} hold except that ${\rm Cov}\qut{\bmz, \bmz \mid \bmx} \!=\! (1+\beta) {\rm Cov}\qut{\bmn, \bmn \mid \bmx}$ where $\beta\!>\!-1$. Then $\mathcal{J}\qut{R} = \bbE\qut{ \qutn{{R(\bmyh) - \bmx}}^2} - 2 \bbE\qut{\qutan{\bme, {\bmn - \bmz/\alpha}}} + c + 2 \beta \bbE\qut{\qutan{\bmz, L\bmz}}
$. 
\begin{proof}
The proof is similar to that of Proposition \ref{prop:2}. 

Firstly, we repeat Equation \eqref{eq:prop2-eq1} here 
\begin{equation*}\label{eq:prop2-eq1_repeat}
    \begin{split}
    R\qut{\bmyh} - \qut{\bmy - \bmz/\alpha} 
     = & \qut{g\qut{\bmx} + L \bmnh + \bme - \bmx} \\ & - \qut{\bmn - \bmz/\alpha}.
\end{split}
\end{equation*}
Taking expectation of both sizes over $\bmx$, $\bmn$ and $\bmz$, one has
\begin{equation}\label{eq:appendixA_eq1}
    \begin{split}
    \mathcal{J}\qut{R}
    = 
    & \bbE\qut{ \qutn{R(\bmyh) - \bmx}^2}
     + \bbE\qut{ \qutn{\bmn - \bmz/\alpha}^2} \\
    & - 2\bbE\qut{ \qutan{g(\bmx) - \bmx, \bmn - \bmz/\alpha}} \\
    & - 2\bbE\qut{ \qutan{\bme, \bmn - \bmz/\alpha}} \\
    & - 2\bbE\qut{ \qutan{L\bmn+\alpha L\bmz, \bmn - \bmz/\alpha}}. \\
    \end{split}
\end{equation}

Secondly, by the assumption ${\rm Cov}\qut{\bmz, \bmz \mid \bmx} \!=\! (1+\beta) {\rm Cov}\qut{\bmn, \bmn \mid \bmx}$, we have the following equalities associated with the last term of Equation \eqref{eq:appendixA_eq1}:
\begin{equation}\label{eq:noise_compensate_env}
\begin{split}
& \bbE\qut{\qutan{ L{\bmn} + \alpha L{\bmz}, \bmn - \bmz/\alpha } \mid \bmx}\\ 
= & \bbE\qut{\qutan{ L{\bmn}, \bmn }\mid \bmx} +
\bbE\qut{\qutan{\alpha L{\bmz},  -\bmz/\alpha }\mid \bmx} \\
= & - \beta \bbE\qut{\qutan{ L{\bmz},  \bmz }\mid \bmx}.
\end{split}
\end{equation}
Also, the term $\bbE\qut{ \qutan{g(\bmx) - \bmx, \bmn - \bmz/\alpha}} = 0$ according to the proof of Proposition \ref{prop:2}. 

Finally, let $c:= \bbE\qut{\qutn{\bmn - \bmz/\alpha}^2}$, then the desired equality follows from Equations \eqref{eq:noise_compensate_env} and
\eqref{eq:appendixA_eq1}. 
The proof is completed.  
\end{proof}

\propositiond*
\begin{proof}
To simplify the notations, let $\bmp:=R^{\epsilon,\delta}\qut{\bmyh}$ and $\hat{\bmp}:=\hat{R}^{\epsilon,\delta}\qut{\bmyh}$. 

\textbf{I}). By Lemma \ref{lm:convex},  $R^{\epsilon,\delta}+t\qut{\hat{R}^{\epsilon,\delta}-R^{\epsilon,\delta}} \in \mathcal{R}_\epsilon$ for $t \in [0,1]$. It follows from the definition of $R^{\epsilon,\delta}$ in \eqref{eq:rmmse} that 
\begin{equation}\label{eq:prop4-definition1}
 \bbE\qut{\qutn{\bmp-\bmx}^2} - \delta \leq \bbE\qut{\qutn{\bmp+t\qut{\bmph-\bmp}-\bmx}^2}
\end{equation}
Besides, the right hand side of the above inequality can be decomposed as 
\begin{equation}\label{eq:prop4-eq1}
\begin{split}
   & \bbE\qut{\qutn{\bmp+t\qut{\bmph-\bmp}-\bmx}^2} \\
   = & \bbE\qut{\qutn{\bmp-\bmx}^2} + 2 t \bbE\qut{\qutan{\bmp-\bmx, \bmph-\bmp}} \\
   & + t^2 \bbE\qut{\qutn{\bmph-\bmp}^2}. 
\end{split}    
\end{equation}
Combining \eqref{eq:prop4-definition1} and \eqref{eq:prop4-eq1}, one has 
\[2\bbE\qut{\qutan{\bmp-\bmx, \bmph-\bmp}} \geq - s \bbE\qut{\qutn{\bmph-\bmp}^2 }- \delta/s, \ \forall s \in (0,1].\] 
Letting $t=1$, Equation \eqref{eq:prop4-eq1} and the last inequality force
\begin{equation}\label{eq:ineq1}
\begin{split}
    & \qut{1-s}\bbE\qut{\qutn{\bmph - \bmp}^2} \\
    \leq  & \bbE\qut{\qutn{\bmph - \bmx}^2} -  \bbE\qut{\qutn{\bmp - \bmx}^2} + \delta /s, \ \forall s \in (0,1].    
\end{split}
\end{equation}

\textbf{II}). With the same argument as in \textbf{I}), for $\hat{R}^{\epsilon,\delta}$ one can show that for any $s \in (0,1]$, %
\begin{equation}\label{eq:ineq2}
\begin{split}
    \qut{1-s} \bbE\qut{\qutn{\bmph - \bmp}^2} & \leq \bbE\qut{\qutn{\bmp - \qut{\bmy-\bmz/\alpha}}^2} \\ 
    & \quad - \bbE\qut{\qutn{\bmph - \qut{\bmy-\bmz/\alpha}}^2} + \delta /s   \\
    & = \mathcal{J}\qut{R^{\epsilon,\delta}} - \mathcal{J}\qut{\hat{R}^{\epsilon,\delta}} + \delta /s.  
\end{split}
\end{equation}
Note that \eqref{eq:ineq1} and \eqref{eq:ineq2} give two upper bounds of $\bbE\qut{\qutn{\bmph - \bmp}^2}$. Averaging both sides of the two inequalities,
\[
\begin{split}
    \qut{1-s} \bbE\qut{\qutn{\bmph - \bmp}^2} \leq 
 & \ \bbE\qut{\qutn{\bmph - \bmx}^2}/2 - \mathcal{J}\qut{\hat{R}^{\epsilon,\delta}}/2 \\
&
+ \mathcal{J}\qut{R^{\epsilon,\delta}}/2
- \bbE\qut{\qutn{\bmp - \bmx}^2}/2 \\
& + \delta /s.
\end{split}
\]
If we set $s = \sqrt{\delta}$, the desired inequality then follows from Proposition \ref{prop:2}.
\end{proof}

\propositiondb*
\begin{proof}
Let 
$\bmp:=R^{\epsilon,\delta}\qut{\bmyh}$ and $\hat{\bmp}:=\hat{R}^{\epsilon,\delta}\qut{\bmyh}$. 
According to Equation \eqref{eq:part-lin}, we have the decomposition
\begin{equation}\label{eq:1211}
\begin{split}
\bmp & = g_1 \qut{\bmx} + L_1 \bmnh + \bme_1, \\
\bmph & = g_2 \qut{\bmx} + L_2 \bmnh + \bme_2,
\end{split}    
\end{equation}
where $L_i$ is linear and depends on $\bmx$, and $\qutn{\bme_i} \leq \epsilon$ for $i \in \{0,1\}$. Let $\bm{q} := \bmp - \bme_1 + \bme_2$, then 
\begin{equation}\label{eq:1212}
\bm{q}\qut{\bmyh} = g_1\qut{\bmx} + L_1 \bmnh + \bme_2,
\end{equation}
as a function of $\bmyh$, is well defined because the distribution of $\bmx$ is a delta function and $\bm{q}\qut{\hat{y}_1}=\bm{q}\qut{\hat{y}_2}$ if $\hat{y}_1=\hat{y}_2$ for any realizations $\hat{y}_1$, $\hat{y}_2$ of $\bmyh$. Furthermore, it is easy to check that $\bm{q} \in \mathcal{R}_\epsilon$. 

Analogous to the argument in \eqref{eq:ineq2}, for $\bm{q} \in \mathcal{R}_\epsilon$ we have
\begin{equation}\label{eq:1213}
\begin{split}
    \qut{1-s} \bbE\qut{\qutn{\bmph - \bm{q}}^2} & \leq \bbE\qut{\qutn{\bm{q} - \qut{\bmy-\bmz/\alpha}}^2} \\ 
    & \quad - \bbE\qut{\qutn{\bmph - \qut{\bmy-\bmz/\alpha}}^2} + \delta /s, 
\end{split}
\end{equation}
for any $s \in (0,1]$. 
It follows from Equation \eqref{eq:part-lin-mse} that
\[
\begin{split}
\bbE\qut{\qutn{\bm{q} - \qut{\bmy-\bmz/\alpha}}^2}  = & \bbE\qut{\qutn{\bm{q} - \bmx }^2} - 2\bbE\qut{ \qutan{\bme_2, \bmn-\bmz/\alpha}  } \\
& + c,\\
\bbE\qut{\qutn{\bmph - \qut{\bmy-\bmz/\alpha}}^2}  = & \bbE\qut{\qutn{\bmph - \bmx }^2} - 2\bbE\qut{ \qutan{\bme_2, \bmn-\bmz/\alpha}  }  \\
& + c\\
\end{split}
\]
for some constant $c$. This together with \eqref{eq:1213} gives 
\begin{equation*}\label{eq:prop5eq2}
    \begin{split}
    \qut{1-s} \bbE\qut{\qutn{\bmph - \bm{q}}^2} \leq 
 &  \bbE\qut{\qutn{\bm{q} - \bmx }^2} - \bbE\qut{\qutn{\bmph - \bmx }^2} + \delta /s.
\end{split}
\end{equation*}
Combining the last inequality with \eqref{eq:ineq1}, one has 
\begin{equation}\label{eq:prop5eq100}
\begin{split}
    & \qut{1-s} \bbE\qut{\qutn{\bmph - \bm{q}}^2} + \qut{1-s} \bbE\qut{\qutn{\bmph - \bm{p}}^2} \\
    \leq 
 &  \bbE\qut{\qutn{\bm{q} - \bmx }^2} - \bbE\qut{\qutn{\bmp - \bmx }^2} + 2 \delta /s.
\end{split}    
\end{equation}
Next, we will show that  $\bbE\qut{\qutn{\bm{q} - \bmx }^2} - \bbE\qut{\qutn{\bmp - \bmx }^2} \leq \epsilon^2$ which together with \eqref{eq:prop5eq100} gives the desired result \eqref{eq:coro6-1} by setting $s = \sqrt{\delta}$. 

From Remark \ref{remark1}, the definitions of $L_1$ and $\bme_1$ in 
\eqref{eq:1211} imply
\[
\begin{split}
\bbE\qut{\qutn{\bme_1}^2\mid \bmx} & = \bbE\qut{\qutn{\bmp - L_1 \bmnh}^2|\bmx}\\ & \leq \bbE\qut{\qutn{\bmp - L_1 \bmnh + tL \bmnh}^2|\bmx}    
\end{split}
\]
for any $t\in \mathbb{R}$ and any linear $L$. It forces that $\bbE\qut{\qutan{L \bmn, \bme_1}\mid \bmx} = 0$. Furthermore, the conditional mean of $\bme_1$ is zero, i.e., $\bbE\qut{\bme_1\mid \bmx} = 0$. Similarly, one can also show that $\bbE\qut{\qutan{L \bmnh, \bme_2}\mid \bmx} = 0$ and $\bbE\qut{\bme_2\mid \bmx} = 0$. Therefore,
\[
\bbE\qut{\qutan{g_1\qut{\bmx} + L_1 \bmnh - \bmx, \bme_i} \mid \bmx} = 0, {\rm~for~} i \in \{1, 2\},
\]
and consequently, 
\[
\begin{split} 
& \bbE\qut{\qutn{\bm{q} - \bmx }^2 \mid \bmx}  \\
 = & \bbE\qut{\qutn{g_1\qut{\bmx} + L_1 \bmnh + \bme_2 - \bmx}^2\mid \bmx} \\
\leq & \bbE\qut{\qutn{\bme_2}^2\mid \bmx} + \bbE\qut{\qutn{g_1\qut{\bmx} + L_1 \bmnh - \bmx}^2\mid \bmx} \\
\leq & \bbE\qut{\qutn{\bme_2}^2\mid \bmx} + \bbE\qut{\qutn{g_1\qut{\bmx} + L_1 \bmnh + \bme_1 - \bmx}^2\mid \bmx} \\
= & \bbE\qut{\qutn{\bme_2}^2\mid \bmx} + \bbE\qut{\qutn{\bmp - \bmx }^2 \mid \bmx}.
\end{split}
\]
Taking expectation over $\bmx$, the last inequality leads to 
\[
\bbE\qut{\qutn{\bm{q} - \bmx }^2} \leq \epsilon^2 + \bbE\qut{\qutn{\bmp - \bmx }^2},
\]
and the proof is completed. 
\end{proof}

\section{Additional details for Subsection \ref{subsect:microscopydenoising}}
In this section, further details about the experiments on real microscopy images will be given. We first provide the training specification for the denoising networks. The details of the estimation of noise variance , required by the training process, will be discussed next. Finally, we present a noise variance refinement method which helps us to find a more accurate noise variance. The general experimental setting is similar to that of Subsection \ref{sect:denoisingexp}. 

\textit{Training specification.}
We use the network architecture DnCNN \cite{zhang2017beyond} throughout the experiments. During training, the images are copped into overlapping patches, using a sliding window of size $40\!\times\!40$ and a step size of $10$. Additionally, for the N2DH-GOWT1 and C2DL-MSC datasets, since the major regions in the images are of low intensity corresponding to the background, we reorder the patches by their averaged pixel intensity (from low to high), and randomly delete $80\%$ of the first $80\%$ patches. This aims to let the network put more attention on the foreground objects during training. For all three datasets, we apply random flipping to augment the training images. The parameter $\gamma$ is set to $2^2$ for N2DH-GOWT1 and C2DL-MSC, and $2^6$ for TOMO110. The optimization process follows the ones described in Subsection \ref{sect:denoisingexp}, but one needs to use the estimated noise variance which will be introduced below.

\textit{Noise variance estimation.} 
We estimate the noise variance using the methods that are introduced next. 

1). For the datasets N2DH-GOWT1 and C2DL-MSC, the variance is estimated via the mean square difference of neighboring pixels at smooth regions of the images. Firstly, given the noisy image $y$, we define the average squared difference at pixel $i$ as
\[
 d_i^2 =  \sum_{j \in \mathcal{N}(i)}(y_i - y_j)^2 / (2 |\mathcal{N}(i)|) 
\]
where $\mathcal{N}(i)$ denotes the 4-connected pixels of $i$. Secondly, we compute a smooth version of $y$ 
\[
s = k * y
\]
where $*$ denotes convolution operations and $k$ is a uniform kernel matrix of $10\!\times\!10$ (the entries sum to $1$). Thirdly, to find the smooth regions of the images, we define 
\[
\mathcal{F} = \qut{ i \mid 
f_i \leq 0.02  (s_i-\min_j(s_j) ) }
\]
where $f = k * \qut{(s - y)\odot(s - y)}$ and $\odot$ denotes element-wise multiplication. Finally, we assume that the noise variance depends only on the intensity $v$ of pixels, and the estimated variance is computed as 
\[
V(v) = {\rm mean} \qutc{  
d_i^2 \mid i \in \mathcal{F}, |s_i - v| < 0.5/255}. 
\]
In our experiment, $V(v)$ is averaged over all images $y$ in the training set. The computed variance is plotted in Fig. \ref{fig:micropyenv}, where the blue squares indicate $V(v)$. We observe that $V(\cdot)$ has a linear pattern, therefore we approximate $V(v)$ with a linear function $V(v) \approx (v-\mu)/\lambda$ (see the red lines in Fig. \ref{fig:micropyenv}) where $\mu$ and $\lambda$ are parameters. For dataset N2DH-GOWT1, $\mu = 0, \lambda = 69$. For C2DL-MSC $\mu = 0.09, \lambda = 76$, and  $V(v)\!=\!0$ for $v \leq \mu$. 

\begin{figure}[h!]
    \centering
    \setlength{\tabcolsep}{1pt}
    \begin{tabular}{cc}
        \includegraphics[width=0.5\linewidth, trim=0 7 0 10, clip]{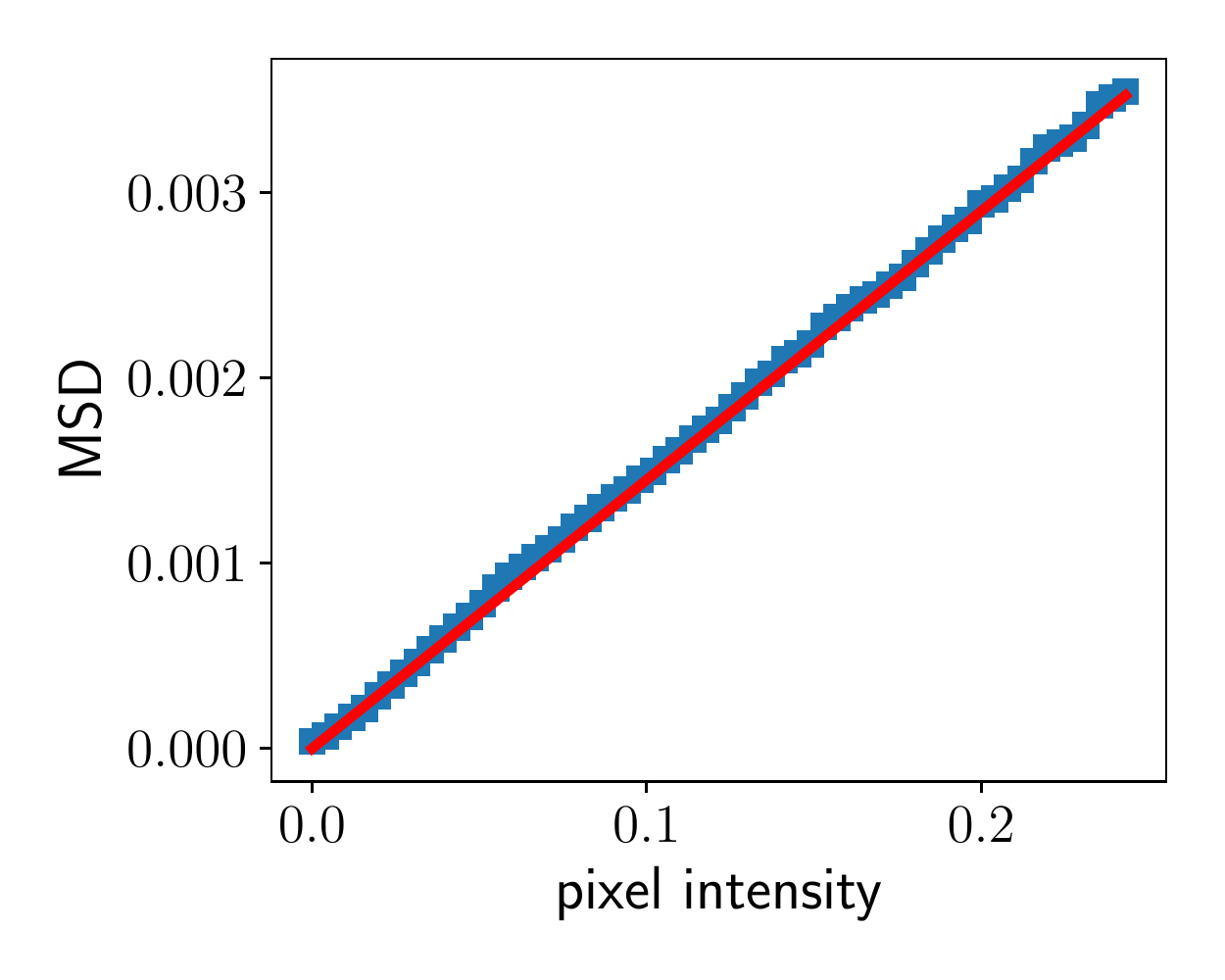}
        &
        \includegraphics[width=0.5\linewidth,trim=0 7 0 10, clip]{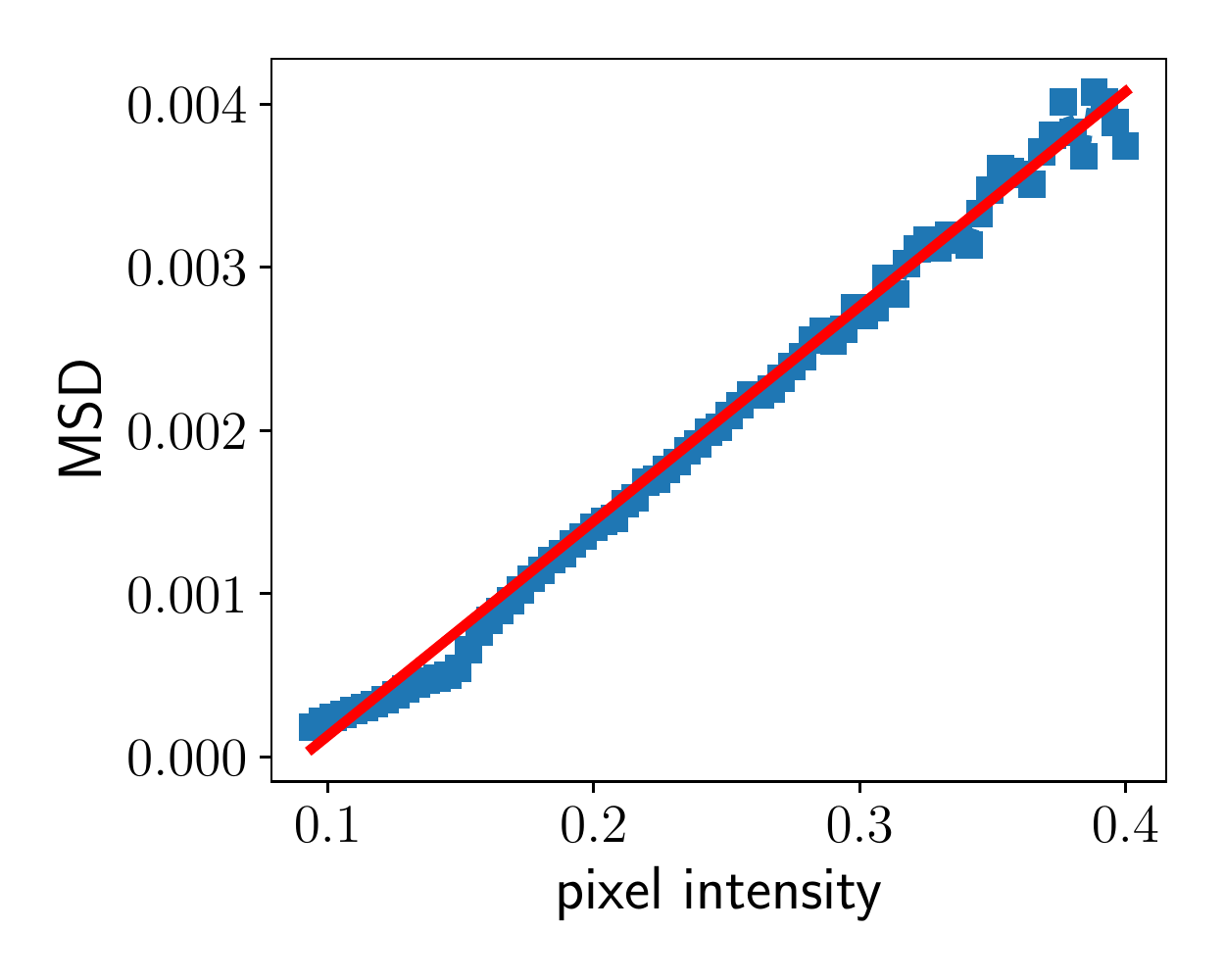}\\
        \footnotesize{ (a).  N2DH-GOWT1} & \footnotesize{ (b). C2DL-MSC} 
    \end{tabular}
    \caption{initial noise variance estimates.}\label{fig:micropyenv}
\end{figure}

2). For TOMO110, since the training image has multiple frames (as a consequence of dose-fractionated acquisition), we estimate the variance by the differences among the frames. Specifically, the raw dose-fractionated data contains $a^{(1)}, a^{(2)}, \cdots, a^{(n)}$ where $n$ is the number of frames, and the frames $a^{(j)}$ have independent noise. In the dataset being considered, for example, $n=10$. The observed image is the sum of all frames, i.e., $y = \sum_{j=1}^n a^{(j)}$. By computing the variance among $a^{(j)}$, the noise variance at pixel $i$ is estimated as
\[
V(i) = \frac{n}{n-1} \sum_{j=1}^n \qut{ a_i^{(j)} - y_i/n }^2.
\]
We note that $V(i)$ is an estimate of the noise variance for the pixel $i$, and our method denoise the whole image $y$ instead of individual frames $a^{(j)}$. In the Noise2Noise approach, in contrast, the network leans to denoise half of the frames, therefore does not make full use of $y$ during the denoising process. 

\textit{Noise variance refinement.} The rough noise variance estimates for datasets N2DH-GOWT1 and C2DL-MSC are refined by the trained models. Specifically, a denoising network is first trained based on the rough noise estimates $V(v)$. The trained network is then fine tuned for several different values of $\lambda$ that are close to the estimates. For fine tuning, we apply $50000$ optimization steps, and the other hyperparameters remain unchanged. This results in several denoising models (corresponding to different $\lambda$). For each model, we compute the term $\bbE\qut{\qutan{\bmz, L\bmz}}$ for a constant image $\bmx$, by using Remark \ref{remark1}. Since the ground truth noise variance is unknown, here we generate $\bmn$ using the estimated variance that is parameterized by $\lambda$. The results are shown in Fig. \ref{fig:micropyrefine}. Motivated by experimental results in Subsection \ref{subsect:env} and Fig. \ref{fig:noisevarianceerror}, we select the $\lambda$ values which lead to a small positive $\bbE\qut{\qutan{\bmz, L\bmz}}$. Our final estimated value of $\lambda$ is $69$ for N2DH-GOWT1 and $88$ for C2DL-MSC. No noise estimate refinement steps are performed for TOMO110 because its noise variance is estimated on individual pixels. 

\begin{figure}[h!]
    \centering
    \setlength{\tabcolsep}{1pt}
    \begin{tabular}{cc}
        \includegraphics[width=0.5\linewidth, trim=0 7 0 10, clip]{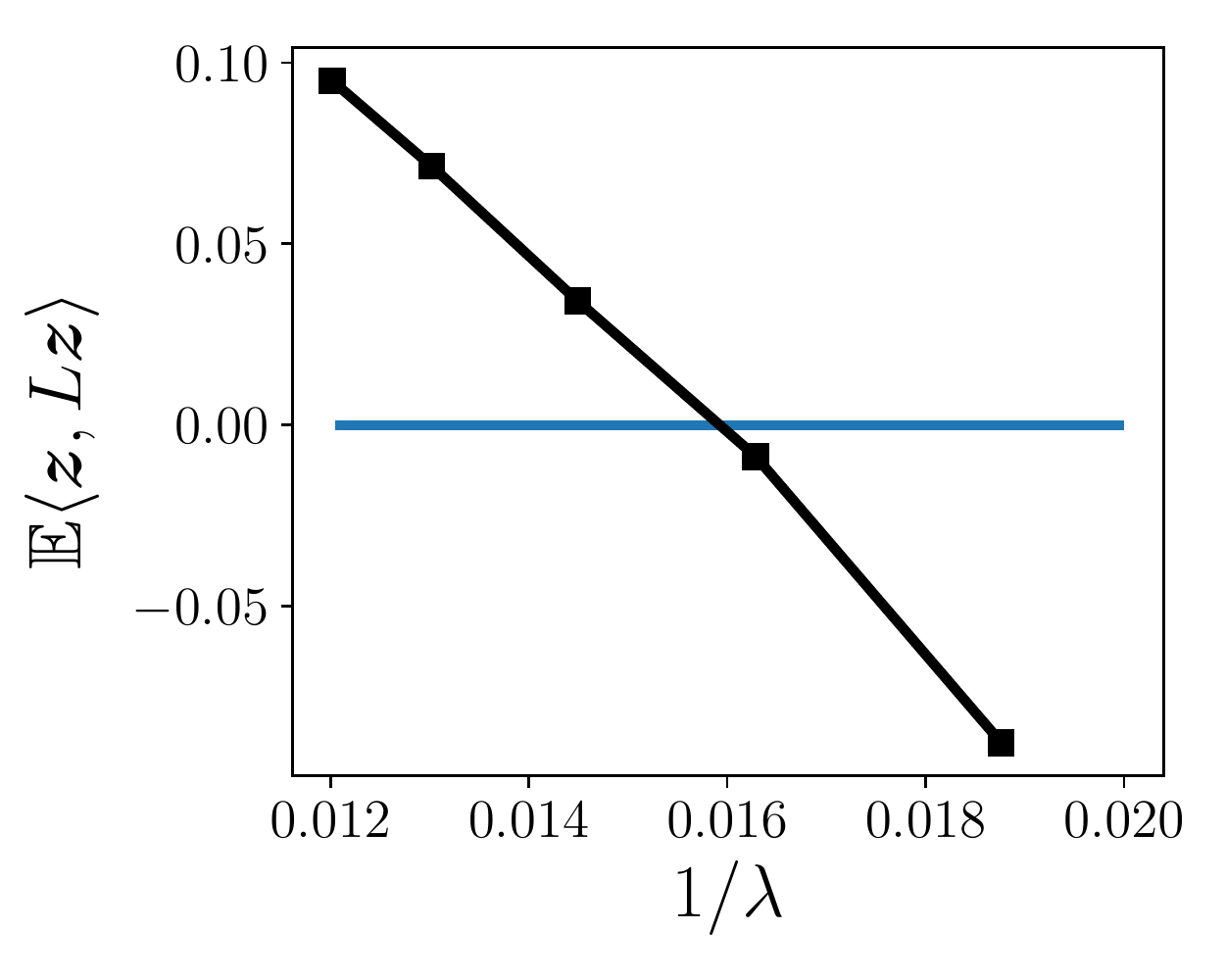}
        &
        \includegraphics[width=0.5\linewidth,trim=0 7 0 10, clip]{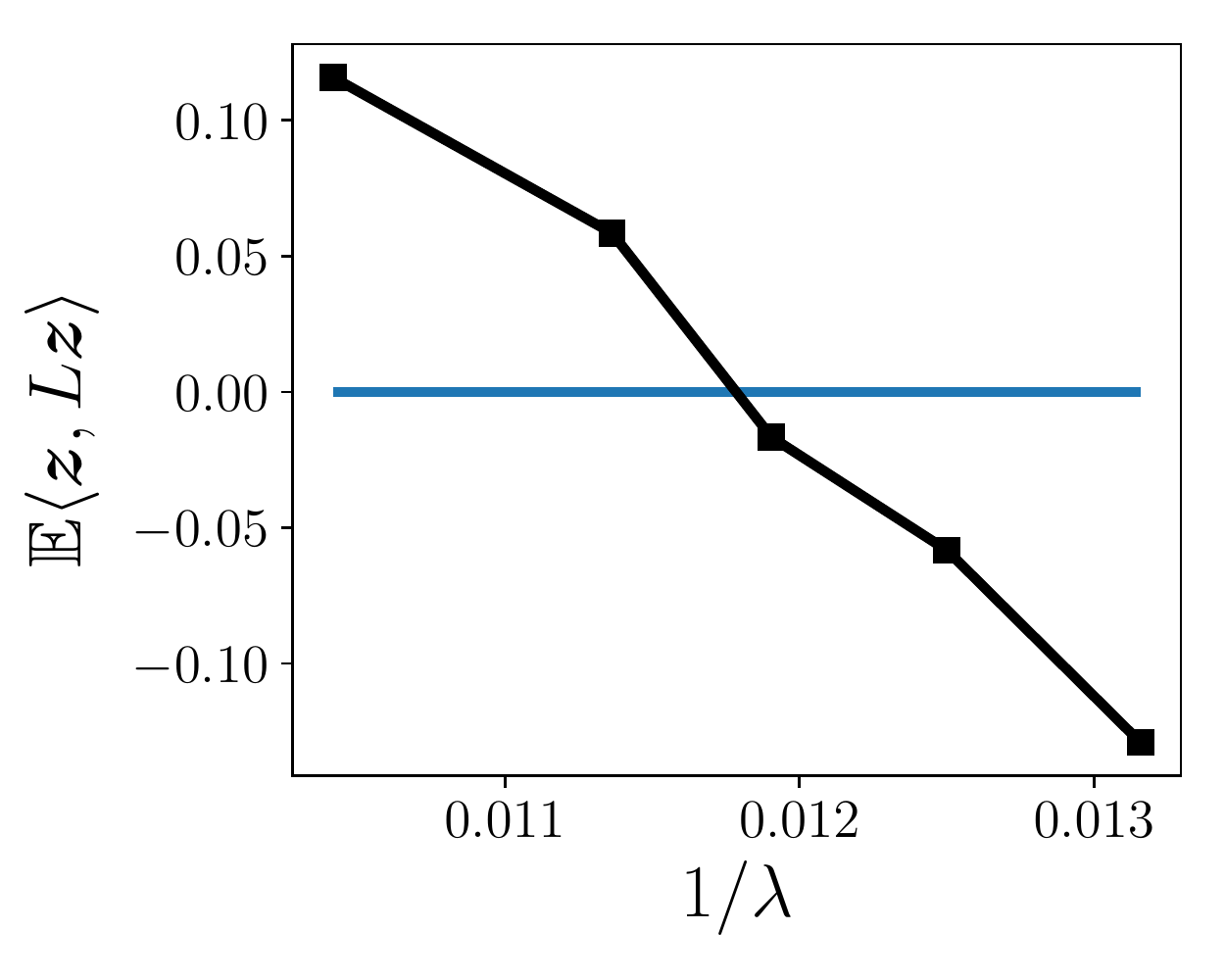}\\
        \footnotesize{ (a).  N2DH-GOWT1} & \footnotesize{ (b). C2DL-MSC} 
    \end{tabular}
    \caption{Refining noise variance estimates using the learned denoisers.}\label{fig:micropyrefine}
\end{figure}

\def\widthfive{0.18\linewidth}
\def\getdemod#1#2#3{
\begin{tikzpicture}[spy using outlines]
  \node{\includegraphics[width=\widthfive]{#1}};
    \spy [red, magnification=4, width=\widthfive, height=.618*\widthfive, line width=2] on (#2*\widthfive,#3*\widthfive) in node at (0,-0.83*\widthfive);
\end{tikzpicture}
}
\begin{figure*}[h!]
    \centering
    \setlength{\tabcolsep}{1pt}
    \begin{tabular}{ccccc}
        \addlinespace[0.2cm]
        \getdemod{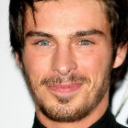}{0.18}{0.2} 
        & 
        \getdemod{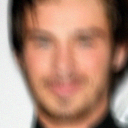}{0.18}{0.2}& 
        \getdemod{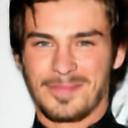}{0.18}{0.2}
        & 
        \getdemod{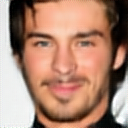}{0.18}{0.2}
        & 
        \getdemod{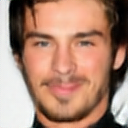}{0.18}{0.2}
        \\ 
        \addlinespace[0.2cm]
        \getdemod{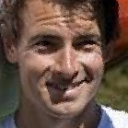}{0.22}{-0.1}
        &
        \getdemod{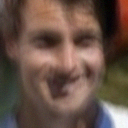}{0.22}{-0.1}
        & 
        \getdemod{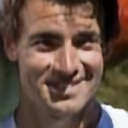}{0.22}{-0.1}
        & 
        \getdemod{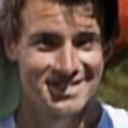}{0.22}{-0.1}
        & 
        \getdemod{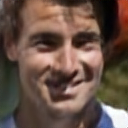}{0.22}{-0.1}
        \\
        \addlinespace[0.2cm]
        \getdemod{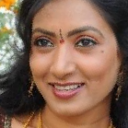}{-0.22}{-0.17}
        & 
        \getdemod{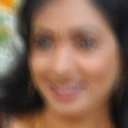}{-0.22}{-0.17}
        &
        \getdemod{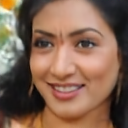}{-0.22}{-0.17}
        &
        \getdemod{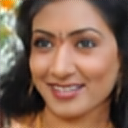}{-0.22}{-0.17}
        &
        \getdemod{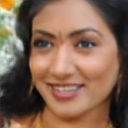}{-0.22}{-0.17}
        \\
        Ground truth & Blurry image & Fully supervised & Xia et al. \cite{xia2019training} & Ours \\
    \end{tabular}
    \caption{Additional visual results about deblurring. The images are taken from the test sets of Helen \cite{le2012interactive} and CelebA \cite{liu2015deep}}
    \label{fig:deblurpart2}
\end{figure*}
\section{Additional results about deblurring}
Fig. \ref{fig:deblurpart2} contains additional visual outputs of our model, in comparison with the fully supervised baseline and Xia et al. \cite{xia2019training}. The experimental setting is described in Subsection \ref{subsect:blind_deblurring}. The results show that, learning without ground truth images, our model recovers major structures of the faces, and performance is comparable to the Noise2Noise based method Xia et al. \cite{xia2019training}.

\ifCLASSOPTIONcaptionsoff
  \newpage
\fi

\end{document}